%% file: intrinsic_select.tex
\documentclass[12pt]{article}
\usepackage[margin=0.9in]{geometry}
\usepackage{amsthm,amsmath,amssymb,amscd,enumerate, graphicx,subcaption,pstricks,bm,booktabs}
\usepackage{setspace,array}
\usepackage{thmtools}
\usepackage{thm-restate}
\usepackage{color}
\usepackage{natbib}
\usepackage[mathscr]{euscript}
\usepackage[shortlabels]{enumitem}
\usepackage[colorlinks=true, urlcolor=blue, linkcolor=blue, citecolor=blue]{hyperref}
\usepackage{algpseudocode}
\usepackage{algorithm}
\algrenewcommand\algorithmicindent{0.4em}
\usepackage{natbib}
\usepackage[T1]{fontenc}
\usepackage[utf8]{inputenc}
\usepackage{authblk}

\declaretheorem[name=Theorem]{thm}
\declaretheorem[name=Lemma]{lemma}

\DeclareMathOperator*{\argmax}{argmax}

\title{Flexible variable selection in the presence of missing data}
\author[1,2]{Brian~D. Williamson}
\author[2,3]{Ying Huang}
\affil[1]{Biostatistics Division, Kaiser Permanente Washington Health Research Institute}
\affil[2]{Vaccine and Infectious Disease Division, Fred Hutchinson Cancer Center}
\affil[3]{Department of Biostatistics, University of Washington}

\begin{document}

\maketitle

\begin{abstract}

  In many applications, it is of interest to identify a parsimonious set of features, or panel, from multiple candidates that achieves a desired level of performance in predicting a response. This task is often complicated in practice by missing data arising from the sampling design or other random mechanisms. Most recent work on variable selection in missing data contexts relies in some part on a finite-dimensional statistical model, e.g., a generalized or penalized linear model. In cases where this model is misspecified, the selected variables may not all be truly scientifically relevant and can result in panels with suboptimal classification performance. To address this limitation, we propose a nonparametric variable selection algorithm combined with multiple imputation to develop flexible panels in the presence of missing-at-random data. We outline strategies based on the proposed algorithm that achieve control of commonly used error rates. Through simulations, we show that our proposal has good operating characteristics and results in panels with higher classification and variable selection performance compared to several existing penalized regression approaches in cases where a generalized linear model is misspecified. Finally, we use the proposed method to develop biomarker panels for separating pancreatic cysts with differing malignancy potential in a setting where complicated missingness in the biomarkers arose due to limited specimen volumes.

  \begin{center}{\small \textbf{Keywords:} variable selection; missing data; machine learning; nonparametric statistics; multiple imputation; variable importance.}\end{center}
\end{abstract}

\doublespacing

\section{Introduction}\label{sec:intro}

Missing data present a common challenge in many scientific domains. This challenge is compounded if a goal of the analysis is to identify a parsimonious set of features that are related to the response, a notion that has been referred to as variable selection. Many existing approaches to variable selection in missing-data contexts rely in some part on a finite-dimensional statistical model, including generalized linear models \citep[see, e.g.,][]{little1985,long2015,liu2019}. While variable selection based on generalized linear models has been shown to perform well in many cases, recovering the true set of important variables and selecting few unimportant variables, model misspecification or correlated features may impact the performance of these methods \citep{bang2005}. This motivates the consideration of approaches to variable selection with missing data that are more robust to model misspeficiation. These approaches should incorporate flexible algorithms, ensuring that complex relationships between the features and the outcome can be captured reliably.

Traditional approaches to variable selection with missing data can be broadly categorized into two groups. In the first, variable selection methods valid with fully-observed data are adapted to the missing-data paradigm using either likelihood-based methods \citep[see, e.g.,][]{little1985} or inverse probability weighting methods \citep[see, e.g.,][]{tsiatis2007,bang2005,johnson2008,wolfson2012}. These approaches, while useful in many contexts, often are tailored to a specific data-generating distribution or missing data process or can only be used with estimating functions for regression parameters. Additionally, inverse probability weighting is challenging in cases with non-monotone missing data \citep[see, e.g.,][]{sun2018}, limiting its more widespread adoption. The second group of approaches is based on multiple imputation \citep[MI;][]{rubin1987}, and is widely used \citep[see, e.g., ][]{long2015,liu2019}. Among the advantages of MI over other approaches are that imputation is easily done with existing software and the imputation process is disentangled from the variable selection procedure. The imputation process must be specified with care, because methods that rely too heavily on modelling assumptions may still be subject to bias in cases with misspecification. Multiple imputation by chained equations \citep[see, e.g.,][]{vanbuuren2018} allows flexible imputation models to be used in an effort to reduce the risk of imputation model misspecification. Once an imputation procedure has been specified, variable selection methods developed for fully-observed data can be used on the imputed datasets.

Methods for variable selection valid with fully-observed data include the lasso \citep{tibshirani1996} and smoothly clipped absolute deviation \citep[][]{fan2001} and extensions thereof \citep[see, e.g.,][]{meinshausen2010}. These approaches all rely on parametric modeling assumptions. The knockoff procedure \citep{barber2015} has seen recent focus, including towards making the procedure more robust to model misspecification \citep[see, e.g.,][]{candes2016}, but often some level of assumptions are necessary for valid error control or inference \citep[see, e.g.,][]{barber2020}. Other methods have been proposed that generate pseudo-variables for variable selection \citep[see, e.g.,][]{wu2007,boos2009}, similar to knockoffs. Stability selection \citep{meinshausen2010,shah2013} has also been shown to provide error control for lasso-based procedures. However, as noted above, in some contexts model misspecification may result in poor performance of these procedures \citep[see, e.g.,][]{leng2006}, motivating the consideration of alternatives that are not based upon generalized linear models. Additionally, when multiple imputation is used to address missing data, the results from these variable selection approaches must be combined after being applied to each imputed dataset separately \citep{peterson2021}. Often, variables that are selected in some proportion of the imputed datasets are designated in the final set \citep[see, e.g.][]{heymans2007, long2015}. This threshold can be difficult to choose, and applying a possibly misspecified regression model to several datasets can amplify the burden of any misspecification.

In this article, we propose an approach to more flexible variable selection in contexts with missing data. Our proposed approach to variable selection is based upon an algorithm-agnostic definition of intrinsic variable importance \citep{williamson2021}. Intrinsic importance quantifies the population-level prediction potential of features. Importantly, though recent theoretical developments have led to a procedure for doing inference on intrinsic importance \citep{williamson2020c,williamson2021}, making inference on this importance in general missing-data cases and using the importance as part of a variable selection procedure have not been studied. To allow for flexible modeling of the missing data process, we propose that missing data be imputed using multiple imputation by chained equations. Our proposed intrinsic approach to variable selection builds on the Shapley population variable importance measure \citep[][]{williamson2020c} and formally incorporates variability in the imputation process into the variable selection procedure using Rubin's rules \citep{rubin1987}, circumventing the need for post-hoc combination of multiple selected variable sets. This approach results in a single set of variables explicitly selected based on estimated population importance. We provide theoretical results guaranteeing control over several commonly-used error rates, including the generalized family-wise error rate and the false discovery rate \citep[see, e.g.,][]{lehmann2012}.

The remainder of this paper is organized as follows. In Section~\ref{sec:select}, we discuss the connection between intrinsic variable importance and selection and a procedure for selecting an initial set of variables in cases with fully-observed or missing data. In Section~\ref{sec:error-control}, we describe an approach to augmenting this initial set and provide theoretical results guaranteeing control over variable selection error rates. We provide numerical experiments illustrating the use of our proposed approach and detailing its operating characteristics in Section~\ref{sec:sims}. Finally, we select possible important biomarkers for pancreatic cancer early detection in Section~\ref{sec:data}, and provide concluding remarks in Section~\ref{sec:discussion}. All technical details and results from additional simulation studies can be found in the Supplementary Material.

\section{Intrinsic variable selection}\label{sec:select}

\subsection{Data structure and notation}\label{sec:notation}

Suppose that observations $Z_1,\ldots,Z_n$ are drawn independently from data-generating distribution $P_0$ known only to belong to a rich class of distributions $\mathcal{M}$. Suppose further that $Z_i := (Y_i, X_i)$, where $X_i := (X_{i1},\ldots,X_{ip}) \in \mathcal{X} \subseteq \mathbb{R}^p$ is a vector of covariates and $Y_i \in \mathbb{R}$ is the outcome of interest. We refer to the vector $Z := (Y, X)$ as the complete-data unit or the ideal-data unit in cases with no missing data and with missing data, respectively. Let $\Delta_i := (\Delta_{i0}, \ldots, \Delta_{ip}) \in \{0,1\}^{p+1}$ denote a pattern of missing data for the outcome and covariates, where $\Delta_0 = 1$ implies that the outcome is observed and $\Delta_j = 1$ implies that covariate $X_j$ is observed for $j = 1, \ldots, p$. We denote the observed data by $O_1,\ldots,O_n$, where $O_i := (\Delta_i, \Delta_{i0}Y_i, \Delta_{i1}X_{i1}, \ldots, \Delta_{ip}X_{ip})$, and denote the observed data unit by $O := (\Delta, \Delta_0Y, \Delta_1X_1, \ldots, \Delta_p X_p)$. We denote the observed-data distribution, which includes the missing-data mechanism, by $Q_0$.

For each index set $s \subseteq \{1, \dots, p\}$, we consider the class of functions $\mathcal{F}_s := \{f \in \mathcal{F}\,:\, f(u) = f(v) \text{ for all } u, v \in \mathcal{X} \text{ satisfying } u_s = v_s\}$, where $\mathcal{F}$ is a large class of functions. We also consider a scientifically meaningful predictiveness measure $V(f, P)$, where larger values of $V$ are assumed to be better; examples of $V$ include $R^2$ and classification accuracy \citep[see, e.g.,][]{williamson2021}. For each $s \subseteq \{1, \dots, p\}$, we define the predictiveness-maximizing function $f_{0,s} \in \argmax_{f \in \mathcal{F}_s} V(f, P_0)$.

\subsection{Estimating intrinsic variable importance in missing-data settings}\label{sec:intrinsic-select}

To circumvent the need to rely on potentially restrictive parametric modelling assumptions, we can consider an approach to variable selection that is based on intrinsic variable importance. We propose to perform intrinsic variable selection using the Shapley population variable importance measure \citep[SPVIM;][]{williamson2020c}, which we denote by $\psi_0 := \{\psi_{0,j}\}_{j=1}^p$. The ideal-data SPVIM for feature $X_j$ is
\begin{align*}
  \psi_{0,j} := \sum_{s \in \{1,\ldots,p\}\setminus \{j\}} \binom{p - 1}{\lvert s \rvert}^{-1}\frac{1}{p}\{V(f_{0,s\cup j}, P_0) - V(f_{0,s}, P_0)\},
\end{align*}
and quantifies the increase in population prediction potential, as measured by $V$, of including $X_j$ in each possible subset of the remaining features $\{1, \ldots, p\} \setminus \{j\}$. This definition provides a useful dichotomy: if $\psi_{0,j} > 0$, feature $X_j$ has some utility in predicting the outcome in combination with at least one subset of the remaining features; if $\psi_{0,j} = 0$, then feature $X_j$ does not improve population prediction potential if added to any subset of the remaining features. This key fact suggests that estimators of the SPVIM may be used to screen out variables with no intrinsic utility. More formally, for each $j \in \{1, \ldots, p\}$, we define the null hypothesis $H_{0,j}: \psi_{0,j} = 0$. We can then define the following sets of variables:
\begin{align}
  S_0 \equiv S_0(P_0) :=& \ \{j \in \{1, \ldots, p\}: \psi_{0,j} > 0\} \text{ and } \label{eq:active-set}\\
  S_0^c \equiv S_0^c(P_0) :=& \ \{j \in \{1, \ldots, p\}: \psi_{0,j} = 0\}. \label{eq:null-set}
\end{align}
We will refer to $S_0$ as the active set and $S_0^c$ as the null set. The goal of a variable selection procedure can be recast into identifying $S_0$ while ignoring $S_0^c$; these sets and the true ideal-data SPVIM values are all defined relative to the underlying population $P_0$.

Prior to considering missing-data settings, we provide a brief overview of the ideal-data estimation procedure detailed more fully in \citet{williamson2020c}. Since obtaining an estimator $f_{n,s}$ of $f_{0,s}$ for each $s \subseteq \{1, \ldots, p\}$ is generally computationally prohibitive, this estimation procedure is based on sampling a fraction $c$ of all possible subsets of $\{1, \ldots, p\}$. The authors describe the efficient influence function \citep[EIF; see, e.g.,][]{pfanzagl1982} of the ideal-data SPVIM and propose an estimator $\psi_{c,n} := \{\psi_{c,n,j}\}_{j=1}^p$ for each SPVIM based on $K$-fold cross-fitting that is asymptotically efficient in complete-data settings. Under regularity conditions, $n^{1/2}(\psi_{c,n} - \psi_0) \sim N_p(0, \Sigma_0)$, where $\Sigma_0 = E_0 [\phi_0(O)\phi_0(O)^\top]$ and $\phi_0(o)$ is the vector of EIF values for each $j$. We provide the exact conditions (A1)--(A7) in the Supplementary Materials (Section 6.1), but briefly describe them here. The conditions ensure that: estimation of $f_{0,s}$ only contributes to the higher-order behavior of $\psi_{c,n}$, and this contribution is asymptotically negligible; $\psi_{c,n}$ is a consistent estimator of $\psi_0$; and $\Sigma_0$ is based on the EIF. These conditions hold for many common choices of the predictiveness measure $V$ and estimators of $f_{0,s}$ \citep{williamson2021}. While individual algorithms (e.g., generalized linear models or random forests) could be used to obtain estimators of $f_{0,s}$ necessary for estimating $\psi_{0,j}$, we advocate instead for using a Super Learner ensemble of candidate estimators \citep{vanderlaan2007}. The Super Learner ensemble is an implementation of regression stacking that is guaranteed to have risk equal to the risk of the oracle estimator, asymptotically \citep{vanderlaan2007}. The risk of model misspecification can be reduced by specifying a large library of candidate algorithms.

In many cases, including our analysis in Section~\ref{sec:data}, we do not observe the ideal data unit $Z$ but instead observe $O$, where data on covariates, the outcome, or some subset of these are missing. In these cases, a strategy for properly handling these missing data is necessary to perform variable selection and establish control of error rates. Our goal remains to do variable selection based on the ideal-data intrinsic importance described above. 

One strategy involves defining an observed-data intrinsic variable importance measure based on $O$ that identifies the ideal-data intrinsic importance under assumptions on the missing-data process, such as the positivity assumption \citep{bang2005}. However, this strategy is inherently tied to the measure $V$ under consideration, and the assumptions must be carefully specified. For each combination of $V$ and missing-data process, a different EIF must be analytically derived. Additionally, in many cases with non-monotone patterns of missing data, the positivity assumption may not hold.

The strategy that we employ in this manuscript involves multiple imputation due to its potential for flexibility in modeling both monotone and non-monotone missing data patterns. Once an imputation model is determined, $M$ imputed datasets $\tilde{Z}_1$, \ldots, $\tilde{Z}_M$ are created. This imputation model must be sufficiently flexible to reduce the risk of model misspecification. 

We will use MI to do inference on the ideal-data intrinsic importance using Rubin's rules \citep{rubin1987}. Suppose that for each of the $M$ imputed datasets, we have computed SPVIM estimator $\psi_{m,c,n}$ of $\psi_0$ and its corresponding variance estimator $\sigma^2_{m,n}$. Define $\psi_{M,c,n} := M^{-1}\sum_{m=1}^M \psi_{m,c,n}$, $\sigma^2_{M,n} := M^{-1}\sum_{m=1}^M \sigma^2_{m,n}$, $\tau^2_{M,n} := (M - 1)^{-1}\sum_{m=1}^M (\psi_{m,c,n} - \psi_{M,c,n})^2$, and $\omega^2_{M,n} := \sigma^2_{M,n} + (M+1)M^{-1}\tau^2_{M,n}$. Before stating a formal result, we first introduce a regularity condition for the use of Rubin's rules. Below, all expectations are with respect to the full data.
\begin{itemize}
    \item[(A8)] \textit{(consistency of imputations)}
        \begin{itemize}
            \item[(A8a)] $\lim_{M \to \infty} E(\psi_{M,c,n} \mid Z) = \psi_{c,n}$;
            \item[(A8b)] $\lim_{M \to \infty} E(\sigma^2_{M,n} \mid Z) = \sigma^2_n$;
            \item[(A8c)] $\lim_{M \to \infty} E(\tau^2_{M,n} \mid Z) = \lim_{M \to \infty} var(\psi_{M,c,n} \mid Z)$.
        \end{itemize}
\end{itemize}
These conditions are commonly referred to as the essential conditions for proper MI \citep[see, e.g.,][]{rubin1996}, and in turn provide conditions for the approximate asymptotic normality of appropriately centered and scaled version of $\psi_{M,c,n}$. The missing data must be missing completely at random or missing at random \citep{rubin1987}. The following result describes the asymptotic distribution of the imputation-based estimator $\psi_{M,c,n}$.
\begin{lemma}\label{lem:mi-normality}
    Provided that conditions (A1)--(A8) hold and the data are missing at random, then $n^{1/2}(\psi_{M,c,n} - \psi_0)$ is approximately asymptotically normally distributed with consistent variance estimator $\omega^2_{M,n}$.
\end{lemma}
We can thus use the imputation-based estimator $\psi_{M,c,n}$ and its variance estimator $\omega^2_{M,n}$ to make inference on $\psi_0$. We adopt a two-stage strategy towards variable selection: first, select an initial set of variables using a procedure with possibly strict multiple-testing control; and second, augment this set of variables while maintaining control of generalized error rates. We describe these stages in the following sections.

\subsection{Selecting an initial set of variables}\label{sec:initial-set}

Suppose that conditions (A1)--(A8) hold. Let $\omega_{M,n,j}^2$ denote the $j$th component of the diagonal of the estimated covariance matrix based on the estimated EIF and imputation variance; if the data are fully observed, then there is only a single dataset and no imputation component to the variance. Based on the estimated variance and importance, we can define test statistics $T_{M,n,j} := \omega_{M,n,j}^{-1}(\psi_{M,c,n,j} - \psi_{0,j})$. The test statistics $T_{M,n} := (T_{M,n,1},\ldots,T_{M,n,p})$ follow a multivariate normal distribution under the joint null hypothesis, which we denote $\mathcal{P}_0$.

Armed with these test statistics, we select an initial set of variables. For a given $\alpha \in (0, 1)$ and possibly random cutoff functions $c_j(t, \mathcal{P}_0, \alpha)$, we define adjusted p-values $\tilde{p}_{M,n,j} := \inf \{\alpha \in [0,1]: T_{M,n,j} > c_j(T_{M,n}, \mathcal{P}_0, \alpha)\}$ \citep[see, e.g.,][]{dudoit2008}, resulting in
\begin{align}
    S_{M,n}(\alpha) = \{j \in \{1, \ldots, p\}: \tilde{p}_{M,n,j} \leq \alpha\}. \label{eq:intrinsic-initial-set}
\end{align}
The procedure for determining the adjusted p-values will determine how and whether any multiple-testing control is achieved in determining $S_{M,n}(\alpha)$. Below, we will provide an example of the adjusted p-values using a Holm procedure \citep{holm1979}. We define $R_{M,n}(\alpha) := \lvert S_{M,n}(\alpha) \rvert$ to be the number of rejected null hypotheses after this initial variable selection step. In settings with complete data, where multiple imputation is not necessary, we refer to these objects as $S_n(\alpha)$ and $R_n(\alpha)$, respectively.

An ideal selection procedure will result in $S_{M,n}(\alpha) \to_P S_0$ and $R_{M,n}(\alpha) \to_P \lvert S_0 \rvert$ as $n \to \infty$ and $M \to \infty$ while maintaining control of the number of falsely selected variables. In other words, we want to minimize the number of type I errors $Q_{M,n}(\alpha) := \lvert S_{M,n}(\alpha) \cap S_0^c \rvert$ while maximizing the number of selected truly important variables $\lvert S_{M,n}(\alpha) \cap S_0 \rvert$. However, many procedures, including the Holm procedure, provide control over the familywise error rate, which may be too strict in some settings \citep{lehmann2012}. In the next section, we describe a procedure for augmenting the set $S_{M,n}(\alpha)$, obtained using the estimated intrinsic importance values, to provide control over possibly less strict error rates.

\subsection{Augmenting the initial set to ensure error rate control and persistence}\label{sec:error-control}

Before detailing our full procedure and providing our main results, we introduce some additional notation. First, we define three commonly used error rates. For a given integer $k \geq 0$, the generalized family-wise error rate, of at least $k + 1$ type I errors, is defined as $gFWER(k) := Pr_{P_0}(Q_{M,n}(\alpha) \geq k + 1) = 1 - F_{Q_{M,n}(\alpha)}(k + 1),$ where $F_{Q_{M,n}}$ is the cdf of $Q_{M,n}$ and $gFWER(0)$ is the family-wise error rate. The proportion of false positives among the rejected variables at level $q \in (0,1)$ is defined as $PFP(q) := Pr_{P_0}(Q_{M,n}(\alpha) / R_{M,n}(\alpha) > q)$. Finally, we define the false discovery rate to be $FDR := E_{P_0}(Q_{M,n}(\alpha) / R_{M,n}(\alpha))$.

Next, we define the collection of sets of functions $\mathcal{C}_n := \bigcup_{s \subseteq \{1, \ldots, p\}: \lvert s \rvert = k_n} \mathcal{F}_s$ for $k_n \leq p$ and let $f_* \in \argmax_{f \in \mathcal{C}_n} V(f, P_0)$ denote the predictiveness-maximizing function over all function classes that make use of $k_n$ variables. We say that a variable selection procedure $S_n$ that selects $k_n$ variables is persistent \citep[see, e.g.,][]{greenshtein2004} if $V(f_{n, S_n}, P_0) - V(f_{*}, P_0) \to_P 0,$
where $f_{n, S_n}$ is an estimator of $f_{0, S_n}$, the predictiveness-maximizing function that uses the variables selected by $S_n$. In other words, a persistant procedure ensures that the true predictiveness of the empirical prediction function using the selected variables converges to the true predictiveness of the best possible prediction function making use of the same number of variables. Our definition of persistance can be seen as a nonparametric generalization of \citet{greenshtein2004}.

Based on a chosen multiple-testing control procedure, under conditions (A1)--(A8) we obtain $S_{M,n}(\alpha)$ as described in Equation~\eqref{eq:intrinsic-initial-set}. To provide control over the error rates defined above, we propose to augment $S_{M,n}(\alpha)$. For an integer $k \in \{0, \ldots, p - R_{M,n}(\alpha)\}$, we define the augmentation set
\begin{align}\label{eq:intrinsic-augment}
  A_{M,n}: (k, \alpha) \in \{0, \ldots, p - R_{M,n}(\alpha)\} \times (0,1) \mapsto \begin{cases} \emptyset & k = 0 \\
  \{s \subseteq S_{M,n}^c(\alpha): \tilde{p}_{M,n,\ell} \leq \tilde{p}_{M,n,(k)} \text{ for all } \ell \in s\} & k > 0,
  \end{cases}
\end{align}
where $a_{(j)}$ denotes the $j$th order statistic of a vector $a$. Two examples of augmentation sets are $A_{M,n}(\alpha, 0) = S_{M,n}(\alpha)$ and $A_{M,n}(\alpha, p - R_{M,n}(\alpha)) = S_{M,n}^c(\alpha)$ (i.e., the unselected variables). This results in an augmented set of selected variables $S_{M,n}^+(k, \alpha) = S_{M,n}(\alpha) \cup A_{M,n}(k, \alpha)$, augmented number of selected variables $R_{M,n}^+(k, \alpha) = \lvert S_{M,n}^+(k, \alpha) \rvert$, and augmented number of type I errors $Q_{M,n}^+(k, \alpha) = \lvert S_{M,n}^+(k, \alpha) \cap S_0^c \rvert$. Finally, we define the following set of conditions:
\begin{itemize}
  \item[(B1)] \textit{(finite-sample familywise error rate control)} $Pr_{P_0}(Q_{M,n}(\alpha) > 0) = \alpha_n$ for all $n$;
  \item[(B2)] \textit{(asymptotic familywise error rate control)} $\limsup_{n\to\infty} Pr_{P_0}(Q_{M,n}(\alpha) > 0) = \alpha^* \leq \alpha$;
  \item[(B3)] \textit{(perfect asymptotic power)} $\lim_{n\to\infty}Pr_{P_0}(S_0 \subseteq S_{M,n}(\alpha)) = 1$;
  \item[(B4)] \textit{(limited number of initial rejections)} $\lim_{n\to\infty}Pr_{P_0}(S_{M,n}(\alpha) \leq p - k) = 1$.
\end{itemize}
\begin{thm}\label{thm:error-control}
  If conditions (A1)--(A8) and (B1)--(B2) hold, then for any $k \geq 0$ and $q \in (0,1)$, $S_{M,n}^{+}(k, \alpha)$ provides finite-sample control of $gFWER(k)$ and $PFP(q)$ at level $\alpha_n$:
  \begin{align*}
    Pr_{P_0}(Q_{M,n}^+(k, \alpha) > k) = \alpha_n, \ & \ Pr_{P_0}(Q_{M,n}^+(k, \alpha)/R_{M,n}^+(k, \alpha) > q) = \alpha_n
  \end{align*}
  for all $n$. If additionally (B3)--(B4) hold, then $S_{M,n}^+(k, \alpha)$ provides asymptotic control of these quantities and the FDR, that is,
  \begin{align*}
     \limsup_{n\to\infty} Pr_{P_0}(Q_{M,n}^+(k, \alpha) > k) \leq \alpha, \ & \ \limsup_{n\to\infty} Pr_{P_0}(Q_{M,n}^+(k, \alpha)/R_{M,n}^+(k, \alpha) > q) \leq \alpha, \\
     \limsup_{n\to\infty} E_{P_0}(Q_{M,n}^+(k, \alpha)/R_{M,n}^+(k, \alpha)) \ &\leq \ q(1 - \alpha) + \alpha. 
  \end{align*}
  In complete-data settings, these results hold without reliance on condition (A8).
\end{thm}
This result implies that the user can specify a tolerable threshold for the tail probability of a number of false discoveries, which can result in an augmented set of variables $S_{M,n}^+(\alpha)$ with increased power over the potentially strict initial procedure $S_{M,n}(\alpha)$ while still providing error control. This holds in finite samples and asymptotically, so long as the initial procedure $S_{M,n}(\alpha)$ has high asymptotic power. 

Conditions (B1)--(B4) describe the initial variable selection procedure $S_{M,n}(\alpha)$. While a number of procedures satisfy these conditions under (A1)--(A8), we consider here a Holm-based procedure for simplicity. Based on the p-values $\{p_{M,n,j}\}_{j=1}^p$ from the individual, unadjusted null hypothesis tests, we can construct Holm-adjusted p-values
\begin{align}
   \tilde{p}_{M,n,(j)} := \max_{\ell \in \{1, \ldots, j\}}\{\min\{p_{M,n,(\ell)}(p - \ell + 1), 1\}\} \label{eq:holm-adj-p}.
\end{align}
For $\alpha \in (0, 1)$, we set $S_{M,n}(\alpha) = \{j \in \{1, \ldots, p\}: \tilde{p}_{M,n,j} < \alpha\}$, which guarantees control of the familywise error rate. Next, to control the gFWER, select $k \in \{0, \ldots, p - R_{M,n}(\alpha)\}$; to control the PFP among the selected variables, select $q \in (0,1)$ and set $k = \max\{j \in \{0, \ldots, p - R_{M,n}(\alpha)\}: j\{j + R_{M,n}(\alpha)\}^{-1} \leq q\}$. Define $A_{M,n}(k, \alpha)$ as in Equation~\eqref{eq:intrinsic-augment}, and augment the initial set to obtain $S_{M,n}^+(k, \alpha)$. Other procedures may satisfy (B1)--(B4) and could result in increased power \citep[see, e.g.,][]{dudoit2008}. The general procedure based on any familywise error rate-controlling initial selection step is summarized in Algorithm~\ref{alg:intrinsic-select}. 

The choice of $k$ depends on context. For example, it can be chosen so that no more than a pre-specified number of variables are selected; this may be important in applications where only a limited number of variables can be measured in future experiments. This type of constraint occurs in some cancer early detection studies, which we describe further in Section~\ref{sec:data}. One could instead tune $k$ using cross-validation, which is likely to result in more robust variable selection performance in cases where such a pre-specified threshold is unavailable.

\begin{algorithm}
\caption{Intrinsic variable selection with error rate control}
\label{alg:intrinsic-select}
\begin{algorithmic}[1]
  \State Obtain estimator $\psi_{M,c,n}$ of $\psi_0$ using multiple imputation in settings with missing data, and its corresponding variance estimator $\omega_{M,n}^2$;
  \State For a given $\alpha \in (0,1)$, compute unadjusted p-values $p_{M,n,j}$ for each hypothesis test $H_{0,j}$;
  \State Compute adjusted p-values $\tilde{p}_{M,n,j}$ according to the desired familywise error rate-controlling procedure, e.g., Holm adjusted p-values as in Equation~\eqref{eq:holm-adj-p};
  \State Set $S_{M,n}(\alpha) = \{j \in \{1, \ldots, p\}: \tilde{p}_{M,n,j} < \alpha\}$ as in Equation~\eqref{eq:intrinsic-initial-set};
  \State For a given $k \in \{0, \ldots, p - R_{M,n}(\alpha)\}$, obtain $A_{M,n}(k, \alpha)$ as in Equation~\eqref{eq:intrinsic-augment};
  \State Set $S_{M,n}^+(k, \alpha) = S_{M,n}(\alpha) \cup A_{M,n}(k, \alpha)$.
\end{algorithmic}
\end{algorithm}

The next result describes that under a subset of the conditions of the previous theorem and in complete-data settings, the algorithm described in Algorithm~\ref{alg:intrinsic-select} is persistent.
\begin{lemma}\label{lem:persistence}
  If conditions (A1), (A2), (A5) and (A6) hold for all $s \subseteq \{1, \ldots, p\}$ and conditions (A7) and (B3) hold, then the procedure described in Algorithm~\ref{alg:intrinsic-select} is persistent:
  \begin{align*}
    V(f_{n, S_{n}^+(k, \alpha)}, P_0) - V(f_{*}, P_0) \to_P 0.
  \end{align*}
\end{lemma}
This result implies that $S_{n}^+(k, \alpha)$, the result of Algorithm~\ref{alg:intrinsic-select} in complete-data settings, returns a set of features that has predictiveness converging to the best possible predictiveness among all procedures that select $R_{n}^+(k, \alpha)$ variables. In missing-data settings, if condition (A8) is satisfied, then this result holds when averaged across the imputed datasets and as $M \to \infty$.

\section{Numerical experiments}\label{sec:sims}

\subsection{Experimental setup}

We provide several experiments that are designed to describe the operating characteristics of our proposed intrinsic importance-based variable selection procedure, and compare these procedures with other well-established algorithms. In all cases, our simulated dataset consisted of independent replicates of $(X, Y)$, where $X = (X_1, \ldots, X_p)$ and $Y$ followed a Bernoulli distribution with success probability $\Phi\{\beta_0 + f(\beta, x)\}$ conditional on $X = x$, where $\Phi$ denotes the cumulative distribution function of the standard normal distribution. Under this specification, $Y$ followed a probit model. 

In Scenario 1, we vary $p \in \{30, 500\}$, set $f(\beta, x) = x\beta$, and specify $\beta_0 = 0.5$ and $\beta = (-1, 1, -0.5, 0.5, 1/3, -1/3, \mathbf{0}_{p - 6})^\top$, where $\mathbf{0}_k$ denotes a zero-vector of dimension $k$. We consider $X \sim MVN(0, I_p)$, where $I_p$ is the $p \times p$ identity matrix. In this scenario, procedures that are based on a generalized linear model are correctly specified.

In Scenario 2, we set $p = 6$, add correlation between variables, and specify
\begin{align*}
  f(\beta, x) =& \ 2 [\beta_1 x_2x_3 - \beta_2 \tanh{(x_6)}],
\end{align*}
where $\tanh{}$ denotes the hyperbolic tangent. In this scenario, $\beta_0 = 0.5$, $\beta = (1, 1)^\top$, and $X \sim MVN(0, \Sigma)$, where $\Sigma_{i,j} = 1$ if $i = j$; $\Sigma_{i,j} = \rho_1^{\lvert i - j \rvert}$ for $i, j$ not in the active set; and $\Sigma_{i,j} = \rho_2$ for $i, j$ in the active set. We set $\rho_1 = 0.3$ and $\rho_2 = 0.95$. In this scenario, procedures that are based on a generalized linear model are misspecified.

We first generate complete $(X, Y)$ and then generate missing data using amputation \citep{vanbuuren2018}. The outcome and certain features always have complete data, i.e., $\Delta_{i,j} = 1$ for $j \in \{0, 1, 3, 5\}$ and all $i$. The missing data are missing at random. We specify a monotone missing pattern for $(X_2, X_4, X_6)$, where observing $X_2$ implies that both $X_4$ and $X_6$ are observed. When $p = 500$, 40 noise features have missing data; when $p = 30$, 3 noise features have missing data; the remaining noise features are fully observed. In all scenarios, we consider fully observed data and a maximum of 20\% or 40\% missing data within each column.

For each sample size $n \in \{200, 500, 1500, 3000\}$, we generated 1000 replicates from each combination of data-generating mechanism, number of features, and proportion of missing data. We additionally generated an independent test dataset following the same distribution but with no missing data and with sample size 10,000. We used MI with $M=10$ and predictive mean matching to impute any missing feature information. 

In cases with missing data, we considered three procedures for performing variable selection: the stability-selection based algorithms considered in \citet{long2015} with 100 bootstrap replicates, which we refer to as lasso + SS (LJ) and lasso + SS (BI-BL), denoting stability selection within bootstrap imputation and bootstrap imputation with bolasso \citep{bach2008}, respectively; and intrinsic selection (i.e., Algorithm~\ref{alg:intrinsic-select}) with gFWER, PFP, and FDR control, using AUC to define intrinsic importance, which we refer to as SPVIM + gFWER, PFP, and FDR, respectively. In the latter case, we used a Super Learner to estimate intrinsic importance. In cases with complete data, we used the lasso, lasso with stability selection, lasso with knockoffs, and intrinsic selection (i.e., Algorithm~\ref{alg:intrinsic-select}) to perform variable selection. We attempted to use error-rate control tuning parameters that would provide similar theoretical control over the various error rates across algorithms. The values of the specific algorithms used in the Super Learner, the tuning parameters used in each procedure for error rate control, and the specific R implementations of each algorithm are provided in the Supplementary Material (Sections 7.2 and 7.3). 

After performing variable selection, we estimated the prediction performance of the selected variables by fitting a regression of the outcome on these variables. In cases with missing data, we fit this regression on each of the imputed datasets. To maintain coherence with the assumptions of a given procedure, we used a probit regression in the case of variables selected by the lasso-based methods and used the Super Learner in all other cases. This results in strategies that are based on parametric assumptions for both variable selection and prediction performance estimation, and a strategy that is free of these assumptions in both stages. We then computed the test-set AUC based on the independent sample; in missing-data settings, we averaged the performance on this test set across the prediction functions trained on each imputed dataset. We additionally computed the sensitivity and specificity of the selected set of variables. In the context of variable selection, sensitivity is the proportion of truly important variables that were selected, while specificity is the proportion of truly unimportant variables that were not selected. Finally, we evaluated the average test-set AUC based of the selected variables and the average sensitivity and specificity of each procedure over the 1000 samples.

\subsection{Primary empirical results}

We only show results for the case with 40\% missing data; the results for 20\% missing data and no missing data are similar and are presented in the Supplementary Material (Section 7.4 and 7.6). In Figure~\ref{fig:scenario-1-select}, we display the results of the experiment conducted under Scenario 1; the features are multivariate normal and the outcome-feature relationship follows a linear model. In this scenario, the lasso-based estimators are correctly specified. We observe in Figure~\ref{fig:scenario-1-select} panel A that for both feature-space dimensions $p \in \{30, 500\}$, all estimators have estimated test-set AUC increasing with sample size. In this experiment, the \citet{long2015} lasso and intrinsic variable selection with gFWER control tend to have the highest test-set AUC. The PFP and FDR-controlling intrinsic selection procedures tend to have lower AUC, particularly at smaller sample sizes, reflecting the fact that these procedures provide stricter control of specificity at the cost of sensitivity in these scenarios. A different choice of tuning parameters might lead to a more favorable tradeoff between these two error rates. Additionally, if no variables are selected using the initial procedure (here, using Holm-adjusted p-values less than 0.05), then the PFP and FDR augmentation set is defined as the empty set, suggesting that relaxing FWER control for the initial set of variables could increase sensitivity in this setting.
In Figure~\ref{fig:scenario-1-select} panel B, we observe that empirical sensitivity increases with $n$ towards one for all algorithms regardless of the feature-space dimension, though the PFP- and FDR-controlling intrinsic selection approaches have low sensitivity in the $p = 500$ case. Worryingly, the specificity of the BI-BL lasso is near zero for all cases (Figure~\ref{fig:scenario-1-select} panel C). 

\begin{figure}
  \centering
  \includegraphics[width=1\textwidth]{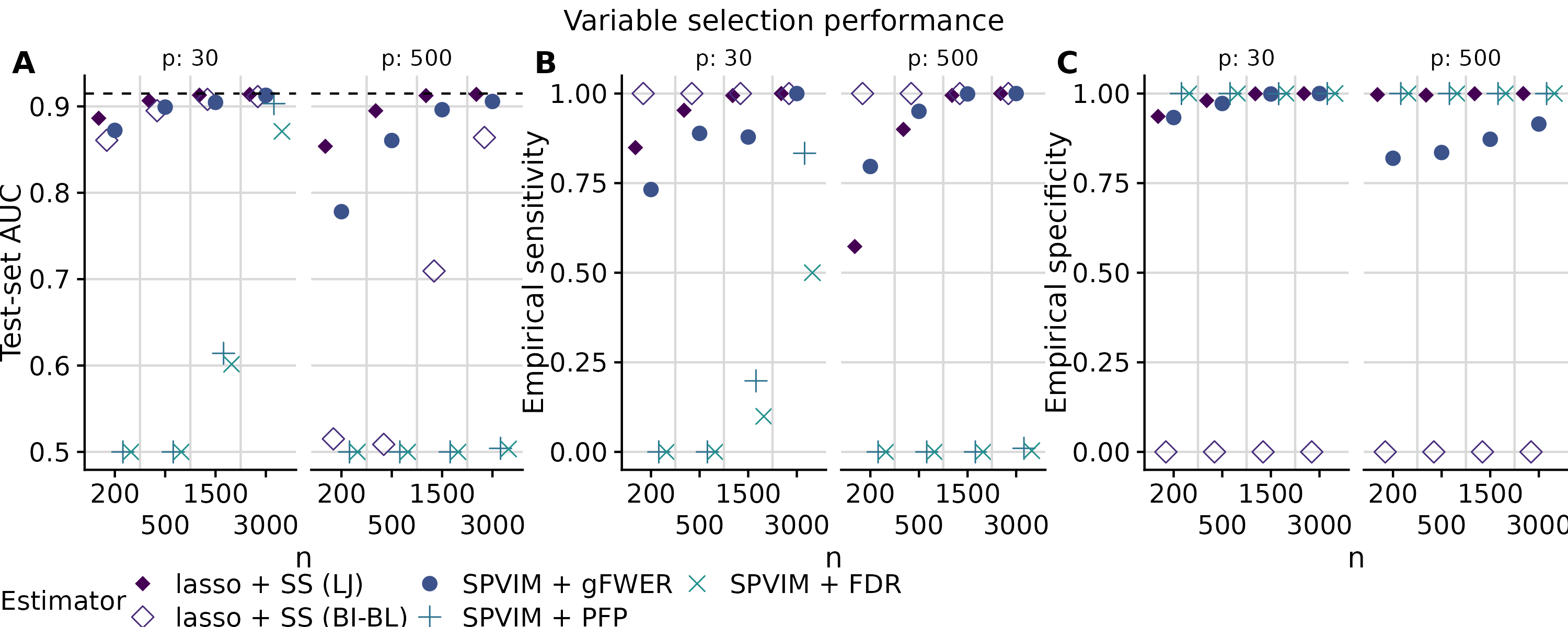}
  \caption{Test-set area under the receiver operating characteristic curve (AUC) (panel A) and empirical variable selection sensitivity (panel B) and specificity (panel C) vs $n$ for each estimator in the case with 40\% missing data}, in Scenario 1 (a linear model for the outcome and multivariate normal features). The dotted line in panel A shows the true (optimal) test-set AUC. The methods compared are: lasso + SS (LJ), the stability-selection within bootstrap imputation algorithm of \citep{long2015}; lasso + SS (BI-BL), the bootstrap imputation with bolasso algorithm of \citep{long2015}; SPVIM + gFWER, intrinsic selection to control the generalized familywise error rate; SPVIM + PFP, intrinsic selection to control the proportion of false positives among the rejected variables; and SPVIM + FDR, intrinsic selection to control the false discovery rate.
  \label{fig:scenario-1-select}
\end{figure}

In Figure \ref{fig:scenario-2-select}, we display the results of the experiment conducted under Scenario 2; the features are correlated multivariate normal and the outcome-feature relationship is nonlinear. We observe test-set AUC near the optimal value for the gFWER-controlling intrinsic selection procedure, while test-set AUC is much lower for the lasso-based procedures. We again observe lower test-set AUC for the PFP and FDR-controlling intrinsic procedures. We observe poor empirical sensitivity for the stability-selection within bootstrap imputation procedure, while we observe high sensitivity for the gFWER-controlling intrinsic procedure. Empirical specificity also tends to be high for this intrinsic procedure; among lasso-based estimators, the stability-selection within bootstrap imputation procedure has the highest empirical specificity, which tends to be lower than specificity for the gFWER-controlling intrinsic procedure.

\begin{figure}
  \centering
  \includegraphics[width=1\textwidth]{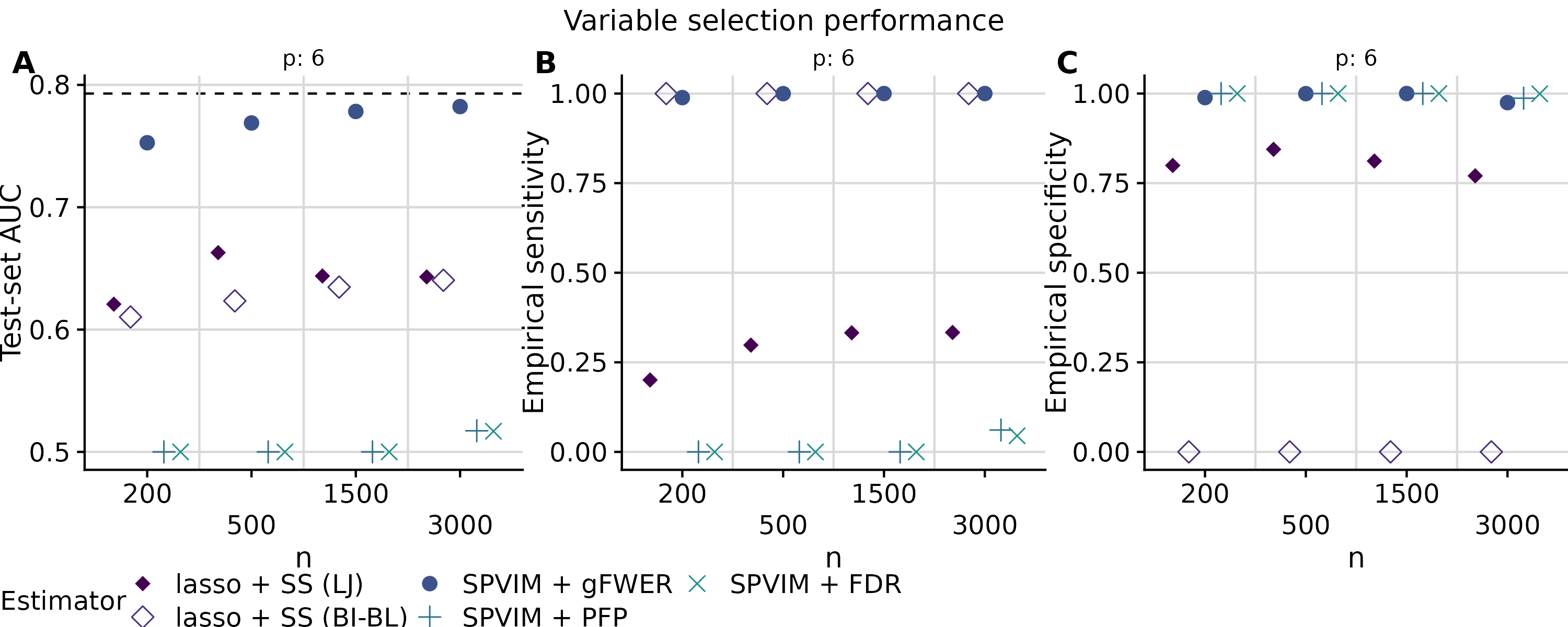}
  \caption{Test-set area under the receiver operating characteristic curve (AUC) (panel A) and empirical variable selection sensitivity (panel B) and specificity (panel C) vs $n$ for each estimator, in Scenario 2 (a nonlinear model for the outcome and correlated multivariate normal features). The dotted line in panel A shows the true (optimal) test-set AUC. The methods compared are: lasso + SS (LJ), the stability-selection within bootstrap imputation algorithm of \citep{long2015}; lasso + SS (BI-BL), the bootstrap imputation with bolasso algorithm of \citep{long2015}; SPVIM + gFWER, intrinsic selection to control the generalized familywise error rate; SPVIM + PFP, intrinsic selection to control the proportion of false positives among the rejected variables; and SPVIM + FDR, intrinsic selection to control the false discovery rate.}
  \label{fig:scenario-2-select}
\end{figure}

This simulation study suggests that the intrinsic variable selection procedures proposed here have good practical performance, as suggested by theory. As is the case with other procedures, we observed a tradeoff between sensitivity and specificity for our proposed procedures. In Scenario 2, where procedures based on a generalized linear model were misspecified, we observed poor variable selection and prediction performance when using lasso-based estimators, whereas our proposed methods protected against this model misspecification.

\subsection{Additional empirical results}

In the Supplementary Material, we present the 20\% missing data case for Scenarios 1 and 2, observing similar results to those presented in Figures~\ref{fig:scenario-1-select} and \ref{fig:scenario-2-select} (Section 7.4). We also consider the complete-data case, observing that our proposed intrinsic selection procedure has similar performance to the lasso with stability selection and knockoffs under Scenario 1, and improved performance under Scenario 2 (Section 7.6). The two scenarios that we described in the previous section are special cases of a more general setup described in Section 7.1 of the Supplementary Material. We consider six additional scenarios scrutinizing the effect of intermediate departures from the linear outcome regression and independent normal feature distribution (Sections 7.5 and 7.6). While a nonlinear outcome regression resulted in decreased test-set prediction performance and decreased probability of selecting some important variables for the lasso-based procedures, a nonnormal feature distribution had a minimal effect on the performance of these procedures. When the variables were equally weakly important, we observed poor performance of lasso-based estimators in cases with correlated predictors. Our intrinsic selection procedure maintained good overall performance in all scenarios, reflecting its robustness to model misspecification.

\section{Developing a biomarker panel for pancreatic cancer early detection}\label{sec:data}

Pancreatic ductal adenocarcinoma is the fourth-leading cause of cancer death in the United States. There is increasing focus on identifying pancreatic cancer at an early stage when treatment should be most effective. Mucinous cysts are one potential precursor lesion to pancreatic ductal adenocarcinoma and might be identified using routine imaging. However, imaging can be prohibitively expensive and current radiographic tests have limited ability to differentiate between benign and pre-malignant cystic neoplasms \citep{brugge2004}. This has spurred development of fluid biomarkers that can be assayed using pancreatic cyst fluid, which is routinely collected during clinical care.

We consider specimens from 321 participants with confirmed surgical pathology diagnosis from the Pancreatic Cyst Biomarker Validation Study \citep{liu2020}, designed to evaluate multiple cystic fluid biomarkers at several research institutes across the United States. The 21 candidate biomarkers are described further in the Supplementary Material (Table S5, Section 8). A main objective of the study is to develop biomarkers or biomarker panels that can be used to separate pancreatic cysts with differential malignant potentials. A major complication in achieving this objective is limited available cystic fluid volume from each study participant. The study statistical team randomly assigned available specimens to validation sites, such that each biomarker was only measured in a subset of the total study participants. This results in a highly non-monotone pattern of missingness in the biomarker data. Here the missing at random assumption holds since the probability of measuring a biomarker from an individual depends on that individual's specimen volume based on the specimen allocation scheme. Our goal here is to develop biomarker panels to separate mucinous cysts from non-mucinous cysts. In the Supplementary Material (Section 9), we present an analysis focused on malignancy potential.

We use the same procedures that we evaluated in the previous section. We assessed the prediction performance of each procedure through repeating an imputation-within-cross-validation procedure 100 times. We used MI with $M = 10$ in all cases, and used an outer layer of five-fold cross-validation to assess prediction performance. We obtained a final set of biomarkers selected by each procedure using Algorithm~\ref{alg:intrinsic-select} on the full imputed datasets. We chose tuning parameters based on similar settings considered in the simulations, leading us to set $k = 5$ and $q = 0.8$. More details on the approaches to estimating prediction performance and obtaining the final panel are provided in the Supplementary Material (Section 8).

We present the results of our analysis in Figure~\ref{fig:mucinous}. The PFP- and FDR-controlling intrinsic selection procedures did not select any variables on average, suggesting that the tuning parameters we selected were too conservative. The gFWER-controlling intrinsic selection procedure had high predictiveness, as measured by cross-validated AUC (CV-AUC), and was the top-performing algorithm with an average estimated CV-AUC of 0.946 and 95\% confidence interval of [0.89, 1]. Performance was worse for the lasso-based estimators, with an average estimated CV-AUC of 0.541 [0.385, 0.697] and 0.539 [0.383, 0.695] for the bootstrap imputation with bolasso and stability selection within bootstrap imputation lasso, respectively. In the Supplementary Material (Table S7, Section 8), we display the final set of biomarkers selected by each procedure. Since $k = 5$, we can interpret the final selected panel using the SPVIM + gFWER approach, which contained 10 variables, as having a 5\% probability of containing greater than five truly unimportant variables (i.e., variables with $\psi_{0,j} = 0$), where importance is defined with respect to increasing the AUC of a prediction algorithm based on all possible combinations of the measured biomarkers. Among the three procedures that selected variables, several biomarkers were selected by all procedures. These include biomarkers related to amphiregulin, glucose, fluorescent protease activity, and protein expression. Amphiregulin has been found to be elevated in adenocarcinoma cells \citep{tun2012}.

\begin{figure}
    \centering
    \includegraphics[width=1\textwidth]{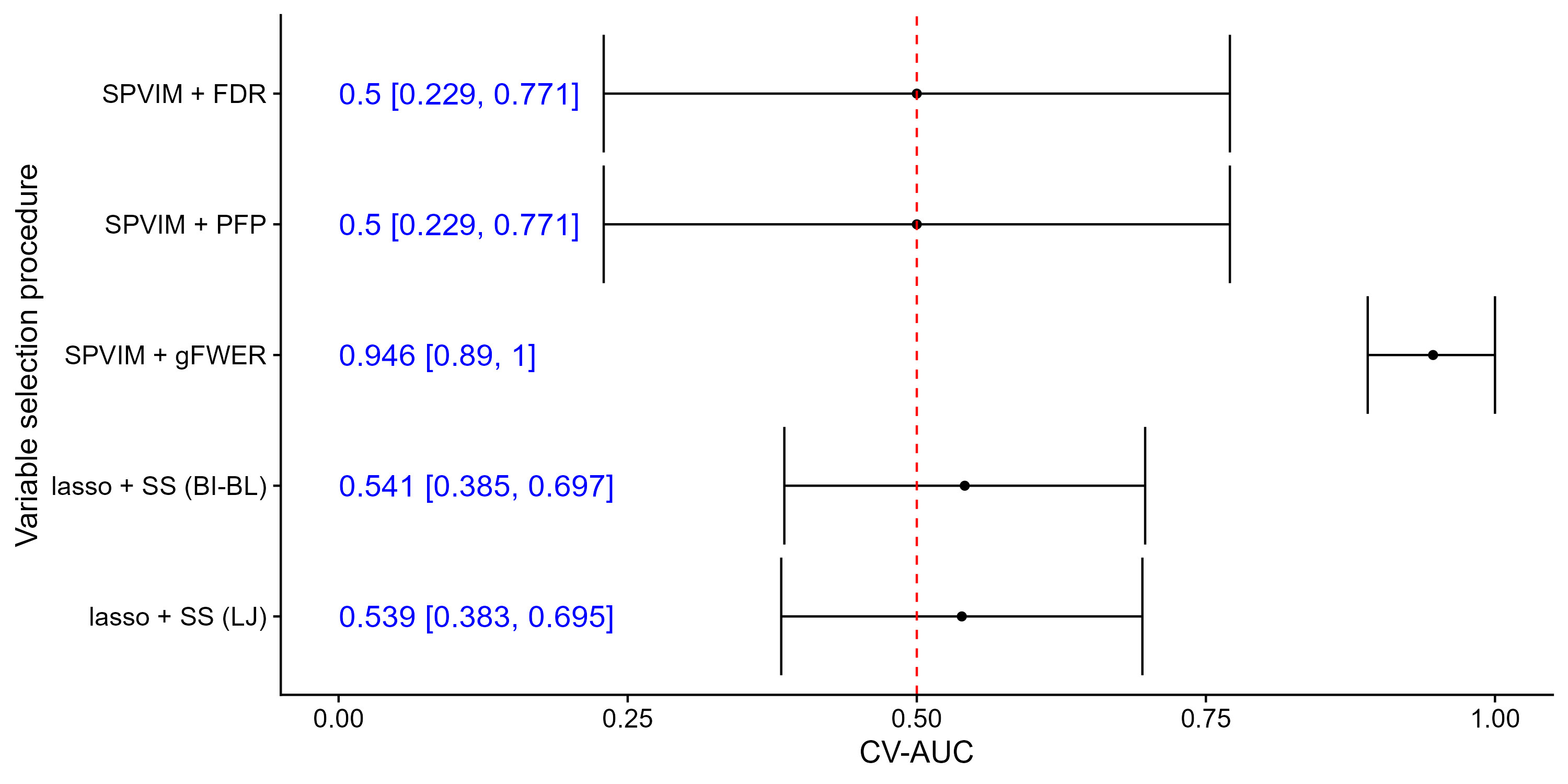}
    \caption{Cross-validated area under the receiver operating characteristic curve (CV-AUC) for predicting whether a cyst is mucinous averaged over 100 replicates of the imputation-within-cross-validated procedure for each variable selection algorithm. Prediction performance for lasso-based methods is based on logistic regression on the selected variables, while performance for Super Learner-based methods is based on a Super Learner. Error bars denote 95\% confidence intervals based on the average variance over the 100 replications. The methods compared are: lasso + SS (LJ), the stability-selection within bootstrap imputation algorithm of \citep{long2015}; lasso + SS (BI-BL), the bootstrap imputation with bolasso algorithm of \citep{long2015}; SPVIM + gFWER, intrinsic selection to control the generalized familywise error rate; SPVIM + PFP, intrinsic selection to control the proportion of false positives among the rejected variables; and SPVIM + FDR, intrinsic selection to control the false discovery rate.}
    \label{fig:mucinous}
\end{figure}

\section{Discussion}\label{sec:discussion}

We have proposed a variable selection procedure that is robust to model misspecification and is valid in settings with missing data, providing an alternative to existing, model-based approaches. We proved that our intrinsic selection procedure is persistant in complete-data settings and that error rate control can be achieved through the use of a tuning parameter, and identified conditions under which Rubin's rules can be used with intrinsic selection to formally incorporate imputation variance in settings with missing data.
We found in simulated examples that our proposal had high sensitivity and specificity and good overall prediction performance, though performance depends greatly on how control over the false discoveries is carried out.
We observed poor performance of the BI-BL lasso in our simulations. There, we followed the advice of \citet{long2015} and \citet{meinshausen2010} and set the threshold tuning parameter equal to 0.9. This poor performance suggests that the results of this procedure are more dependent on the choice of threshold in some settings than previously reported.
Importantly, in settings with missing data where a simple linear outcome regression model is correctly specified, our proposals have similar operating characteristics to the lasso-based procedures proposed in \citet{long2015}. In these settings with complete data, our proposals have similar operating characteristics to the lasso, lasso with stability selection, and lasso with knockoffs, all of which are commonly used. In settings with a nonlinear relationship where the linear outcome regression model was misspecified, weakly important features, and correlated features, we observed that our proposals maintained high sensitivity and specificity, while the performance of the lasso-based procedures suffered, as suggested by theory \citep[see, e.g.,][]{leng2006}.

In settings with missing data, many variable selection procedures require post-hoc harmonization of many selected sets resulting from multiply imputed datasets. A benefit of our proposed intrinsic selection procedure is that Rubin's rules can be used to obtain a single set of point and variance estimates accounting for the across-imputation variance, resulting in a single set of selected variables. In cases where the imputation mechanism is misspecified and incongenial with the analytic approach, it may be necessary to update the variance estimator \citep{robins2000}; however, the form of this estimator is complex. This idea is being pursued in ongoing research.

\section*{Software and supplementary materials}\label{sec:supp}
The proposed methods are implemented in the R package \texttt{flevr}, freely available on \href{https://github.com/bdwilliamson/flevr}{GitHub}. Supplementary Materials, including all technical proofs and code to reproduce all numerical experiments and data analyses, are available on GitHub at \url{https://github.com/bdwilliamson/flevr_supplementary}.

\section*{Acknowledgements}\label{sec:acknowledgements}
This work was supported by the National Institutes of Health (NIH) grants R37AI054165, R01GM106177, U24CA086368 and S10OD028685. The opinions expressed in this article are those of the authors and do not necessarily represent the official views of the NIH.

\vspace{0.1in}

{\small
\bibliographystyle{chicago}
\bibliography{brian-papers}
}

\newpage

\section*{SUPPLEMENTARY MATERIAL}
\input{intrinsic_select_supp.tex}

\end{document}

%% file: intrinsic_select_supp.tex
\renewcommand{\thefigure}{S\arabic{figure}}
\renewcommand{\theequation}{S\arabic{equation}}
\renewcommand{\thetable}{S\arabic{table}}
\renewcommand{\thethm}{S\arabic{thm}}
\renewcommand{\thelemma}{S\arabic{lemma}}
\setcounter{figure}{0}

\doublespacing

\section{Proofs of theorems}\label{sec:proofs}

\subsection{Regularity conditions}

This section is a review of the formal regularity conditions required to specify the distribution of the SPVIM values \citep{williamson2020c}. We define the linear space $\mathcal{R} := \{c(P_1 - P_2) \ : \ c \in \mathbb{R}, P_1, P_2 \in \mathcal{M}\}$ of finite signed measures generated by $\mathcal{M}$. For any $R \in \mathcal{R}$, we consider the supremum norm $\lVert R \rVert_\infty := \lvert c \rvert \sup_z \lvert F_1(z) - F_2(z)\rvert$, where $F_1$ and $F_2$ are the distribution functions corresponding to $P_1$ and $P_2$, respectively, and we have used the representation $R = c(P_1 - P_2)$. For distribution $P_{0,\epsilon} := P_0 + \epsilon h$ with $\epsilon \in \mathbb{R}$ and $h \in \mathcal{R}$, we define $f_{0,\epsilon,s} = f_{P_0,\epsilon,s}$ to be the oracle prediction function with respect to each subset $s \in \{1, \ldots, p\}$. Let $\dot{V}(f, P_0; h)$ denote the G\^ateaux derivative of $P \mapsto V(f, P)$ at $P_0$ in the direction $h \in \mathcal{R}$. The G\^ateaux derivatives for several common choices of $V$ are provided in \citet{williamson2021}. Next, we define the random function $g_{n,s}: z \mapsto \dot{V}(f_{n,s}, P_0 ; \delta_z - P_0) - \dot{V}(f_{0,s}, P_0; \delta_z - P_0)$, where $\delta_z$ is the degenerate distribution on $\{z\}$. For each $s \subseteq \{1,\ldots,p\}$, we require the following conditions to hold:
\begin{itemize}
   \item[(A1)] \textit{(optimality)} there is some $C > 0$ such that for each sequence $f_1, f_2, \cdots \in \mathcal{F}_s$ with $\lVert f_j - f_{0,s}\rVert_{\mathcal{F}_s} \to 0$, there is a $J$ such that for all $j > J$, $\lvert V(f_j, P_0) - V(f_{0,s}, P_0)\rvert \leq C \lVert f_j - f_{0,s}\rVert_{\mathcal{F}_s}^2$;
   \item[(A2)] there is some $\delta > 0$ such that for each sequence $\epsilon_1, \epsilon_2, \ldots \in \mathbb{R}$ and $h, h_1, h_2, \ldots \in \mathcal{R}$ satisfying that $\epsilon_j \to 0$ and $\lVert h_j - h \rVert_\infty \to 0$, it holds that
   \begin{align*}
     \sup_{f \in \mathcal{F}_s: \lVert f - f_{0,s} \rVert_{\mathcal{F}_s} < \delta} \big\lvert \frac{V(f, P_0 + \epsilon_j h_j) - V(f, P_0)}{\epsilon_j} - \dot{V}(f, P_0; h_j)\big\rvert \to 0;
   \end{align*}
   \item[(A3)] $\lVert f_{0,\epsilon,s} - f_{0,s} \rVert_{\mathcal{F}_s} = o(\epsilon)$ for each $h \in \mathcal{R}$;
   \item[(A4)] $f \mapsto \dot{V}(f, P_0; h)$ is continuous at $f_{0,s}$ relative to $\mathcal{F}_s$ for each $h \in \mathcal{R}$;
   \item[(A5)] $\lVert f_{n,s} - f_{0,s} \rVert_{\mathcal{F}_s} = o_P(n^{-1/4})$;
   \item[(A6)] $E_{P_0}[\int \{g_{n,s}(z)\}^2 dP_0(z)] = o_P(1)$;
   \item[(A7)] for $\gamma > 0$ and sequence $\gamma_1, \gamma_2, \ldots \in \mathbb{R}^+$ satisfying that $\lvert \gamma_j - \gamma \rvert \to 0$, $c = \gamma_n n$.
\end{itemize}

In settings with missing data, a modified version of (A5) and (A6) must hold for on average over the imputed datasets:
\begin{itemize}
  \item[(A5)] \textit{(in missing data settings)} $M^{-1}\sum_{m=1}^M\lVert f_{m,n,s} - f_{0,s} \rVert_{\mathcal{F}_s} = o_P(n^{-1/4})$;
  \item[(A6)] \textit{(in missing data settings)} $M^{-1}\sum_{m=1}^ME_{P_0}[\int \{g_{m,n,s}(z)\}^2 dP_0(z)] = o_P(1)$,
\end{itemize}
where $f_{m,n,s}$ is a prediction function estimated using the $m$th imputed dataset, and $g_{m,n,s}$ is defined as above but replacing all instances of $f_{n,s}$ with $f_{m,n,s}$, and replacing the ideal-data unit $z$ with the observed-data unit $o$.

\subsection{Proof of Lemma~\ref{lem:mi-normality}}

The result follows under conditions (A1)--(A8) and an application of results in Chapter 4 of \citet{rubin1987}. Using this result, we can write that
\begin{align*}
    \sqrt{n}(\psi_{M,c,n} - \psi_0) \to_d W \sim N(0, \sigma^2),
\end{align*}
where a consistent estimator of $\sigma^2$ is given by $\sigma^2_{M,n} + \frac{m+1}{m}\tau^2_{M,n}$. Recall that (A8) requires consistency of the imputation-based estimators as $M \to \infty$.

\subsection{Proof of Theorem~\ref{thm:error-control}}

Before proving the theorem, we state and prove a lemma that will be useful.

\begin{lemma}\label{lem:fwer_control}
  For any $\alpha \in (0, 1)$, $k \in \{0, \ldots, p - R_n(\alpha)\}$ and $q \in (0, 1)$, if conditions (A1)--(A6) hold for each $s \subseteq \{1,\ldots,p\}$ and (A7) holds, then the procedure $S_n(\alpha)$ satisfies the following: (a) when based on Holm-adjusted p-values, $FWER \leq \alpha$ both in finite samples and asymptotically; and (b) when based on a step-down maxT or minP procedure, $FWER \leq \alpha$ asymptotically.
\end{lemma}
\begin{proof}
  Under the collection of conditions (A1)--(A7), $\sqrt{n}(\psi_{c,n} - \psi_0) \to_d Z \sim N(0, \Sigma_0)$ by Theorem 1 in \citet{williamson2020c}, where $\Sigma_0 = E_0\{\phi_0(O)\phi_0(O)^\top\}$ and $\phi_0$ is the vector of efficient influence function values provided in \citet{williamson2020c} for each $j$. Therefore, the centered and scaled test statistics $T_n$ follow a multivariate Gaussian distribution.

  Thus, by Proposition 3.8 in \citet{dudoit2008}, when $S_n(\alpha)$ is based on Holm-adjusted p-values the procedure has finite-sample and asymptotic control of the FWER. When $S_n(\alpha)$ is based on a step-down maxT or minP procedure, the procedure has asymptotic control of the FWER as a result of Theorems 5.2 and 5.7 in \citet{dudoit2008}, respectively.
\end{proof}

Under conditions (A1)--(A7) and (B1)--(B2), an application of Lemma~\ref{lem:fwer_control} and Theorem 6.3 in \citet{dudoit2008} to the procedure $S_n^+(k, \alpha)$ yields that
\begin{align*}
  Pr_{P_0}(V_n^+(k, \alpha) > k) = \alpha_n \text{ and } Pr_{P_0}(V_n^+(k, \alpha) / R_n^+(k, \alpha) > q) = \alpha_n \text{ for all } n,
\end{align*}
i.e., the gFWER$(k)$ and PFP$(q)$ are controlled in finite samples at level $\alpha_n$.

If additionally conditions (B3)--(B4) hold, then an application of Lemma~\ref{lem:fwer_control} and Theorem 6.5 in \citet{dudoit2008} to the procedure $S_n^+(k, \alpha)$ yields that
\begin{align*}
  \limsup_{n\to\infty}Pr_{P_0}(V_n^+(k, \alpha) > k) \leq \alpha \text{ and } \limsup_{n\to\infty}Pr_{P_0}(V_n^+(k, \alpha) / R_n^+(k, \alpha) > q) \leq \alpha,
\end{align*}
i.e., the gFWER$(k)$ and PFP$(q)$ are controlled asymptotically at level $\alpha$.

Finally, under the above conditions, an application of Lemma~\ref{lem:fwer_control} and Theorem 6.6 in \citet{dudoit2008} to the procedure $S_n^+(k, \alpha)$ yields that the FDR is controlled asymptotically.

In missing-data settings, we simply require that condition (A8) additionally hold, and modify the above displays to use $S_{M,n}^+(\alpha)$, $V_{M,n}^+(\alpha)$, and $R_{M,n}^+(\alpha)$ in place of $S_n^+(\alpha)$, $V_n^+(\alpha)$, and $R_n^+(\alpha)$.

\subsection{Proof of Lemma~\ref{lem:persistence}}

Suppose that we are in a complete-data setting. Without loss of generality, suppose that we use Holm-adjusted p-values to construct the initial set of selected variables and that the augmented set is chosen so as to control the gFWER$(k)$. For a fixed sample size $n$ and constant $k_n$, this results in selected set $S_n := S_n^+(k_n, \alpha)$, where $\lvert S_n \rvert = k_n$. The claim of persistence is equivalent to showing that
\begin{align*}
  V(f_{n, S_n}, P_0) - V(f_*, P_0) \to_P 0.
\end{align*}
We can decompose the left-hand side of the above expression into two terms:
\begin{align}\label{eq:persistence_decomposition}
  V(f_{n, S_n}, P_0) - V(f_*, P_0) = \{V(f_{n, S_n}, P_0) - V(f_{0, S_n}, P_0)\} - \{V(f_{0, S_n}, P_0) - V(f_*, P_0)\}.
\end{align}
The first term in \eqref{eq:persistence_decomposition} is the contribution to the limiting behavior of $V(f_{n, S_n}, P_0) - V(f_*, P_0)$ from estimating $f_0$ for a fixed $S_n$; by condition (A1),
\begin{align*}
  \lvert V(f_{n, S_n}, P_0) - V(f_{0, S_n}, P_0) \rvert \leq C \lVert f_{n, S_n} - f_{0, S_n} \rVert_{\mathcal{F}_{S_n}}^2 \to_P 0.
\end{align*}

The second term in \eqref{eq:persistence_decomposition} is the contribution to the limiting behavior of $V(f_{n, S_n}, P_0) - V(f_*, P_0)$ from selecting $S_n$ compared to the population-optimal set. To study this term, recall that for a fixed $p$, we have under conditions (A1), (A2), (A5), (A6), and (A7) that $\psi_{c,n,j} \to_P \psi_{0,j}$ for each $j \in \{1, \ldots, p\}$. Thus, for each $j \in S_0$, the p-value $p_{n,j}$ associated with testing the null hypothesis $H_{0,j}: \psi_{0,j} = 0$ converges to 0. This implies that as $n \to \infty$, $S_n(\alpha) \to_P S_0$. Moreover, by condition (B3), $S_0 \subseteq S_n^+(k_n, \alpha)$ as $n \to \infty$. By definition, $\psi_{0,j} > 0$ if and only if $V(f_{0, s \bigcup \{j\}}, P_0) - V(f_{0,s}, P_0) > 0$ for some $s \subseteq \{1, \ldots, p\}$. This implies that for $j \in S_0^c$, $V(f_{0, s \bigcup \{j\}}, P_0) - V(f_{0,s}, P_0) = 0$ for all $s \subseteq \{1, \ldots, p\}$. In particular, for $j \in S_0^c$,
\begin{align*}
  V(f_{0, S_0 \bigcup \{j\}}, P_0) - V(f_{0, S_0}, P_0) = 0.
\end{align*}
This implies that $S_n^+(k_n, \alpha) \to_P S_0$, which further implies that $ \{V(f_{0, S_n}, P_0) - V(f_*, P_0)\} \to_P 0$, proving the claim with
\begin{align*}
  V(f_{n, S_n}, P_0) - V(f_*, P_0) = o_P(n^{-1/2}).
\end{align*}

In a setting with missing data, we consider the imputation-based analogue of the above result. Suppose that we have a selected set $S_{M,n} := S_{M,n}^+(\alpha)$. Then
\begin{align*}
  \frac{1}{M}\sum_{m=1}^M V(f_{m,n,S_{M,n}}, P_0) - V(f_*, P_0) = \frac{1}{M}\sum_{m=1}^M \{V(f_{m,n,S_{M,n}}, P_0) - V(f_{0, S_{M,n}})\} - \{V(f_{0,S_{M,n}} - V(f_*, P_0))\}.
\end{align*}
Under conditions (A1), (A2), and (A5)--(A8), the same logic applies to the second term in the above display as applied to the second term in Equation~\eqref{eq:persistence_decomposition}, so $ \{V(f_{0, S_{M,n}}, P_0) - V(f_*, P_0)\} \to_P 0$. For the first term in the display, an application of (A1) to each of the $m$ terms in the average yields the desired convergence in probability.

\section{Additional numerical experiments}\label{sec:more_sims}

\subsection{Replicating all numerical experiments}

All numerical experiments presented here and in the main manuscript can be replicated using code available on GitHub.

In all cases, our simulated dataset consisted of independent replicates of $(X, Y)$, where $X = (X_1, \ldots, X_p)$ and $Y$ followed a Bernoulli distribution with success probability $\Phi\{\beta_{00} + f(\beta_0, x)\}$ conditional on $X = x$, where $\Phi$ denotes the cumulative distribution function of the standard normal distribution. Under this specification, $Y$ followed a probit model. A summary of the eight scenarios is provided in Table~\ref{tab:scenarios}.

\begin{table}
    \centering
    \begin{tabular}{c|cccc}
        Scenario & Outcome regression & Feature distribution & Importance & $p$ \\
        \hline
        1 & Linear & Independent normal & Mix & $\{30, 500\}$ \\
        2 & Nonlinear & Correlated normal & Weak & $6$ \\
        3 & Linear & Independent nonnormal & Mix & $\{30, 500\}$ \\
        4 & Nonlinear & Independent normal & Mix & $\{30, 500\}$ \\
        5 & Nonlinear & Independent nonnormal & Mix & $\{30, 500\}$ \\
        6 & Linear & Independent normal & Weak & $6$ \\
        7 & Linear & Correlated normal & Weak & $6$ \\
        8 & Nonlinear & Independent normal & Weak & $6$
    \end{tabular}
    \caption{Summary of the eight data-generating scenarios considered in the numerical experiments.}
    \label{tab:scenarios}
\end{table}

In Scenarios 3--5, we investigate the effect of departures from a multivariate normal feature distribution and a linear outcome regression model under a similar setup to Scenario 1. We set $\beta_{00} = 0.5$ and $\beta_0 = (-1, 1, -0.5, 0.5, 1/3, -1/3, \mathbf{0}_{p - 6})^\top$, where $\mathbf{0}_k$ denotes a zero-vector of dimension $k$. We vary $p \in \{30, 500\}$. In Scenario 3, we set $f(\beta_0, x) = x\beta_0$, but in contrast to Scenario 1, $X$ follows a nonnormal feature distribution specified by
\begin{align}
    X_1 \sim N(0.5, 1); \ X_2 \sim Binomial(0.5);& \ X_3 \sim Weibull(1.75, 1.9); \ X_4 \sim Lognormal(0.5, 0.5); \notag \\
    X_5 \sim Binomial(0.5);& \ X_6 \sim N(0.25, 1); \ (X_7,\ldots,X_p) \sim MVN(0, I_{p-6}). \label{eq:nonnormal-dist}
\end{align}
In Scenarios 4 and 5, the outcome regression follows the same nonlinear specification as in Scenario 2. Specifically, using a centering and scaling function $c_j$ for each variable,
\begin{align}
  f(\beta_0, x) =& \ 2[\beta_{0,1} f_1\{c_1(x_1)\} + \beta_{0,2} f_2\{c_2(x_2), c_3(x_3)\} + \beta_{0,3} f_3\{c_3(x_3)\} \notag \\
  &\ \ + \beta_{0,4} f_4\{c_4(x_4)\} + \beta_{0,5} f_2\{c_5(x_5), c_1(x_1)\} + \beta_{0,6} f_5\{c_6(x_6)\}], \label{eq:nonlinear-outcome-supp}\\
  f_1(x) =& \ \sin\left(\frac{\pi}{4}x \right), f_2(x, y) = xy, f_3(x) = \tanh{(x)}, \notag \\
  f_4(x) = & \ \cos\left(\frac{\pi}{4}x\right), f_5(x) = -\tanh{(x)}, \notag
\end{align}
where $\tanh{}$ denotes the hyperbolic tangent. In Scenario 4, $X \sim MVN(0, I_p)$, while in Scenario 5, $X$ follows the distribution specified in Equation~\eqref{eq:nonnormal-dist}.
In these scenarios, only the first six features truly influence the outcome; some of the features are strongly important, while others are more weakly important.

In the final scenarios, we investigate the effect of correlated features and departures from a linear outcome regression model in a setting where the features are equally, and weakly, important; these settings are similar to Scenario 2. In these cases, we set $p = 6$, $\beta_{00} = 0.5$, $\beta_0 = (0, 1, 0, 0, 0, 1)^\top$, and $X \sim MVN(0, \Sigma)$, where $\Sigma_{i,j} = \rho_1^{\lvert i - j \rvert}$ for $i, j$ not in the active set, and $\Sigma_{i,j} = I_p + \rho_2(J_p - I_p)$ for $i, j$ in the active set, where $J_p$ is a $p\times p$ matrix of ones. In Scenarios 6 and 7 we set $f(\beta_0, x) = x\beta_0$, while in Scenario 8 $f$ is specified as in Equation~\eqref{eq:nonlinear-outcome-supp}. In Scenarios 6 and 8 we set $\rho_1 = \rho_2 = 0$, while in Scenario 7 we set $\rho_1 = 0.3$ and $\rho_2 = 0.95$.

\subsection{Tuning parameters for variable selection}

The tuning parameters that specify each variable selection procedure are as follows. For the intrinsic selection algorithm, we determined $k$ and $q$ for error control using a target specificity at $n = 3000$ of 75\% for $p = 6$, 85\% for $p = 30$, and 95\% for $p = 500$. For target specificity denoted by $s_p$ and $s_0 = \sum_{j=1}^p I(\beta_{0j} > 0)$, we set $k = \lceil(1 - s_p)(p - s_0) \rceil$, where $\lceil \cdot \rceil$ denotes the ceiling; and set $q = k\{p^{-1}(p-s_0)(n/200)^{1/2} + k\}^{-1}$. The exact values of $k$ (for $gFWER(k)$ control) and $q$ (for $PFP(q)$ control) are provided in Table~\ref{tab:gfwer_pfp_values}. For stability selection, we specified stability selection threshold equal to 0.9 and target per-comparison type I error rate of 0.04. For the lasso with knockoffs, we set target FDR equal to 0.2. 

For cases with missing data, the methods compared are: stability selection within bootstrap imputation, lasso + SS (LJ); bootstrap imputation with bolasso, lasso + SS (BI-BL); SPVIM + gFWER, intrinsic selection to control the generalized familywise error rate; SPVIM + PFP, intrinsic selection to control the proportion of false positives among the rejected variables; and SPVIM + FDR, intrinsic selection to control the false discovery rate.

For cases with complete data the methods compared are: lasso; lasso + SS, lasso with stability selection; lasso + KF, lasso with knockoffs; SPVIM + gFWER, intrinsic selection to control the generalized familywise error rate; SPVIM + PFP, intrinsic selection to control the proportion of false positives among the rejected variables; and SPVIM + FDR, intrinsic selection to control the false discovery rate.

\begin{table}
    \centering
    \begin{tabular}{cc|cccc}
        $n$ & $p$ & SS$_q$ & Target specificity & $k$ & $q$ \\
        \hline
        200 & 30 & 23 & 0.762 & 6 & 0.882 \\
        500 & 30 & 23 & 0.774 & 6 & 0.826 \\
        1500 & 30 & 23 & 0.809 & 5 & 0.695 \\
        3000 & 30 & 23 & 0.854 & 4 & 0.564 \\
        200 & 500 & 91 & 0.812 & 94 & 0.990 \\
        500 & 500 & 91 & 0.824 & 88 & 0.983 \\
        1500 & 500 & 91 & 0.861 & 69 & 0.962 \\
        3000 & 500 & 91 & 0.904 & 48 & 0.926 \\
    \end{tabular}
    \caption{Values of: the number of variables selected in each bootstrap run of stability selection (SS$_q$), target specificity for $gFWER(k)$ and $PFP(q)$ control, and $k$ and $q$ used for $gFWER$ and $PFP$ control, respectively, in the numerical experiments.}
    \label{tab:gfwer_pfp_values}
\end{table}

\subsection{Super Learner specification}

The specific candidate learners and their corresponding tuning parameters for our Super Learner library are provided in Tables~\ref{tab:sl-algs} (Scenarios 1, 3--5) and \ref{tab:sl-algs-2} (Scenarios 2, 6--8). In both cases, we used a wide variety of algorithms, each with several tuning parameter values, in an effort to be robust to model misspecification. It is possible that with a different library of learners, different results could be obtained. 

For the internal library in our intrinsic selection procedure in Scenarios 1 and 3--5, we first pre-screened variables based on their univariate rank correlation with the outcome, and then fit boosted trees with maximum depth equal to three and shrinkage equal to 0.1. In Scenarios 2 and 6--8, we again first pre-screened variables based on their univariate rank correlation with the outcome, and then fit a logistic regression or boosted trees with maximum depth equal to four, shrinkage equal to 0.1, and number of rounds equal to 100. Recall that within the intrinsic selection procedure, we estimate the optimal prediction function for each subset $s$ of the $p$ features. The univariate rank correlation screen operated as follows: if $\lvert s\rvert \leq 2$, we did no screening; if $2 < \lvert s \rvert < 100$, we picked the top two variables ranked by univariate correlation with the outcome; and if $\lvert s \rvert \geq 100$, we picked the top ten variables ranked by univariate correlation with the outcome. This screening substantially reduced the computation time for the intrinsic selection procedure, and reflects the type of aggressive screen that is used in some cases \citep{neidich2019}. Also, the univariate comparisons of each feature to the null model (with no features) are given high weight in the intrinsic importance measure, so screening should not have much impact on the final intrinsic importance estimate.

\begin{table}
    \centering
    \begin{tabular}{c|ccc}
       Candidate Learner & R & Tuning Parameter & Tuning parameter  \\
       & Implementation & and possible values & description\\ \hline
        Random forests & \texttt{ranger} & \texttt{mtry} $\in \{1/2, 1, 2\}\sqrt{p}$ ${}^{\dagger}$ & Number of variables \\
        & \citep{rangerpkg} & & to possibly split \\
        & & & at in each node \\ \hline
        Gradient boosted & \texttt{xgboost} & \texttt{max.depth} $ \in \{1, 3\}$ &  Maximum tree depth\\
        trees & \citep{xgboostpkg} & & \\ \hline
        Support vector & \texttt{ksvm} &  & \\
        machines & \citep{kernlabpkg} & & \\ \hline
        Lasso & \texttt{glmnet} & $\lambda$ & $\ell_1$ regularization  \\
        & \citep{glmnetpkg} & chosen via 10-fold CV & parameter \\ \hline
    \end{tabular}
    \caption{Candidate learners in the Super Learner ensemble for Scenarios 1 and 3--5 along with their R implementation, tuning parameter values, and description of the tuning parameters. All tuning parameters besides those listed here are set to their default values. In particular, the random forests are grown with 500 trees, a minimum node size of 5 for continuous outcomes and 1 for binary outcomes, and a subsampling fraction of 1; the boosted trees are grown with a maximum of 1000 trees, shrinkage rate of 0.1, and a minimum of 10 observations per node; and the SVMs are fit with radial basis kernel, cost of constraints violation equal to 1, upper bound on training error (\texttt{nu}) equal to 0.2, \texttt{epsilon} equal to 0.1, and three-fold cross-validation with a sigmoid for calculating class probabilities. \\
    ${}^{\dagger}$: $p$ denotes the total number of predictors. }
    \label{tab:sl-algs}
\end{table}

\begin{table}
    \centering
    \begin{tabular}{c|ccc}
       Candidate Learner & R & Tuning Parameter & Tuning parameter  \\
       & Implementation & and possible values & description\\ \hline
        Random forests & \texttt{ranger} & \texttt{min.node.size} $\in $ & Minimum  \\
        &  & $\{1, 20, 50, 100, 250, 500\}$ & node size\\ \hline
        Gradient boosted & \texttt{xgboost} & \texttt{shrinkage} $ \in \{1\times 10^{-2}, 1\times 10^{-1}\}$ &  Shrinkage\\
        trees &  & \texttt{ntrees} $\in \{100, 1000\}$ & Number of trees\\ \hline
        Support vector & \texttt{ksvm} &  & \\
        machines &  & & \\ \hline
        Lasso & \texttt{glmnet} & $\lambda$ & $\ell_1$ regularization  \\
        &  & chosen via 10-fold CV & parameter \\ \hline
    \end{tabular}
    \caption{Candidate learners in the Super Learner ensemble for Scenarios 2 and 6--8 along with their R implementation, tuning parameter values, and description of the tuning parameters. All tuning parameters besides those listed here are set to their default values. In particular, the random forests are grown with 500 trees and a subsampling fraction of 1; the boosted trees are grown with a minimum of 10 observations per node; and the SVMs are fit with radial basis kernel, cost of constraints violation equal to 1, upper bound on training error (\texttt{nu}) equal to 0.2, \texttt{epsilon} equal to 0.1, and three-fold cross-validation with a sigmoid for calculating class probabilities.}
    \label{tab:sl-algs-2}
\end{table}

\subsection{Additional results from Scenarios 1 and 2 with missing data}

In the main manuscript, we presented results with a maximum of 40\% missing data in some variables in Scenarios 1 and 2. In Figure~\ref{fig:scenario-1-select-supp} we present results in an intermediate setting with a maximum of 20\% missing data in some variables; the results in this setting tend to be similar to the results with maximum 40\% missing data.

\begin{figure}
  \centering
  \includegraphics[width=1\textwidth]{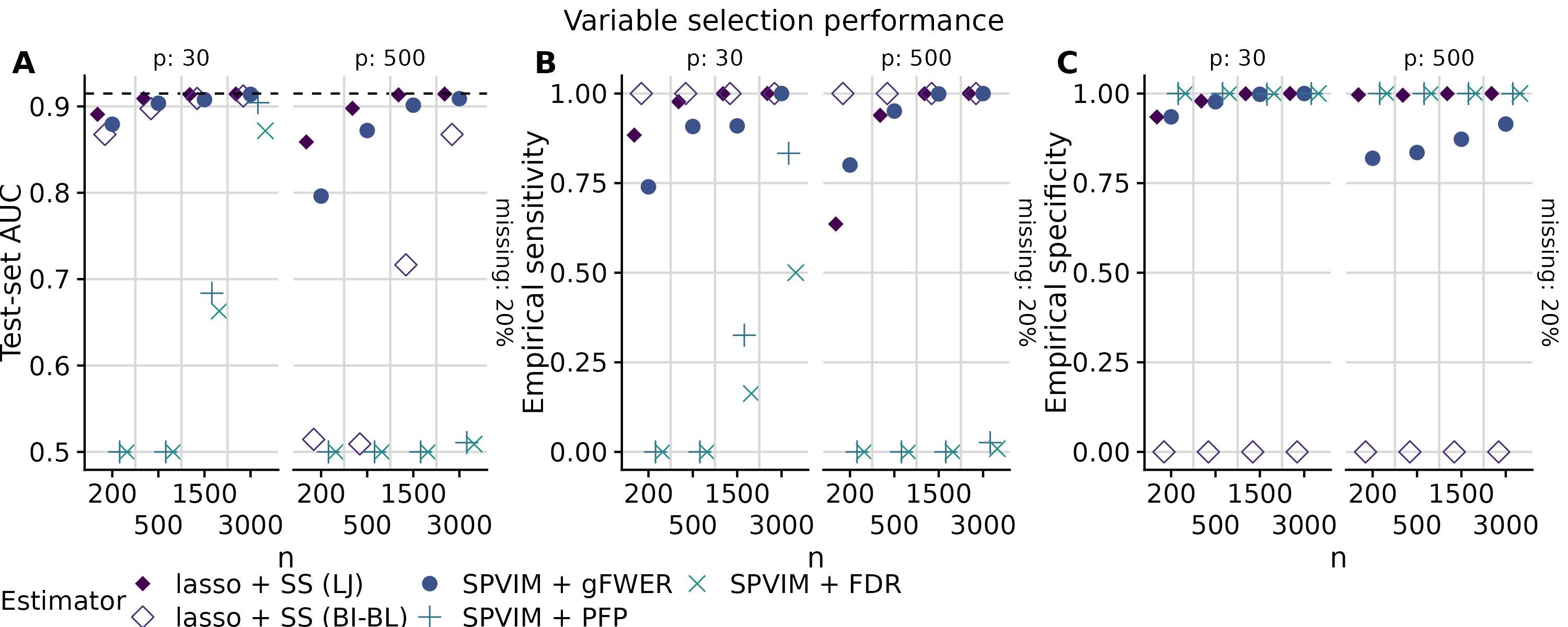}
  \caption{Test-set AUC (panel A) and empirical variable selection sensitivity (panel B) and specificity (panel C) vs $n$ for each estimator and missing data proportion equal to 0.2, in Scenario 1 (a linear model for the outcome and multivariate normal features). The dotted line in panel A shows the true (optimal) test-set AUC.}
  \label{fig:scenario-1-select-supp}
\end{figure}

In In Figure~\ref{fig:scenario-2-select-supp} we present results in an intermediate setting with a maximum of 20\% missing data in some variables, which again tend to be similar to the results with maximum 40\% missing data.

\begin{figure}
  \centering
  \includegraphics[width=1\textwidth]{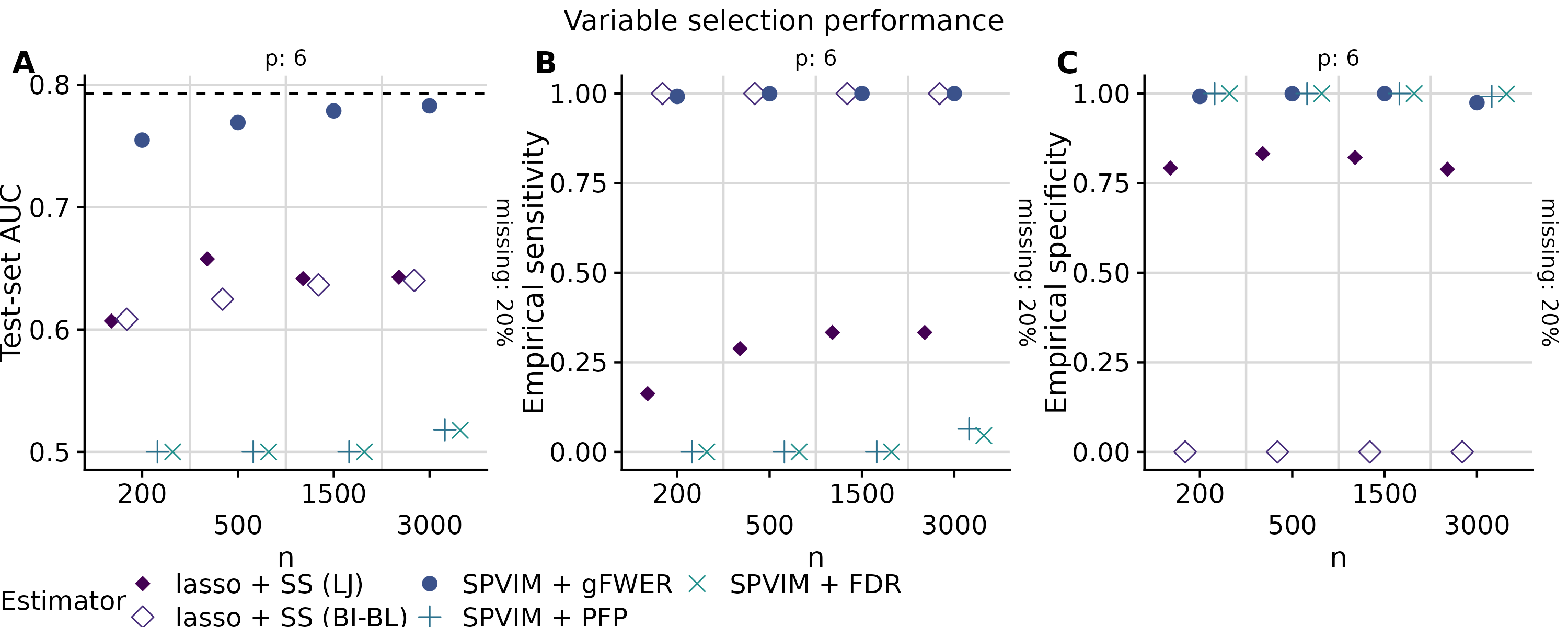}
  \caption{Test-set AUC (panel A) and empirical variable selection sensitivity (panel B) and specificity (panel C) vs $n$ for each estimator and missing data proportion equal to 0.2, in Scenario 2 (a nonlinear model for the outcome and correlated multivariate normal features). The dotted line in panel A shows the true (optimal) test-set AUC.}
  \label{fig:scenario-2-select-supp}
\end{figure}

In Figures~\ref{fig:scenario-1-probs-20}--\ref{fig:scenario-1-probs-40}, we display the empirical selection probability for each active-set variable under each selection algorithm in Scenario 1. All active-set variables are selected with high probability by all procedures, with the exception of SPVIM + FDR and SPVIM + PFP. In small samples, all estimators besides lasso + SS (BI-BL) sometimes fail to select variables 5 and 6, the variables with smallest intrinsic importance; these variables are selected with low probability by SPVIM + PFP and SPVIM + FDR at all sample sizes considered here. In the higher dimensional case, SPVIM + gFWER selects these variables in cases where lasso + SS (LJ) does not. This reflects the low true importance of these variables combined with tuning parameters that provide strict PFP and FDR control. As the proportion of missing data increases, the selection probabilities tend to decrease slightly.

\begin{figure}
  \centering
  \includegraphics[width=1\textwidth]{{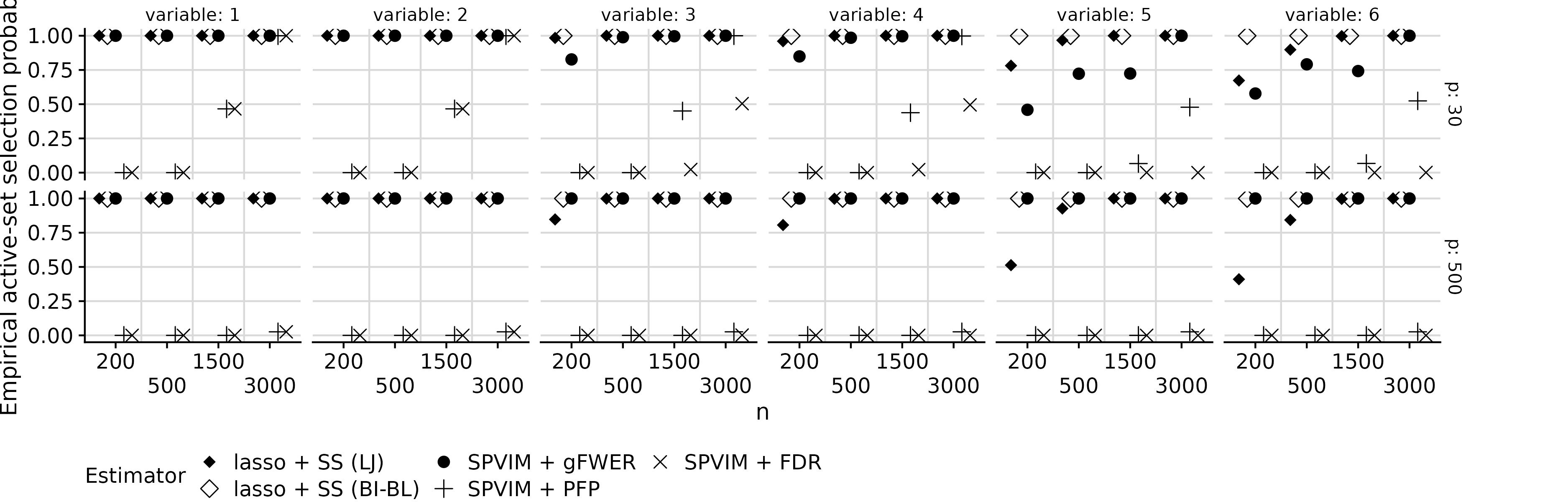}}
  \caption{Empirical selection probability for each active-set variable vs $n$ for each estimator and dimension with missing data proportion equal to 0.2, in Scenario 1 (a linear model for the outcome and multivariate normal features).}
  \label{fig:scenario-1-probs-20}
\end{figure}

\begin{figure}
  \centering
  \includegraphics[width=1\textwidth]{{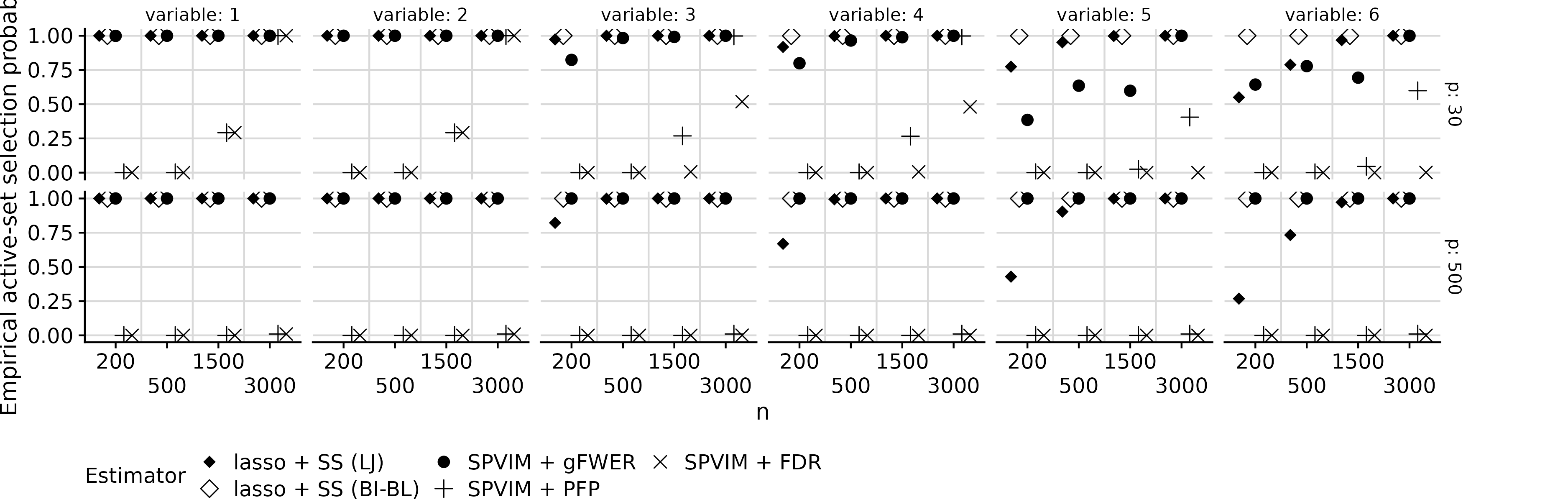}}
  \caption{Empirical selection probability for each active-set variable vs $n$ for each estimator and dimension with missing data proportion equal to 0.4, in Scenario 1 (a linear model for the outcome and multivariate normal features).}
  \label{fig:scenario-1-probs-40}
\end{figure}

In Figures~\ref{fig:scenario-2-probs-20}--\ref{fig:scenario-2-probs-40}, we display the empirical selection probability for each active-set variable under each selection algorithm in Scenario 2. In this scenario, as expected, the selection probability is low for lasso + SS (LJ) and high for SPVIM + gFWER (as reflected in the empirical sensitivity presented in the main manuscript). Variables 2 and 3, which are highly correlated and include an interaction term not modelled by the lasso, have the lowest selection probability for lasso + SS (LJ), as expected (though lasso + SS (BI-BL) has perfect sensitivity, it also has zero specificity).

\begin{figure}
  \centering
  \includegraphics[width=1\textwidth]{{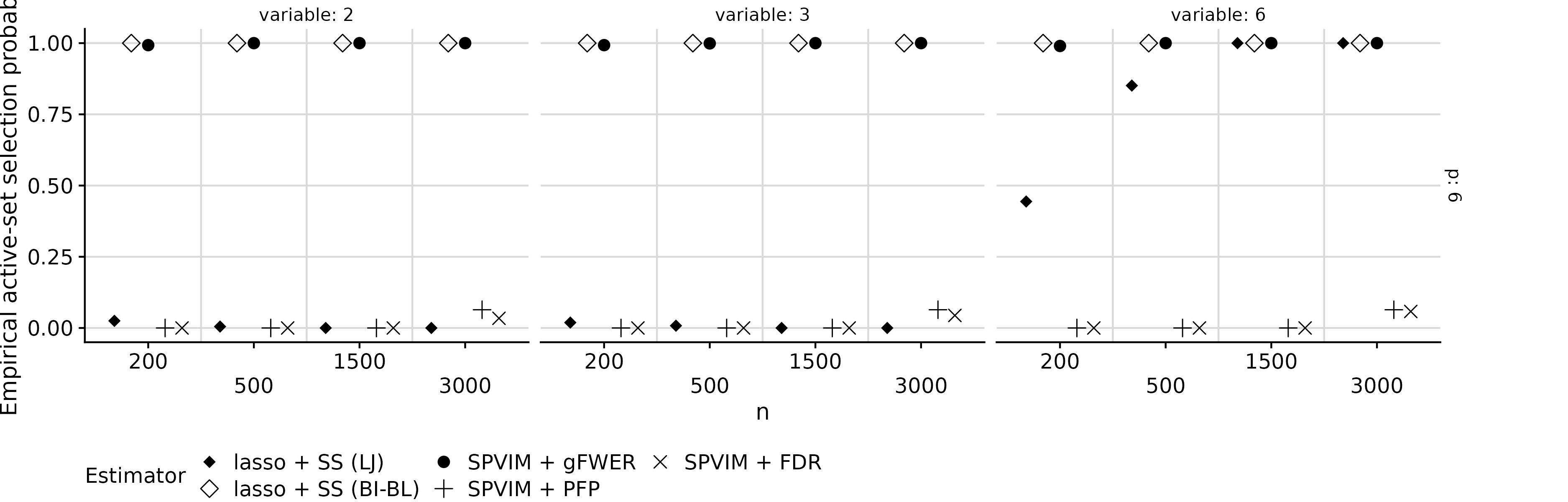}}
  \caption{Empirical selection probability for each active-set variable vs $n$ for each estimator, in Scenario 2 (a nonlinear model for the outcome and correlated multivariate normal features).}
  \label{fig:scenario-2-probs-20}
\end{figure}

\begin{figure}
  \centering
  \includegraphics[width=1\textwidth]{{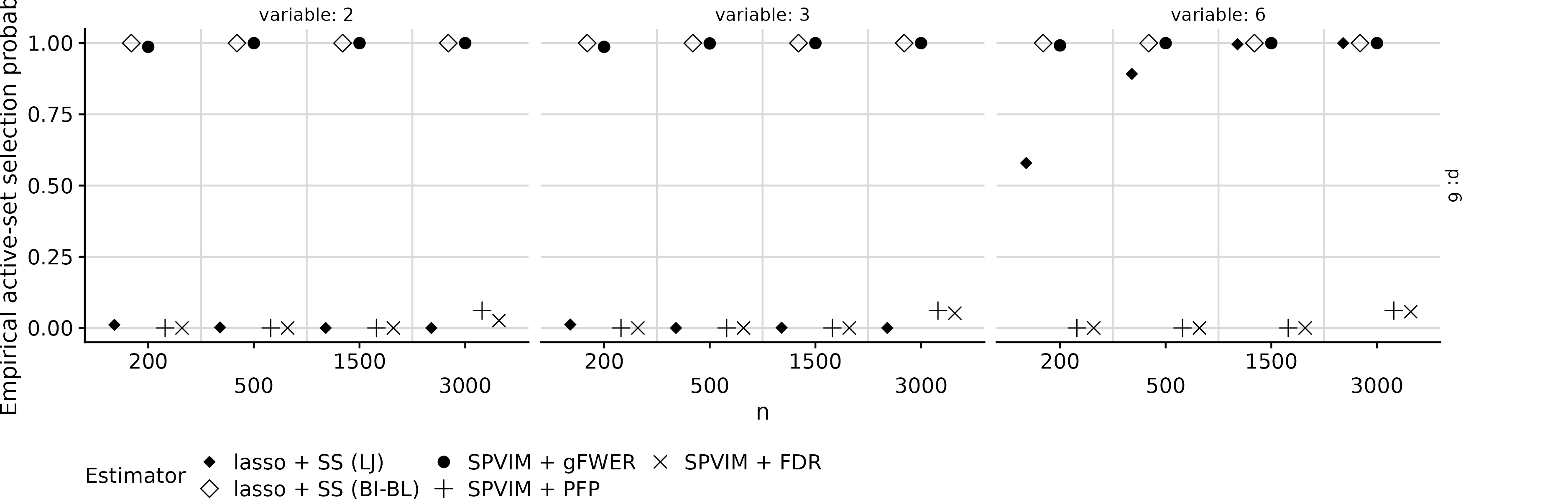}}
  \caption{Empirical selection probability for each active-set variable vs $n$ for each estimator, in Scenario 2 (a nonlinear model for the outcome and correlated multivariate normal features).}
  \label{fig:scenario-2-probs-40}
\end{figure}

\subsection{Results from Scenarios 3--8 with missing data}

In Scenario 3, we generate features from a nonnormal joint distribution and the outcome is a linear combination of these features. We display the results of this experiment in Figure \ref{fig:scenario-3-select}. We observe similar performance in this scenario to the performance we observed in Scenario 1: test-set AUC increases towards the optimal value with increasing sample size for all estimators, though slowest for SPVIM + FDR and SPVIM + PFP; empirical sensitivity and specificity tend to both increase, with the exception of the lasso + SS (BI-BL) algorithm, which has near-zero specificity at all sample sizes considered here.

\begin{figure}
  \centering
  \includegraphics[width=1\textwidth]{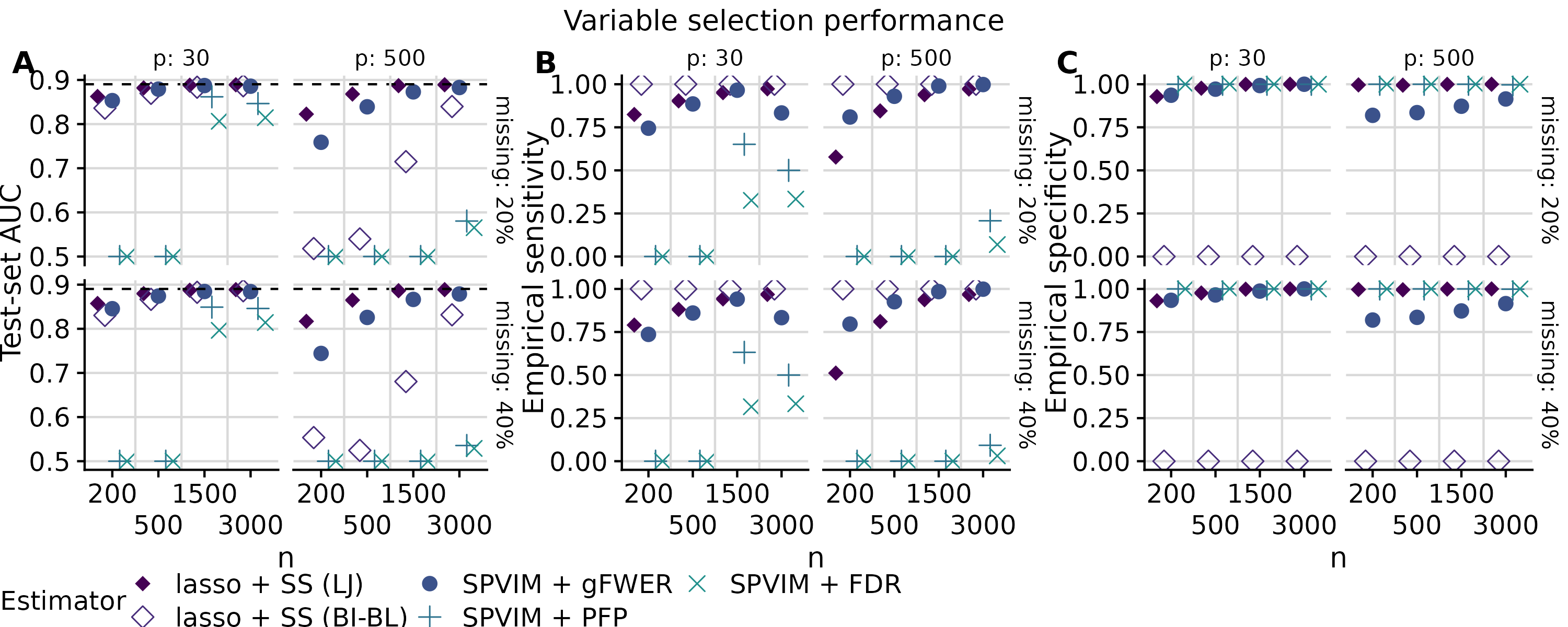}
  \caption{Test-set AUC (panel A) and empirical variable selection sensitivity (panel B) and specificity (panel C) vs $n$ for each estimator and missing data proportion, in Scenario 3 (a linear model for the outcome and nonnormal features). The dotted line in panel A shows the true (optimal) test-set AUC.}
  \label{fig:scenario-3-select}
\end{figure}

In Scenario 4, we generate features from a multivariate normal distribution and the outcome is a nonlinear combination of these features. In this case, lasso-based methods follow a misspecified mean model. We display the results of this experiment in Figure~\ref{fig:scenario-4-select}. We observe that test-set AUC tends to increase quickly towards the optimal AUC with increasing sample size for the SPVIM + gFWER procedure, but increases more slowly for lasso-based procedures; empirical sensitivity and specificity tend to both increase, with the exception of the lasso + SS (BI-BL) algorithm, which again has near-zero specificity at all sample sizes considered here. In this case, among the algorithms with non-zero specificity, SPVIM + gFWER has the highest sensitivity at all sample sizes considered here.

\begin{figure}
  \centering
  \includegraphics[width=1\textwidth]{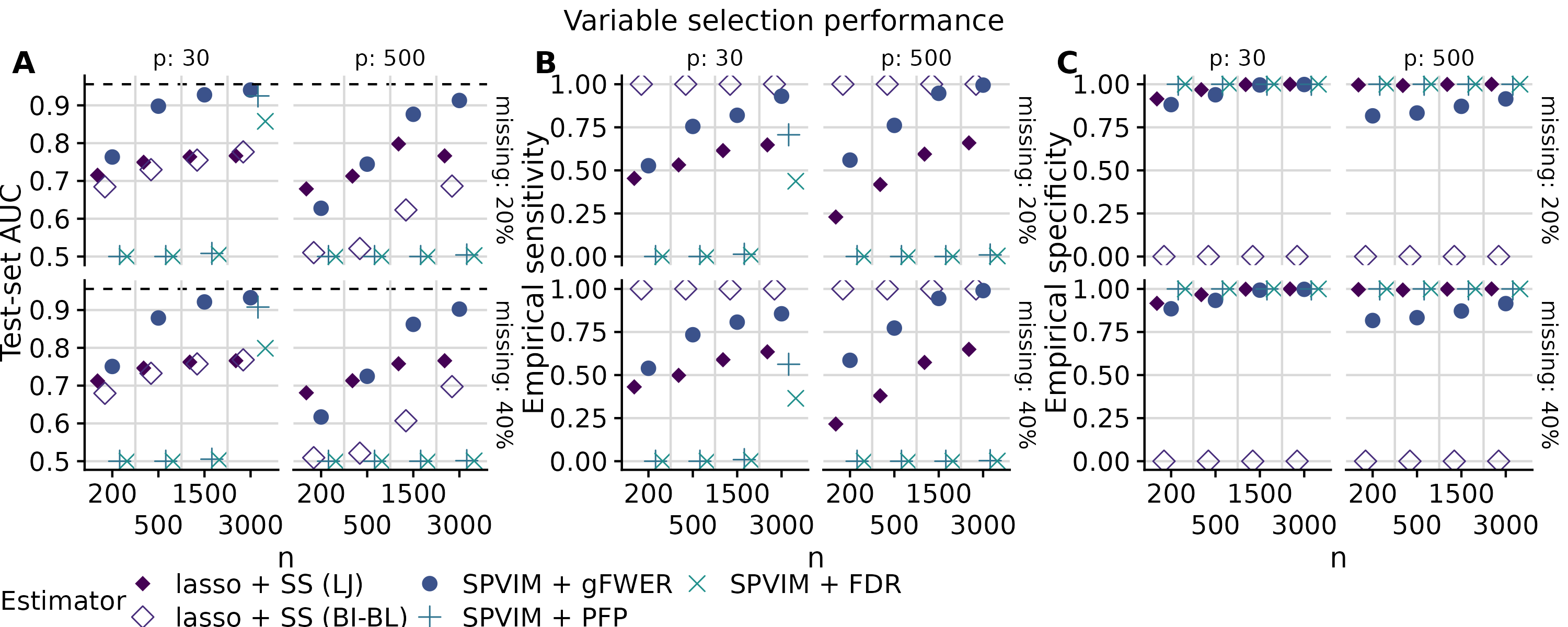}
  \caption{Test-set AUC (panel A) and empirical variable selection sensitivity (panel B) and specificity (panel C) vs $n$ for each estimator and missing data proportion, in Scenario 4 (a nonlinear model for the outcome and normal features). The dotted line in panel A shows the true (optimal) test-set AUC.}
  \label{fig:scenario-4-select}
\end{figure}

In Figure~\ref{fig:scenario-5-select}, we display the results of the experiment conducted under Scenario 5, in which the features are nonnormal and the outcome-feature relationship is nonlinear. In this case, the lasso-based methods are misspecified. In panel A, we observe that lasso-based methods have test-set AUC increasing slowly with $n$, while SPVIM + gFWER has test-set AUC approaching the optimal value more quickly. In panels B and C, we see that sensitivity tends to be lower than in Scenario 1 for all procedures, though still increasing towards one; and that specificity trends are similar to those in Scenario 1. In all cases considered here, SPVIM + gFWER has higher empirical sensitivity than lasso + SS (LJ), and often has comparable specificity, particularly in the lower-dimensional setting.

\begin{figure}
  \centering
  \includegraphics[width=1\textwidth]{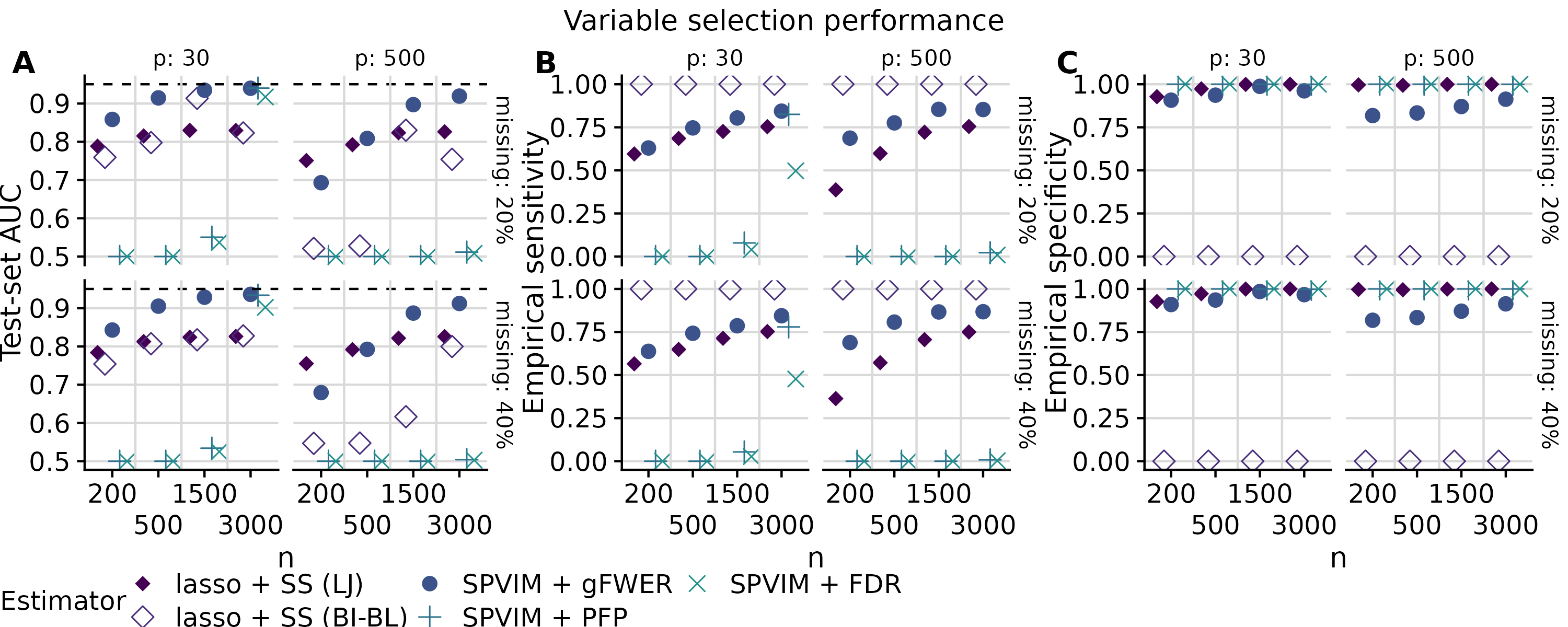}
  \caption{Test-set AUC (panel A) and empirical variable selection sensitivity (panel B) and specificity (panel C) vs $n$ for each estimator and missing data proportion, in Scenario 5 (a nonlinear model for the outcome and nonnormal features). The dotted line in panel A shows the true (optimal) test-set AUC.}
  \label{fig:scenario-5-select}
\end{figure}

In Scenarios 6--8, the features are more weakly important. We present the results of the experiments under these scenarios in Figures~\ref{fig:scenario-6-select}--\ref{fig:scenario-8-select}. In Scenario 7, we observe reduced variable selection performance for the lasso-based procedures compared to Scenario 6. In Scenario 8, we observe similar trends to Scenario 2, though performance for the lasso-based methods tends to be better than the performance we observed in Scenario 2, reflecting that this scenario does not involve correlation among the features. These experiments suggest that correlation makes variable selection more difficult, particularly in combination with a misspecified outcome regression model.

\begin{figure}
  \centering
  \includegraphics[width=1\textwidth]{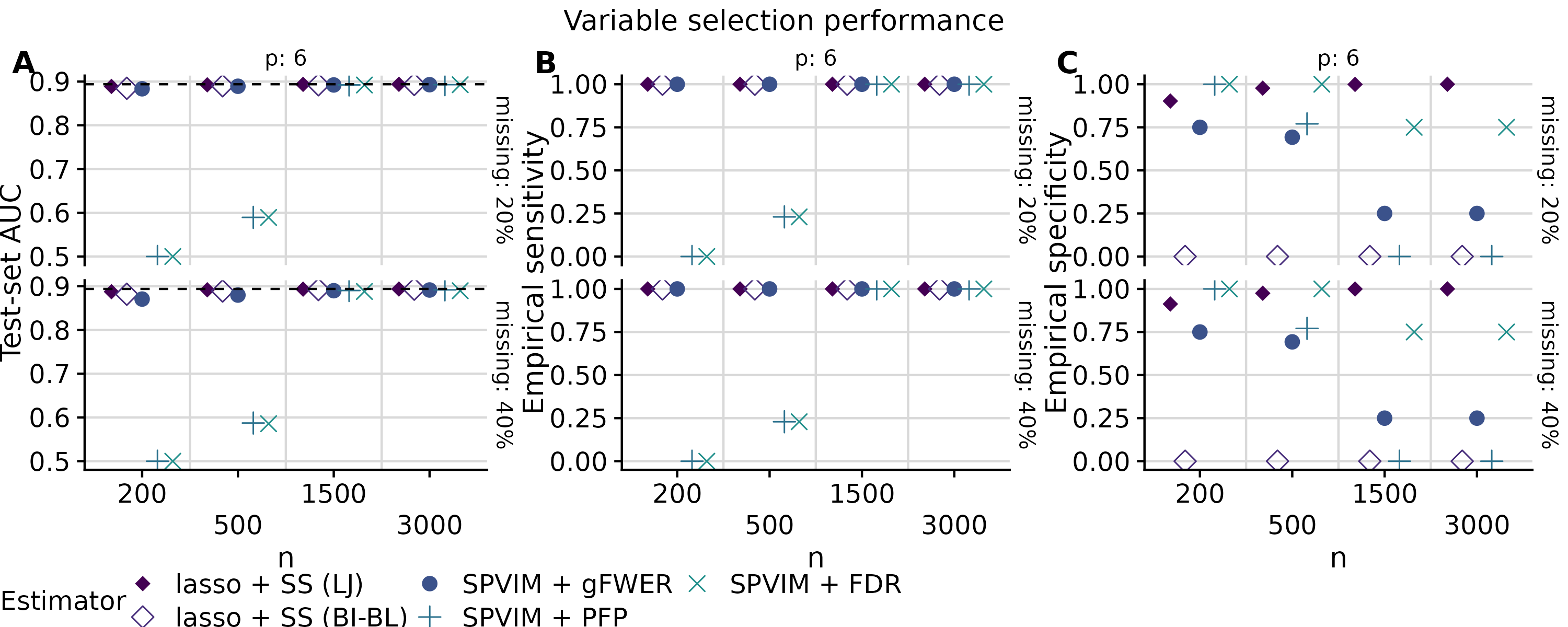}
  \caption{Test-set AUC (panel A) and empirical variable selection sensitivity (panel B) and specificity (panel C) vs $n$ for each estimator and missing data proportion, in Scenario 6 (a weak linear model for the outcome and normal features). The dotted line in panel A shows the true (optimal) test-set AUC.}
  \label{fig:scenario-6-select}
\end{figure}

\begin{figure}
  \centering
  \includegraphics[width=1\textwidth]{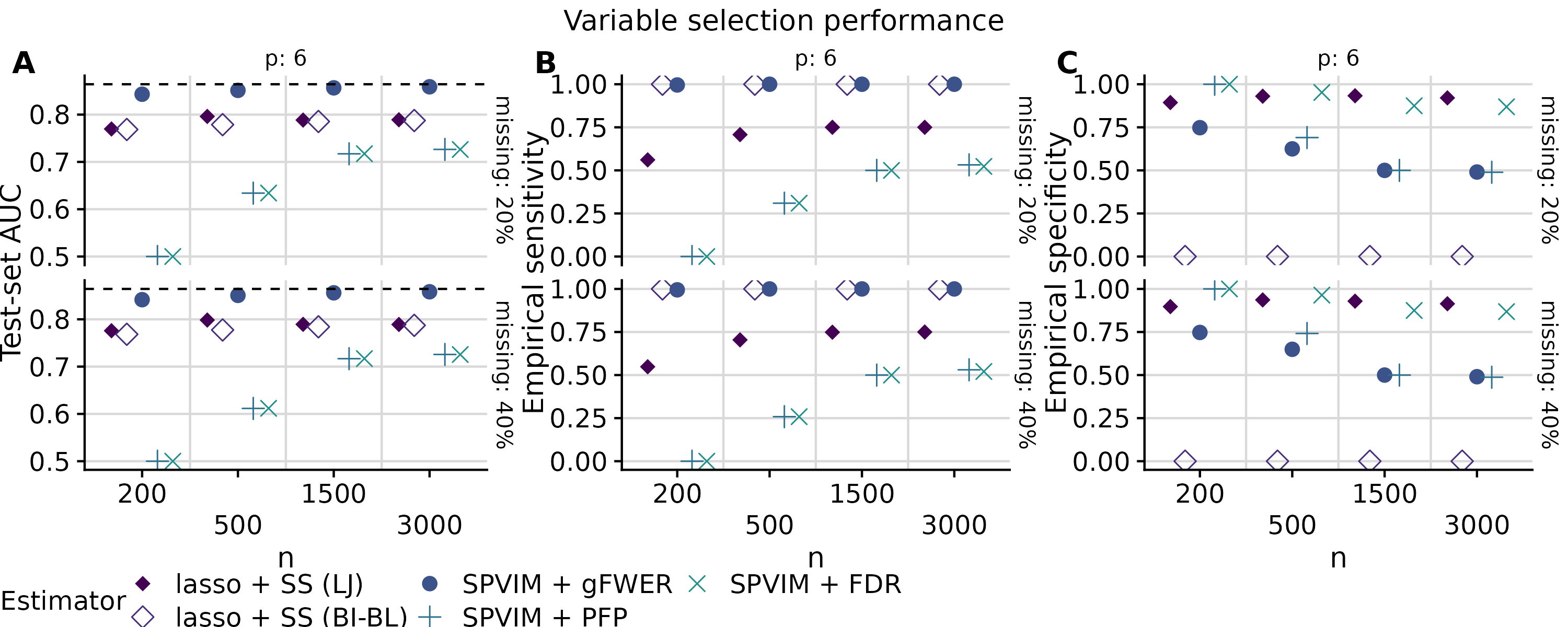}
  \caption{Test-set AUC (panel A) and empirical variable selection sensitivity (panel B) and specificity (panel C) vs $n$ for each estimator and missing data proportion, in Scenario 7 (a weak nonlinear model for the outcome and correlated normal features). The dotted line in panel A shows the true (optimal) test-set AUC.}
  \label{fig:scenario-7-select}
\end{figure}

\begin{figure}
  \centering
  \includegraphics[width=1\textwidth]{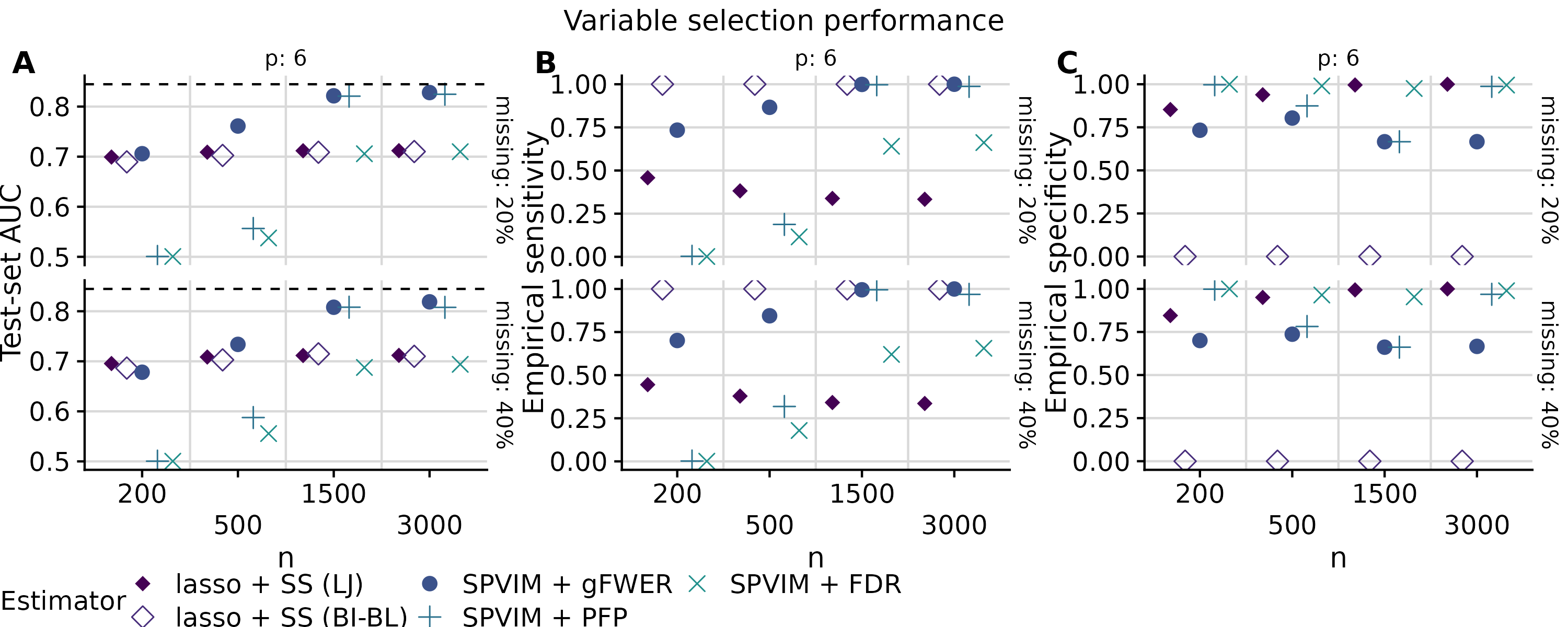}
  \caption{Test-set AUC (panel A) and empirical variable selection sensitivity (panel B) and specificity (panel C) vs $n$ for each estimator and missing data proportion, in Scenario 8 (a weak nonlinear model for the outcome and normal features). The dotted line in panel A shows the true (optimal) test-set AUC.}
  \label{fig:scenario-8-select}
\end{figure}

In Figures~\ref{fig:scenario-3-probs-20}--\ref{fig:scenario-8-probs-40}, we display the empirical selection probability for each active-set variable under each selection algorithm in Scenarios 3--8. We observe similar performance in Scenario 3 as in Scenario 1. In Scenarios 3 and 4, we observe that most procedures select variables 1, 2, 3, 4, and 6 with high probability as sample size increases. However, in the higher-dimensional case lasso-based procedures select variable 5 with lower probability than our proposed intrinsic selection procedure. Variable 5 is moderately important (its coefficient is 1, compared to a maximum coefficient of 2), but the function relating this variable to the outcome is highly nonlinear over its support. In Scenario 6--8, we observe similar patterns to Scenario 5: variables 2 and 3 tend to be selected infrequently by the lasso-based procedures, but with high frequency by the intrinsic selection procedure.

\begin{figure}
  \centering
  \includegraphics[width=1\textwidth]{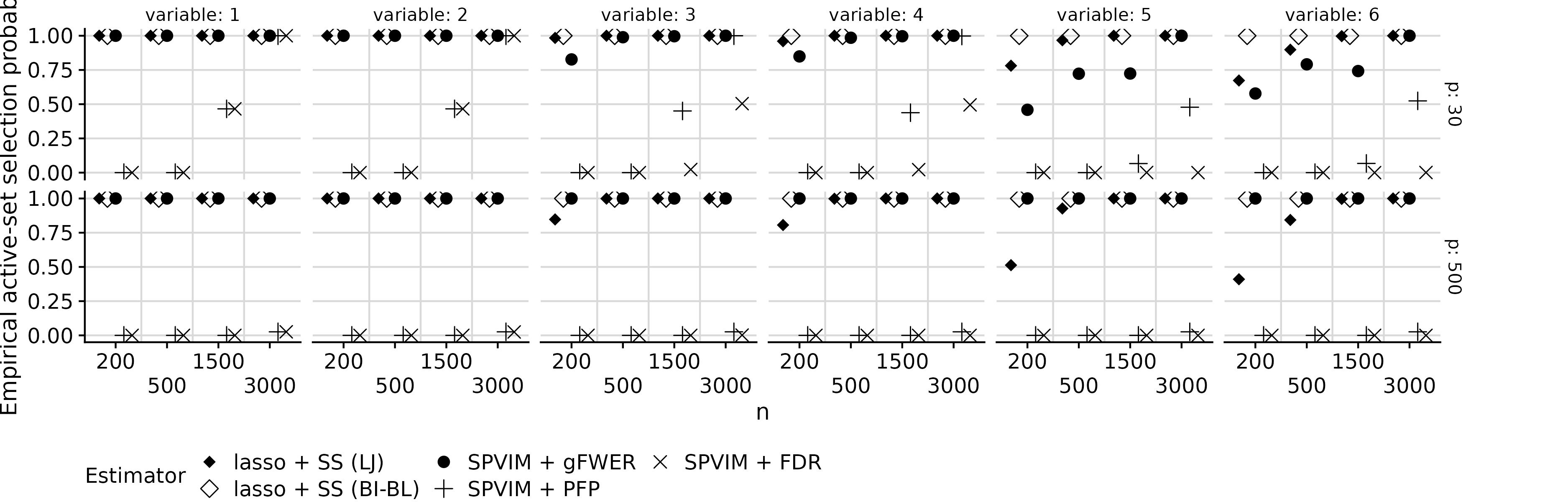}
  \caption{Empirical selection probability for each active-set variable vs $n$ for each estimator and dimension with missing data proportion equal to 0.2, in Scenario 3 (a linear model for the outcome and nonnormal features).}
  \label{fig:scenario-3-probs-20}
\end{figure}

\begin{figure}
  \centering
  \includegraphics[width=1\textwidth]{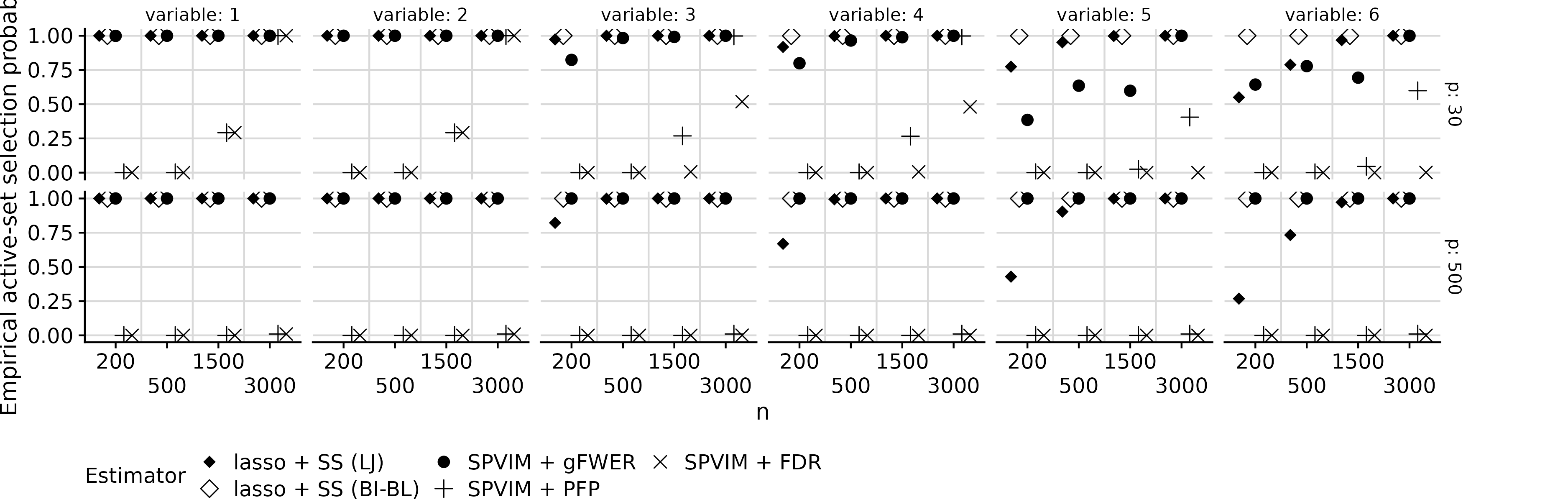}
  \caption{Empirical selection probability for each active-set variable vs $n$ for each estimator and dimension with missing data proportion equal to 0.4, in Scenario 3 (a linear model for the outcome and nonnormal features).}
  \label{fig:scenario-3-probs-40}
\end{figure}

\begin{figure}
  \centering
  \includegraphics[width=1\textwidth]{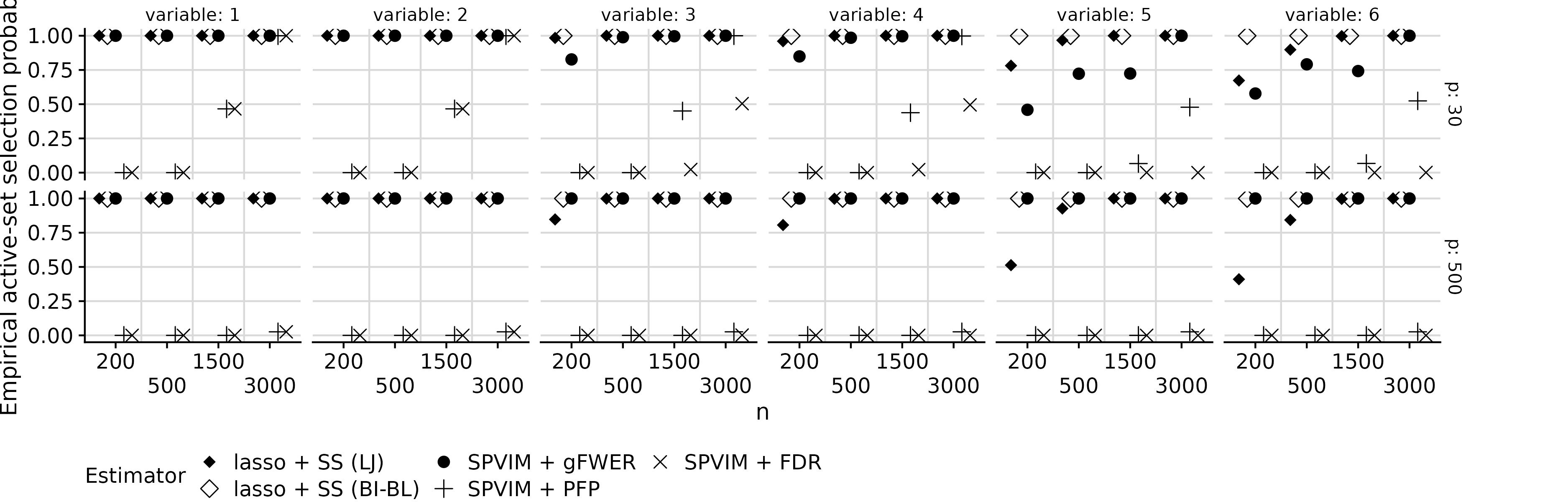}
  \caption{Empirical selection probability for each active-set variable vs $n$ for each estimator and dimension with missing data proportion equal to 0.2, in Scenario 4 (a nonlinear model for the outcome and multivariate normal features).}
  \label{fig:scenario-4-probs-20}
\end{figure}

\begin{figure}
  \centering
  \includegraphics[width=1\textwidth]{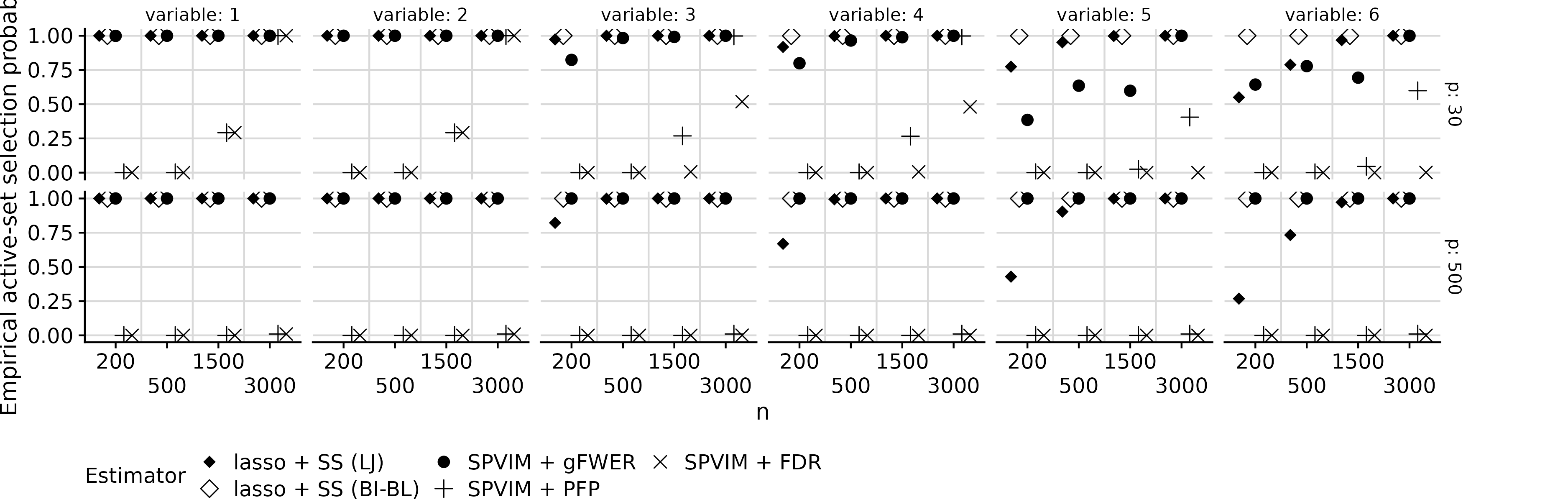}
  \caption{Empirical selection probability for each active-set variable vs $n$ for each estimator and dimension with missing data proportion equal to 0.4, in Scenario 4 (a nonlinear model for the outcome and multivariate normal features).}
  \label{fig:scenario-4-probs-40}
\end{figure}

\begin{figure}
  \centering
  \includegraphics[width=1\textwidth]{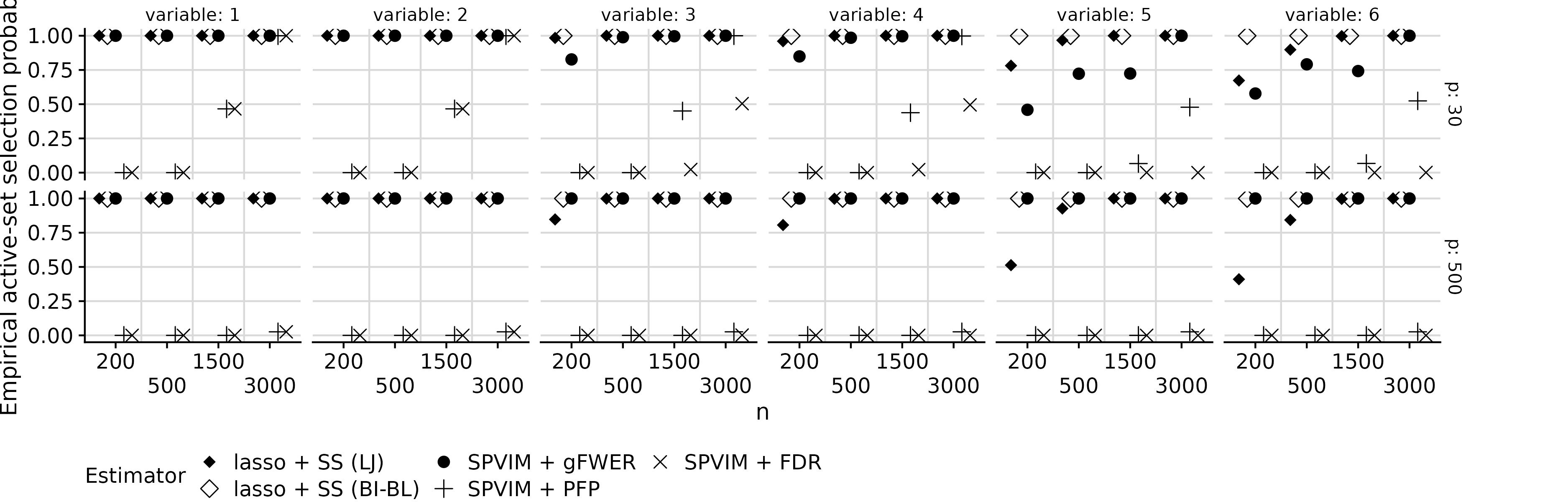}
  \caption{Empirical selection probability for each active-set variable vs $n$ for each estimator and dimension with missing data proportion equal to 0.2, in Scenario 5 (a nonlinear model for the outcome and nonnormal features).}
  \label{fig:scenario-5-probs-20}
\end{figure}

\begin{figure}
  \centering
  \includegraphics[width=1\textwidth]{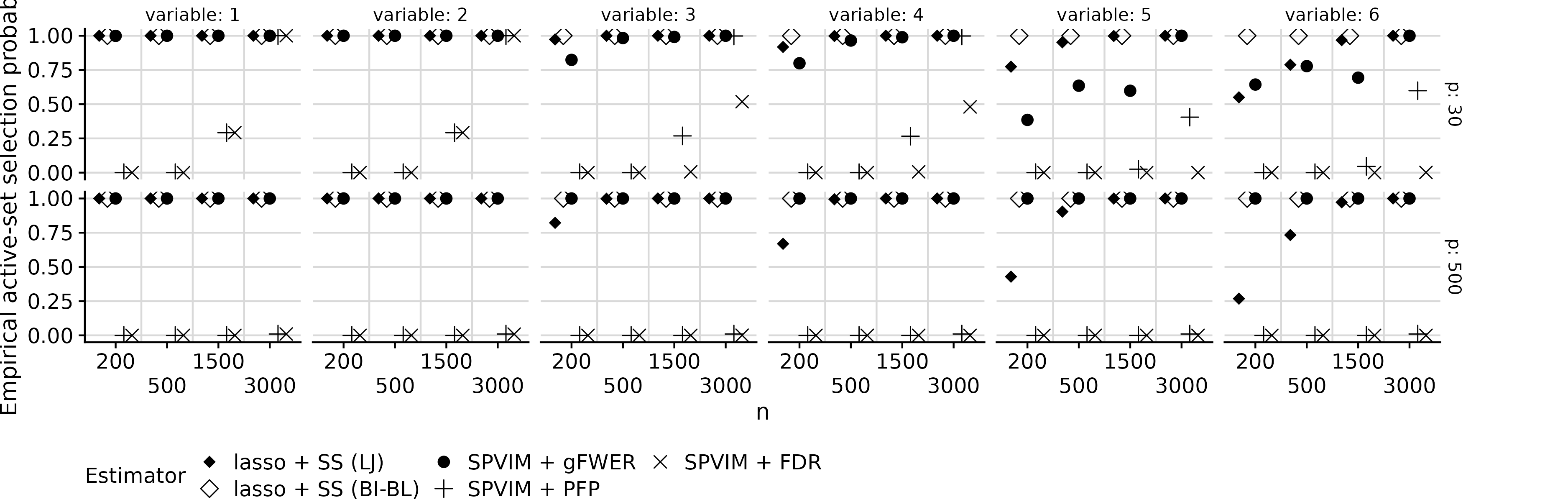}
  \caption{Empirical selection probability for each active-set variable vs $n$ for each estimator and dimension with missing data proportion equal to 0.4, in Scenario 5 (a nonlinear model for the outcome and nonnormal features).}
  \label{fig:scenario-5-probs-40}
\end{figure}

\begin{figure}
  \centering
  \includegraphics[width=1\textwidth]{{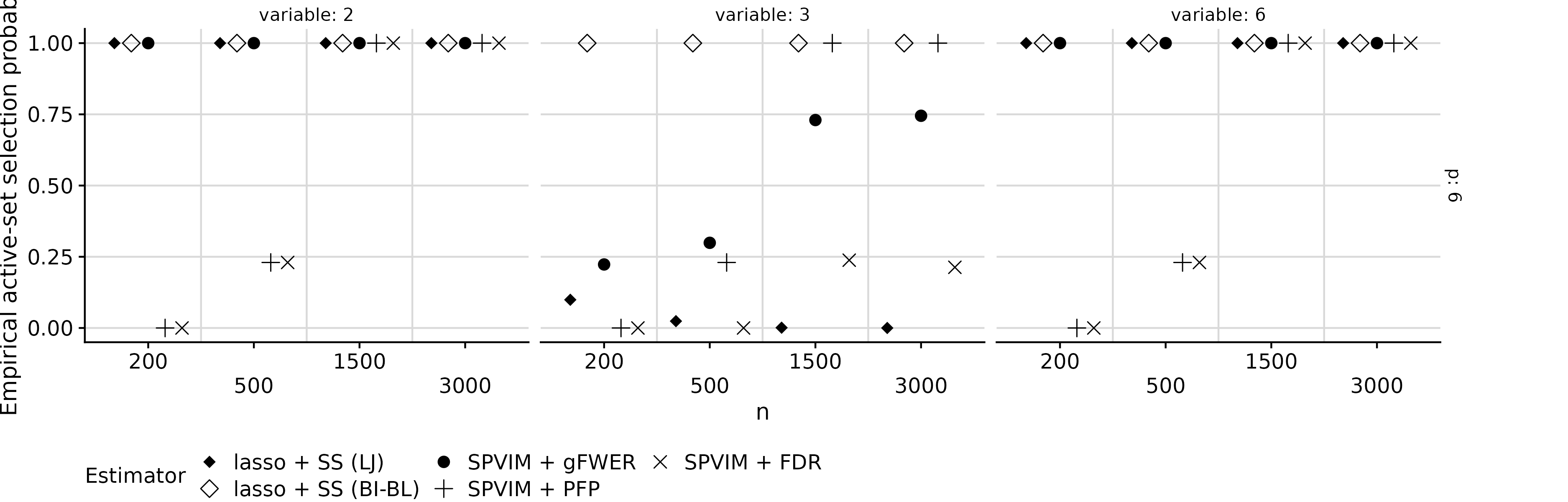}}
  \caption{Empirical selection probability for each active-set variable vs $n$ for each estimator, in Scenario 6 (a weak linear model for the outcome and normal features).}
  \label{fig:scenario-6-probs-20}
\end{figure}

\begin{figure}
  \centering
  \includegraphics[width=1\textwidth]{{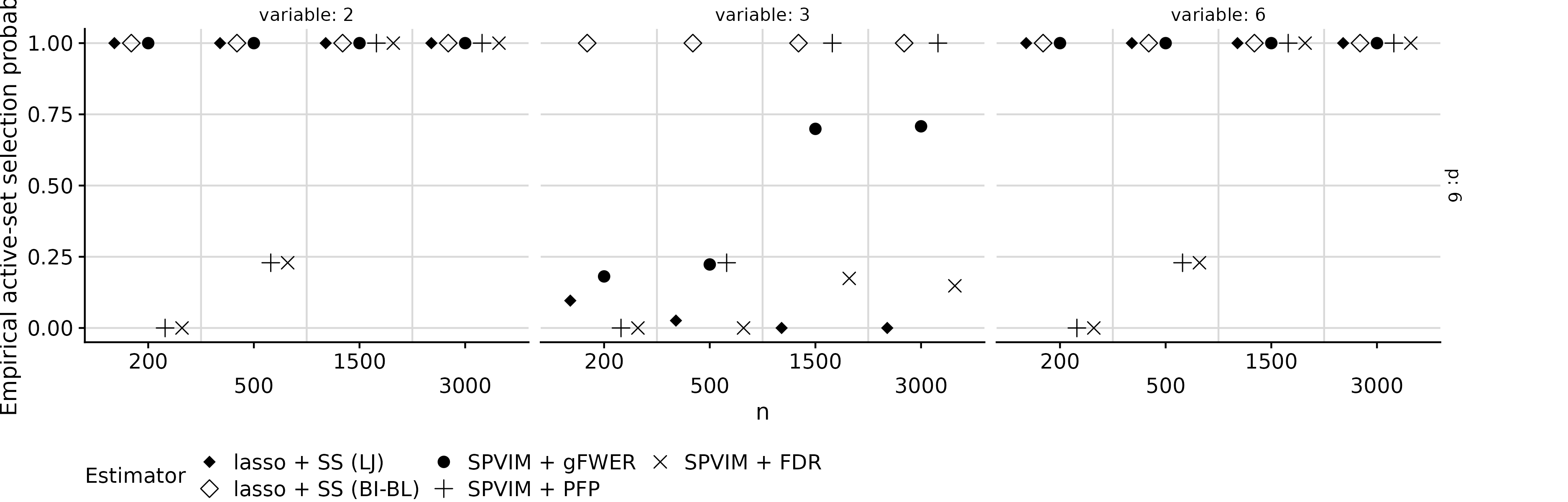}}
  \caption{Empirical selection probability for each active-set variable vs $n$ for each estimator, in Scenario 6 (a weak linear model for the outcome and normal features).}
  \label{fig:scenario-6-probs-40}
\end{figure}

\begin{figure}
  \centering
  \includegraphics[width=1\textwidth]{{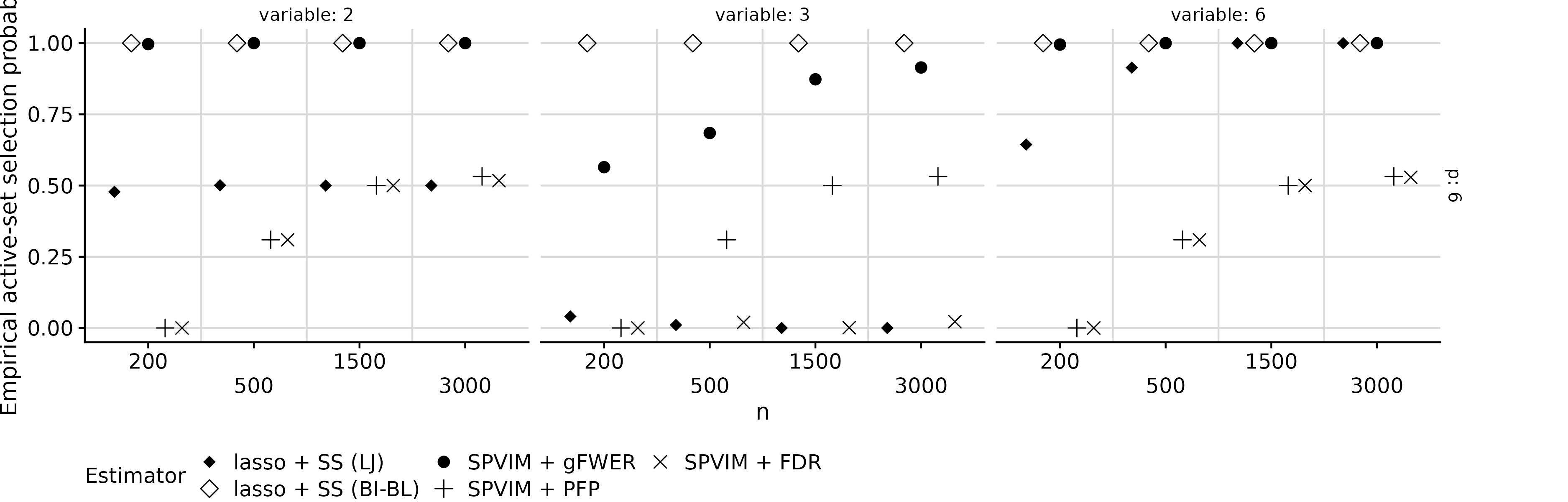}}
  \caption{Empirical selection probability for each active-set variable vs $n$ for each estimator, in Scenario 7 (a weak linear model for the outcome and correlated normal features).}
  \label{fig:scenario-7-probs-20}
\end{figure}

\begin{figure}
  \centering
  \includegraphics[width=1\textwidth]{{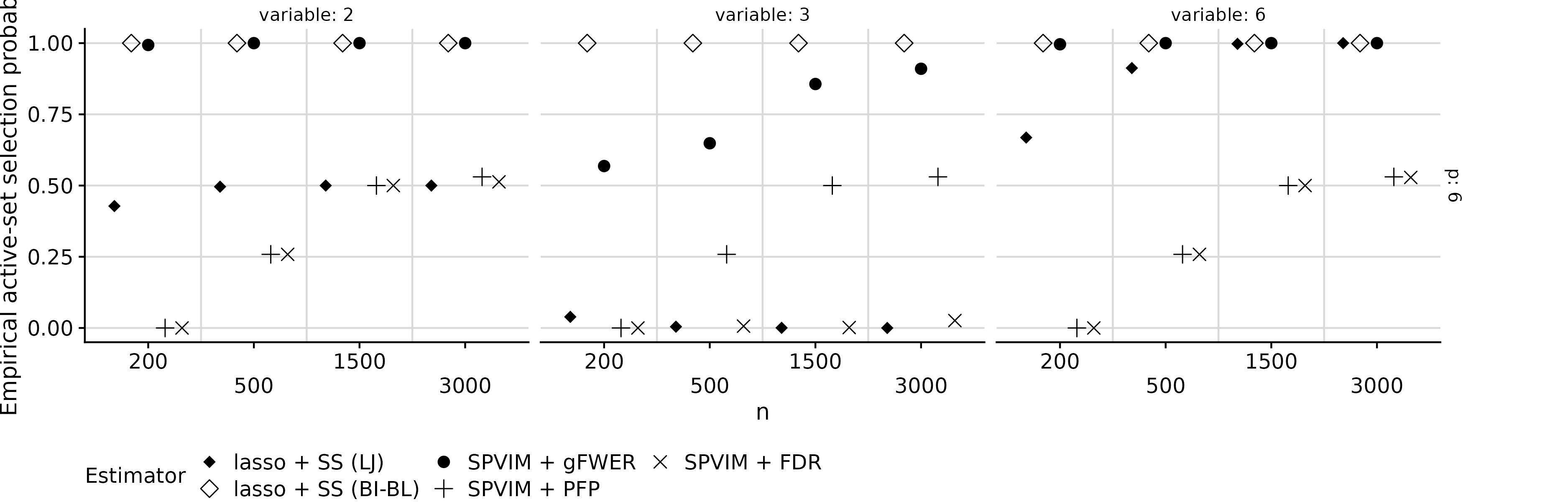}}
  \caption{Empirical selection probability for each active-set variable vs $n$ for each estimator, in Scenario 7 (a weak linear model for the outcome and correlated normal features).}
  \label{fig:scenario-7-probs-40}
\end{figure}

\begin{figure}
  \centering
  \includegraphics[width=1\textwidth]{{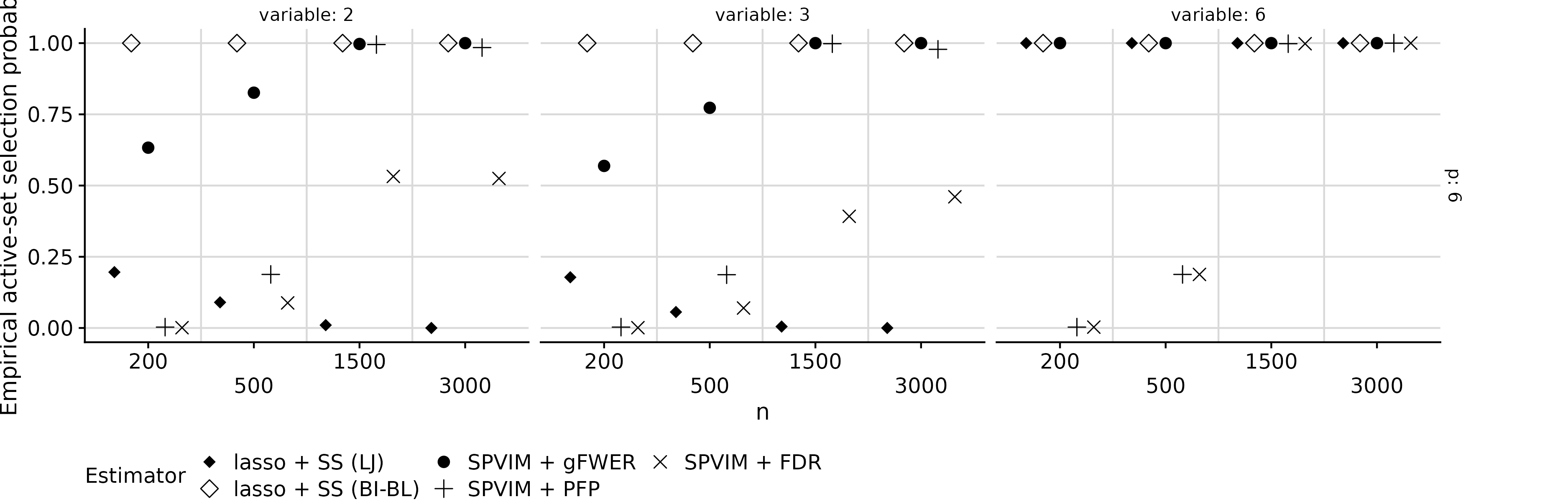}}
  \caption{Empirical selection probability for each active-set variable vs $n$ for each estimator, in Scenario 8 (a weak nonlinear model for the outcome and normal features).}
  \label{fig:scenario-8-probs-20}
\end{figure}

\begin{figure}
  \centering
  \includegraphics[width=1\textwidth]{{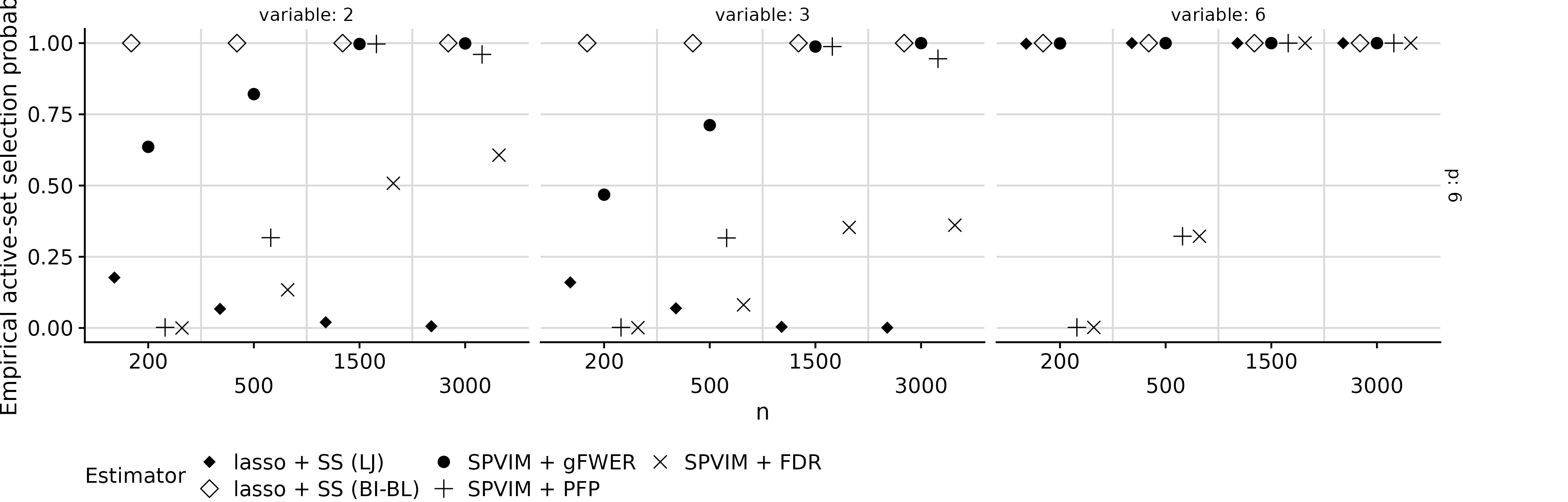}}
  \caption{Empirical selection probability for each active-set variable vs $n$ for each estimator, in Scenario 8 (a weak nonlinear model for the outcome and normal features).}
  \label{fig:scenario-8-probs-40}
\end{figure}

\subsection{Results with completely-observed data}

Here, we consider Scenarios 1--8 with completely-observed data. We compare our intrinsic selection algorithm to the lasso, the lasso with stability selection, and the lasso with knockoffs; these latter three algorithms are often used in variable selection analyses with fully-observed data. In Figures~\ref{fig:scenario-1-select-supp-cc}--\ref{fig:scenario-8-select-supp-cc}, we present the results of these experiments. The results tend to be similar to the results with missing data: when a linear outcome regression model is correctly specified, our intrinsic procedure tends to perform as well as the lasso-based procedures; when the linear outcome regression model is misspecified, our gFWER-controlling procedure tends to perform better than the lasso-based procedures. In settings with more weakly important variables, our intrinsic procedures continue to perform well. We present the proportion of replications where each variable was selected in Figures ~\ref{fig:scenario-1-probs-0}--\ref{fig:scenario-8-probs-0}, again observing similar trends to the missing-data cases.

\begin{figure}
  \centering
  \includegraphics[width=1\textwidth]{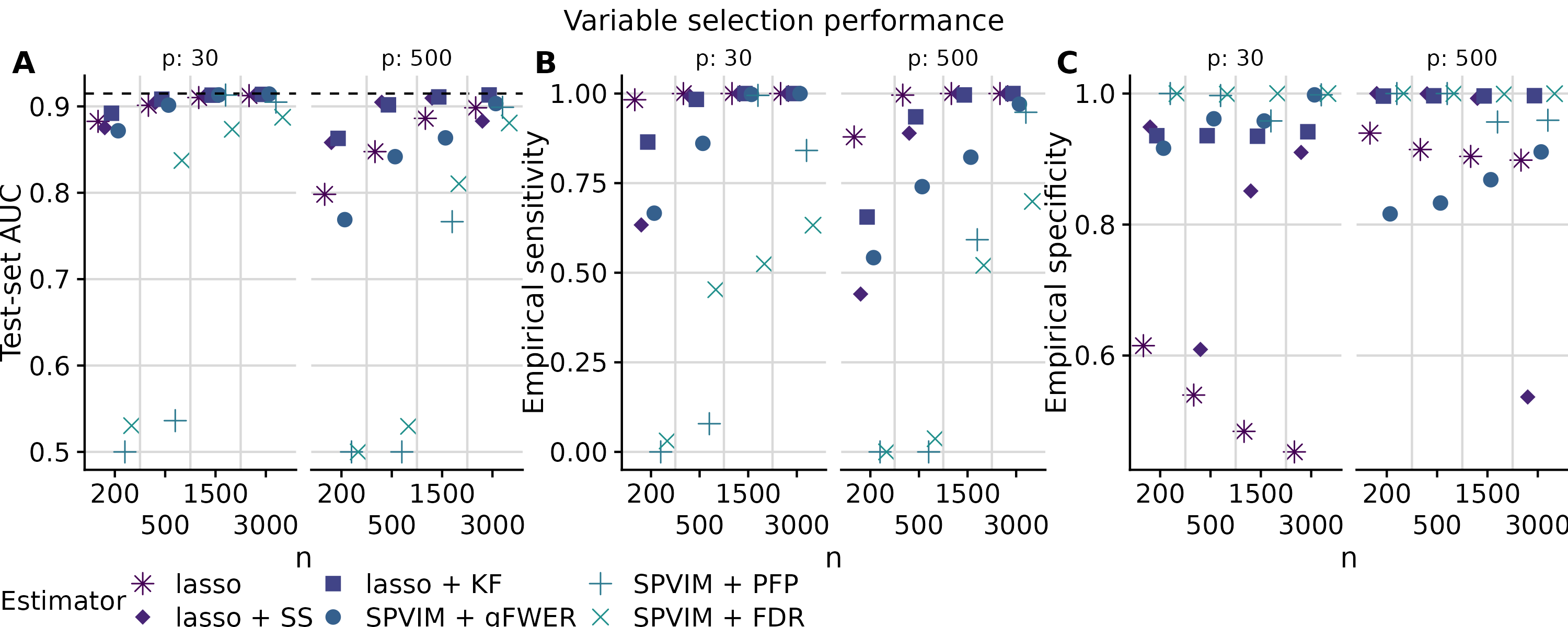}
  \caption{Test-set AUC (panel A) and empirical variable selection sensitivity (panel B) and specificity (panel C) vs $n$ for each estimator and missing data proportion equal to 0, in Scenario 1 (a linear model for the outcome and multivariate normal features). The dotted line in panel A shows the true (optimal) test-set AUC.}
  \label{fig:scenario-1-select-supp-cc}
\end{figure}

\begin{figure}
  \centering
  \includegraphics[width=1\textwidth]{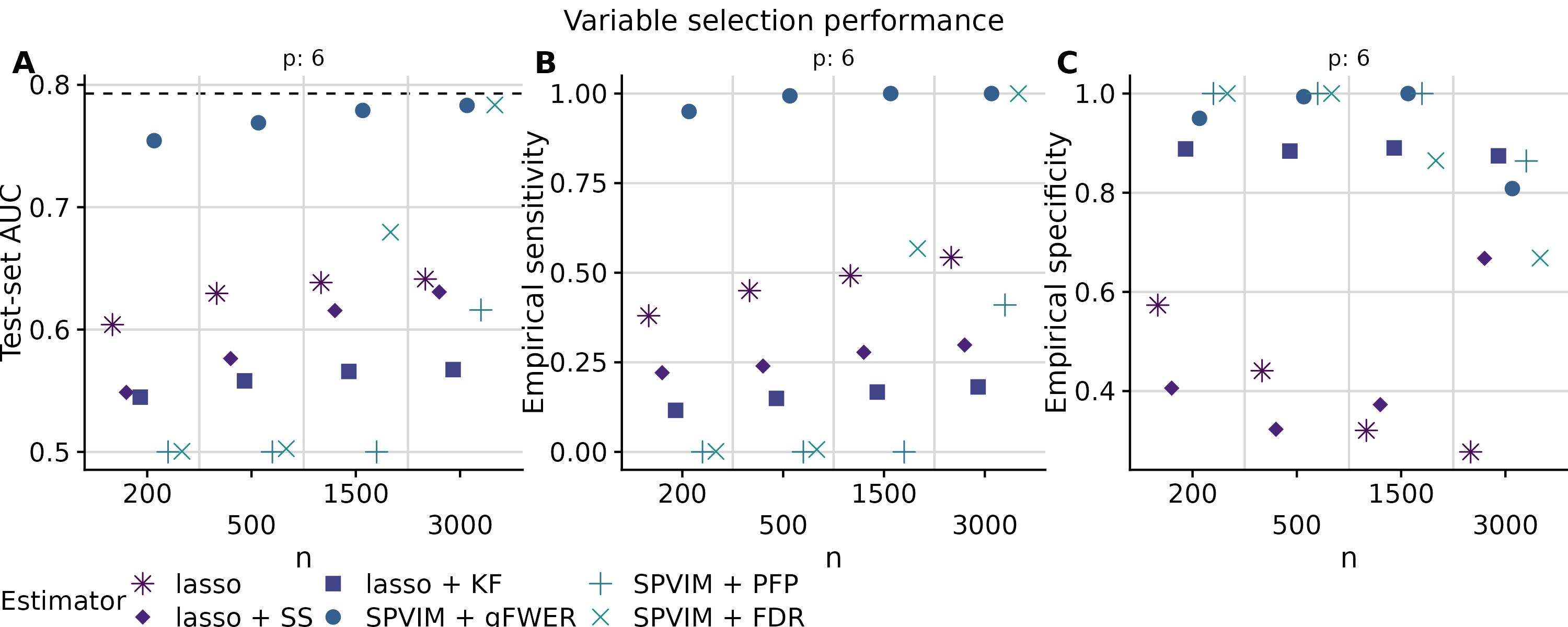}
  \caption{Test-set AUC (panel A) and empirical variable selection sensitivity (panel B) and specificity (panel C) vs $n$ for each estimator and missing data proportion equal to 0, in Scenario 2 (a nonlinear model for the outcome and correlated multivariate normal features), when the data are completely observed. The dotted line in panel A shows the true (optimal) test-set AUC.}
  \label{fig:scenario-2-select-supp-cc}
\end{figure}

\begin{figure}
  \centering
  \includegraphics[width=1\textwidth]{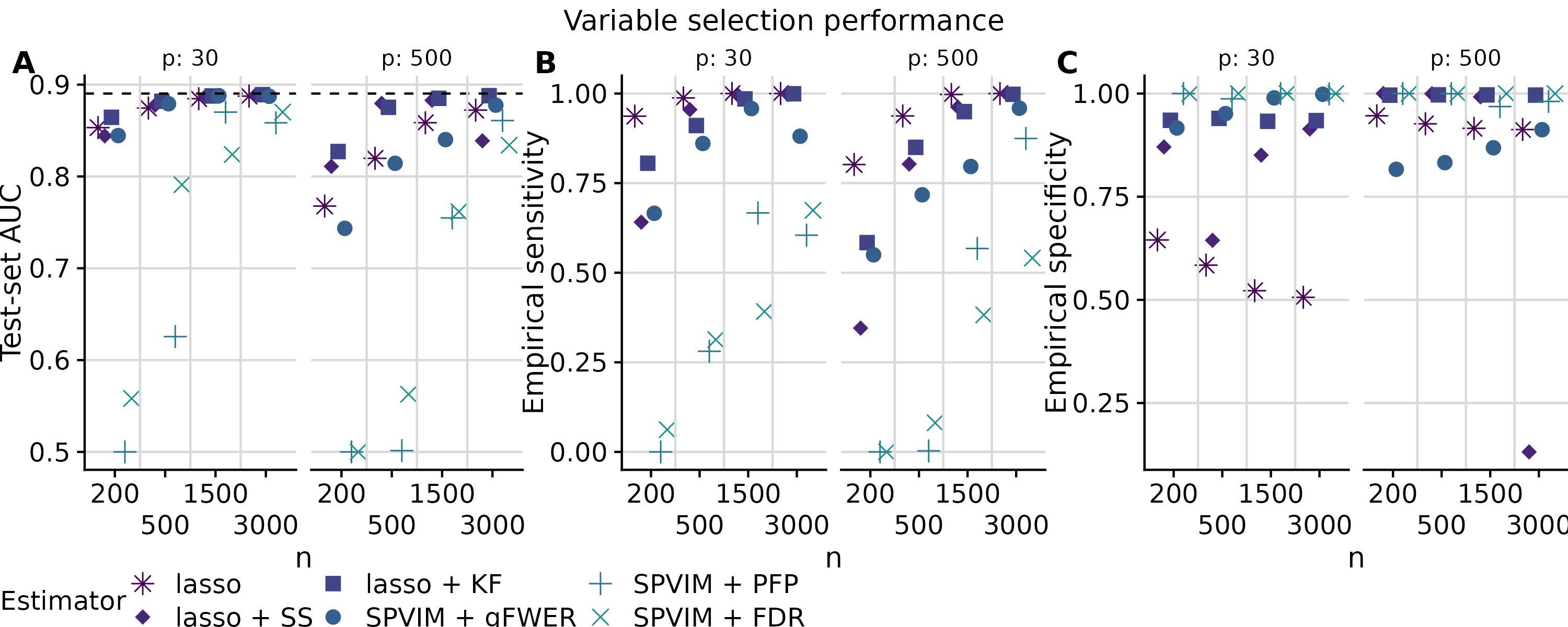}
  \caption{Test-set AUC (panel A) and empirical variable selection sensitivity (panel B) and specificity (panel C) vs $n$ for each estimator and missing data proportion, in Scenario 3 (a linear model for the outcome and nonnormal features), when the data are completely observed. The dotted line in panel A shows the true (optimal) test-set AUC.}
  \label{fig:scenario-3-select-supp-cc}
\end{figure}

\begin{figure}
  \centering
  \includegraphics[width=1\textwidth]{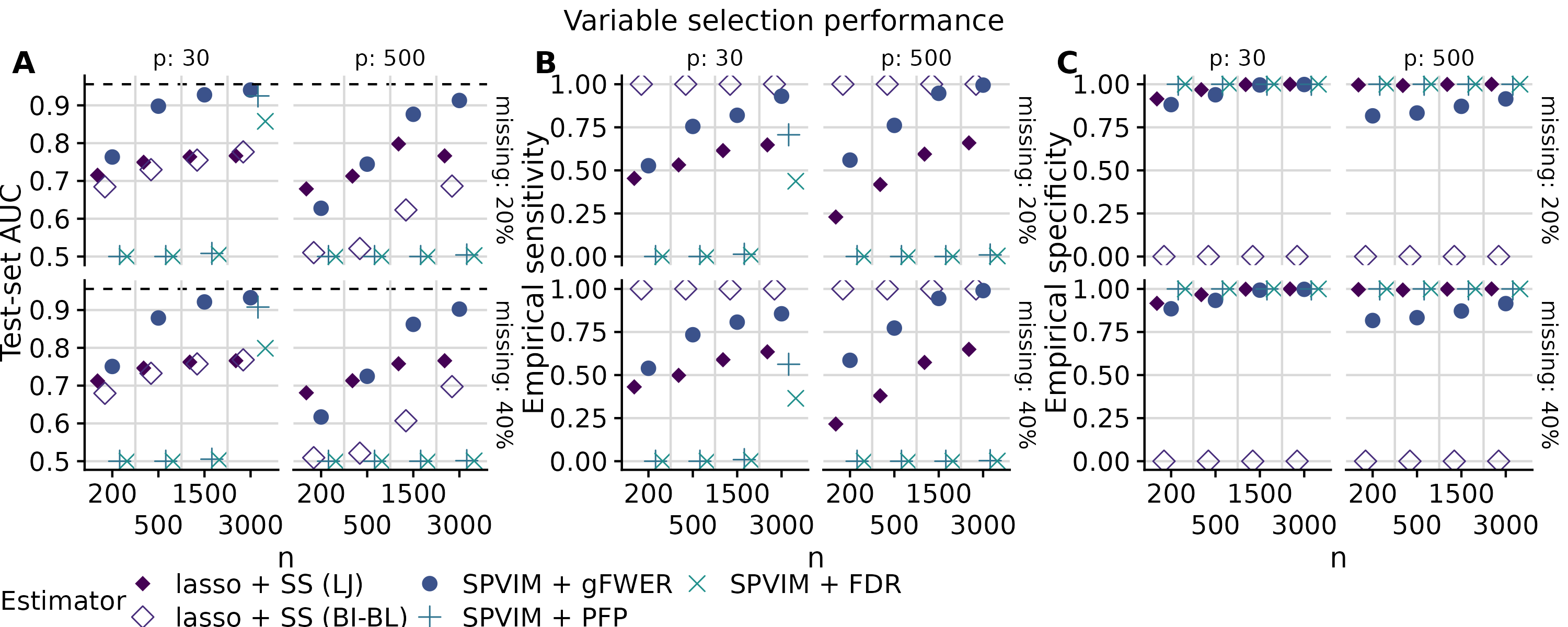}
  \caption{Test-set AUC (panel A) and empirical variable selection sensitivity (panel B) and specificity (panel C) vs $n$ for each estimator and missing data proportion, in Scenario 4 (a nonlinear model for the outcome and normal features), when the data are completely observed. The dotted line in panel A shows the true (optimal) test-set AUC.}
  \label{fig:scenario-4-select-supp-cc}
\end{figure}

\begin{figure}
  \centering
  \includegraphics[width=1\textwidth]{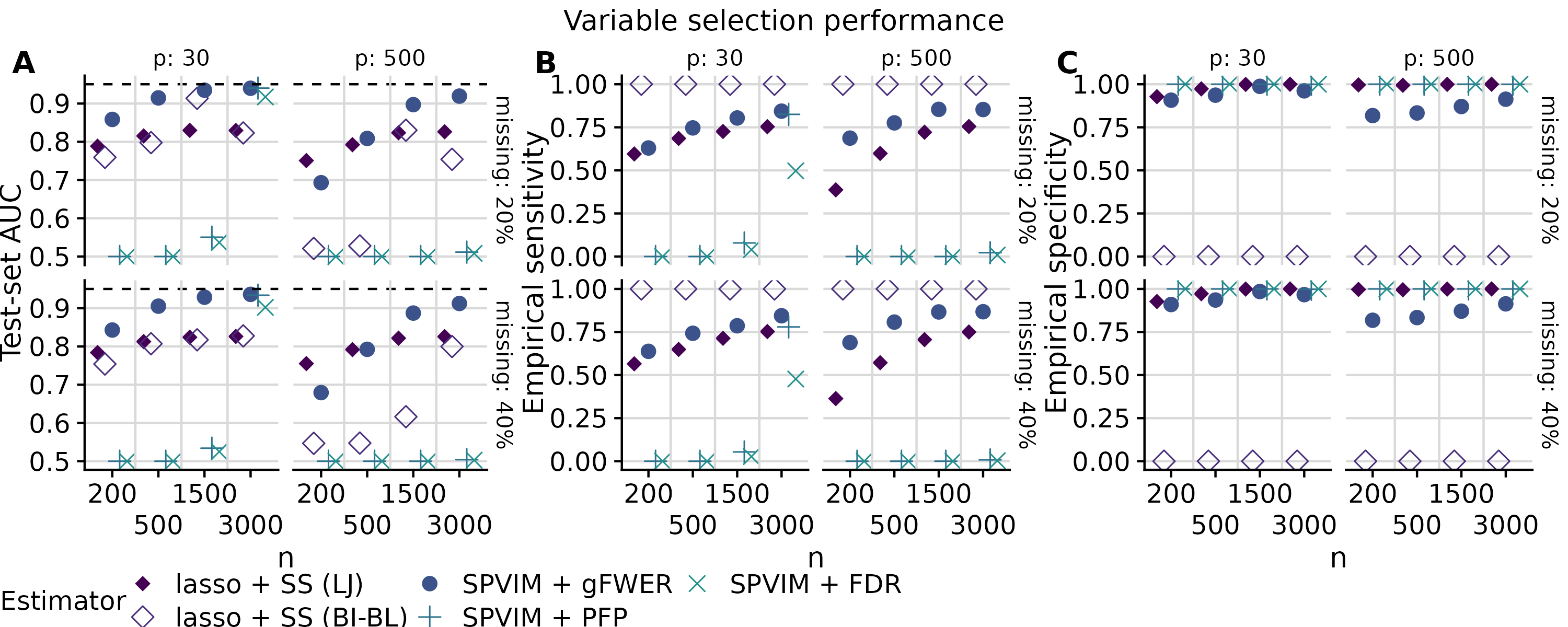}
  \caption{Test-set AUC (panel A) and empirical variable selection sensitivity (panel B) and specificity (panel C) vs $n$ for each estimator and missing data proportion, in Scenario 5 (a nonlinear model for the outcome and nonnormal features), when the data are completely observed. The dotted line in panel A shows the true (optimal) test-set AUC.}
  \label{fig:scenario-5-select-supp-cc}
\end{figure}

\begin{figure}
  \centering
  \includegraphics[width=1\textwidth]{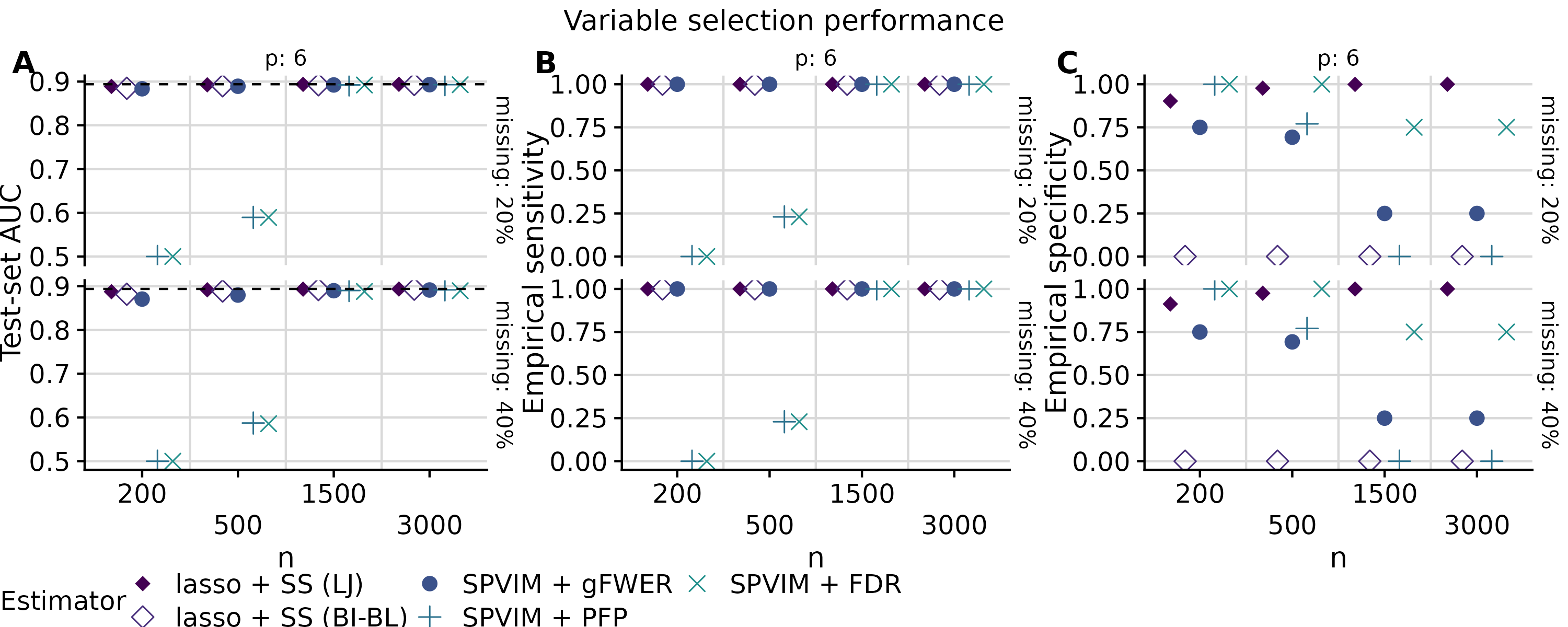}
  \caption{Test-set AUC (panel A) and empirical variable selection sensitivity (panel B) and specificity (panel C) vs $n$ for each estimator and missing data proportion, in Scenario 6 (a weak linear model for the outcome and normal features), when the data are completely observed. The dotted line in panel A shows the true (optimal) test-set AUC.}
  \label{fig:scenario-6-select-supp-cc}
\end{figure}

\begin{figure}
  \centering
  \includegraphics[width=1\textwidth]{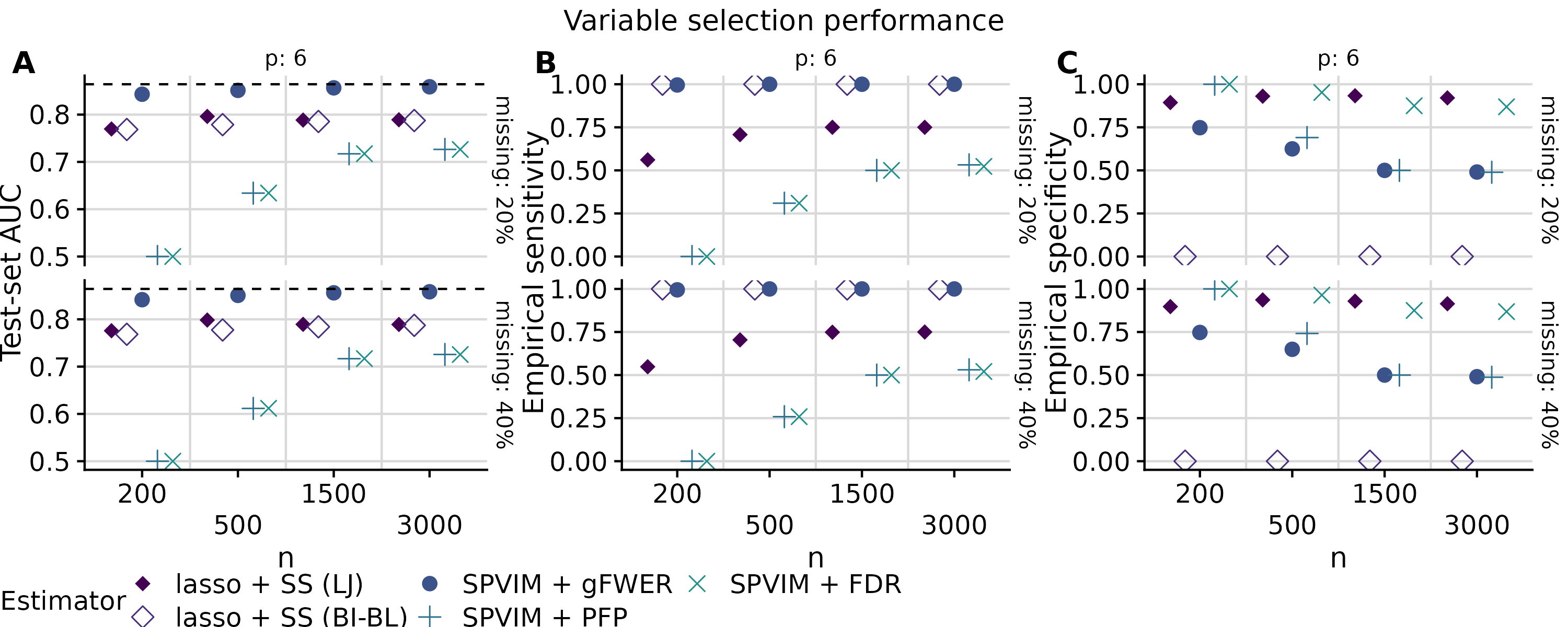}
  \caption{Test-set AUC (panel A) and empirical variable selection sensitivity (panel B) and specificity (panel C) vs $n$ for each estimator and missing data proportion, in Scenario 7 (a weak nonlinear model for the outcome and correlated normal features), when the data are completely observed. The dotted line in panel A shows the true (optimal) test-set AUC.}
  \label{fig:scenario-7-select-supp-cc}
\end{figure}

\begin{figure}
  \centering
  \includegraphics[width=1\textwidth]{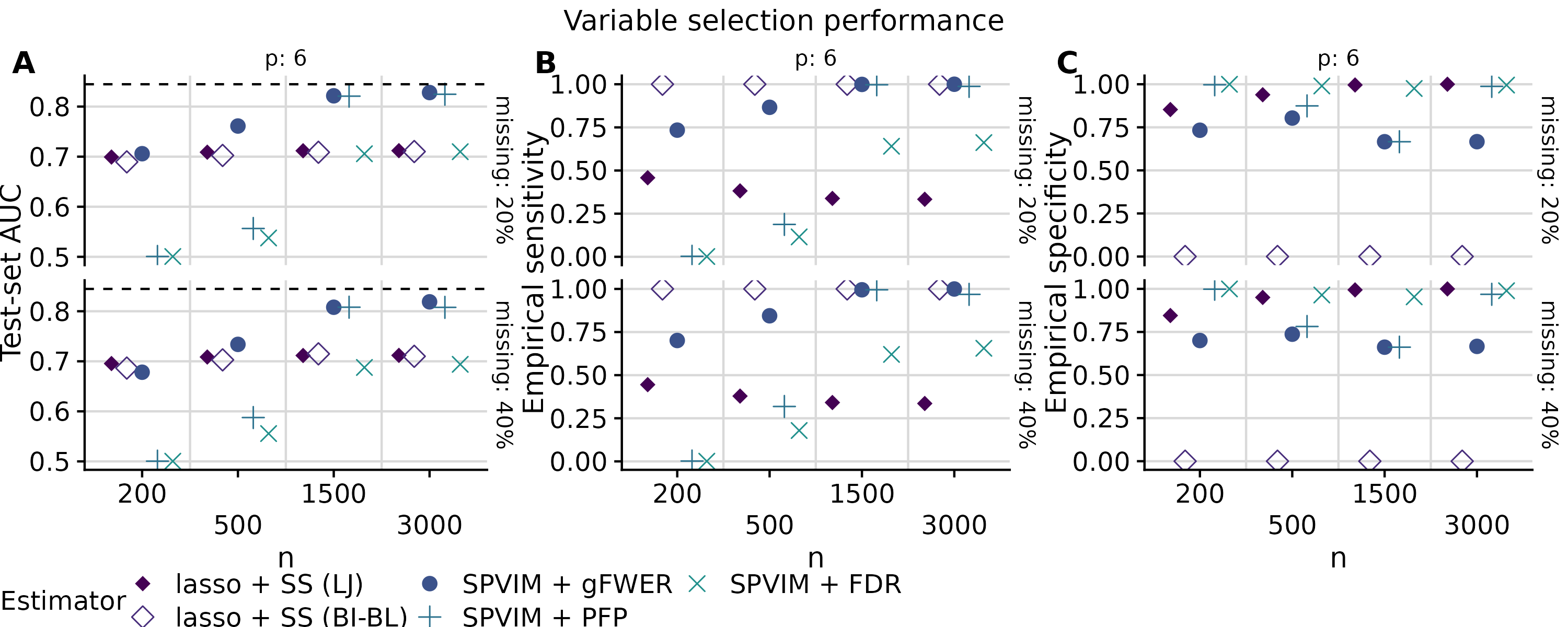}
  \caption{Test-set AUC (panel A) and empirical variable selection sensitivity (panel B) and specificity (panel C) vs $n$ for each estimator and missing data proportion, in Scenario 8 (a weak nonlinear model for the outcome and normal features), when the data are completely observed. The dotted line in panel A shows the true (optimal) test-set AUC.}
  \label{fig:scenario-8-select-supp-cc}
\end{figure}

\begin{figure}
  \centering
  \includegraphics[width=1\textwidth]{{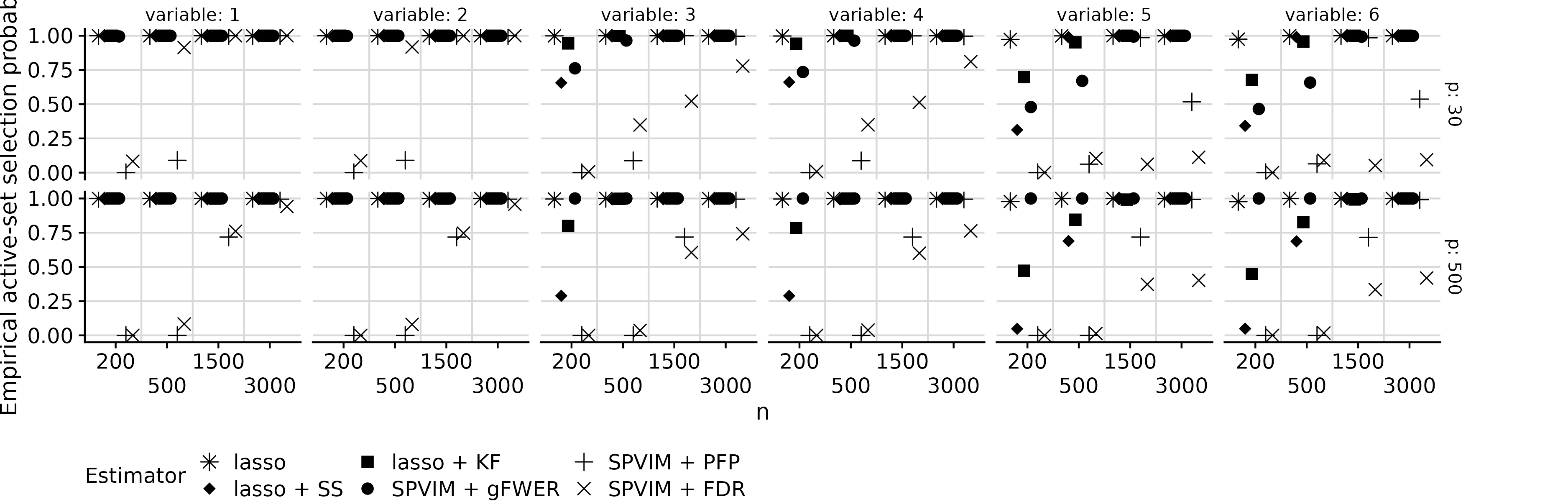}}
  \caption{Empirical selection probability for each active-set variable vs $n$ for each estimator and dimension with missing data proportion equal to 0, in Scenario 1 (a linear model for the outcome and multivariate normal features), when the data are completely observed.}
  \label{fig:scenario-1-probs-0}
\end{figure}

\begin{figure}
  \centering
  \includegraphics[width=1\textwidth]{{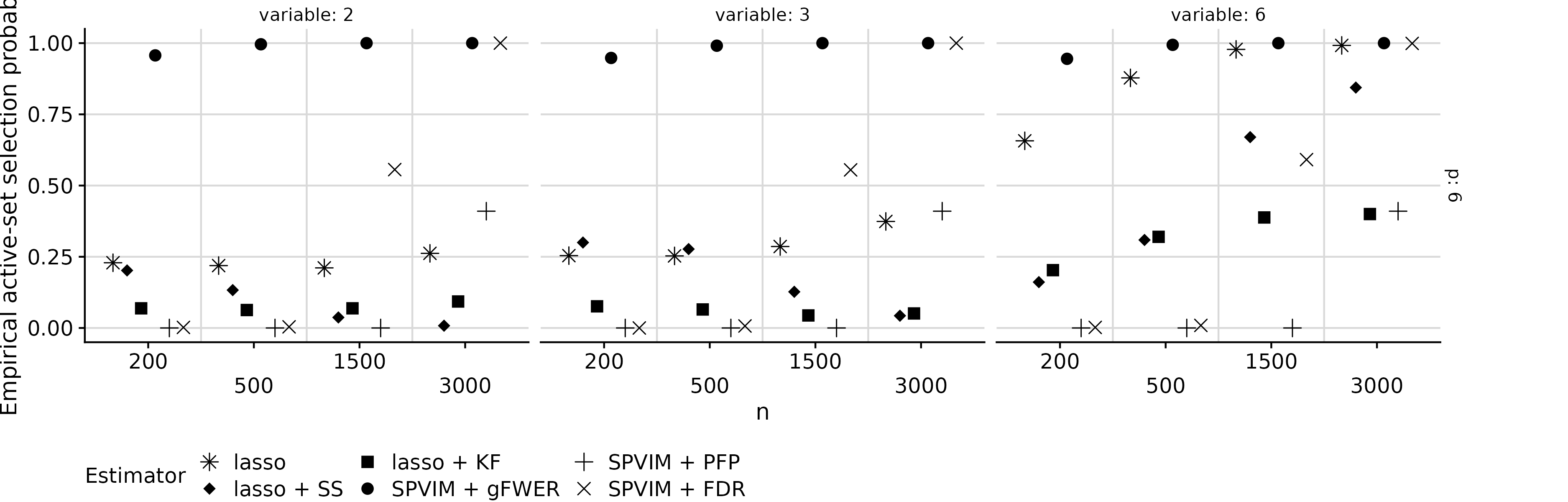}}
  \caption{Empirical selection probability for each active-set variable vs $n$ for each estimator, in Scenario 2 (a nonlinear model for the outcome and correlated multivariate normal features), when the data are completely observed.}
  \label{fig:scenario-2-probs-0}
\end{figure}

\begin{figure}
  \centering
  \includegraphics[width=1\textwidth]{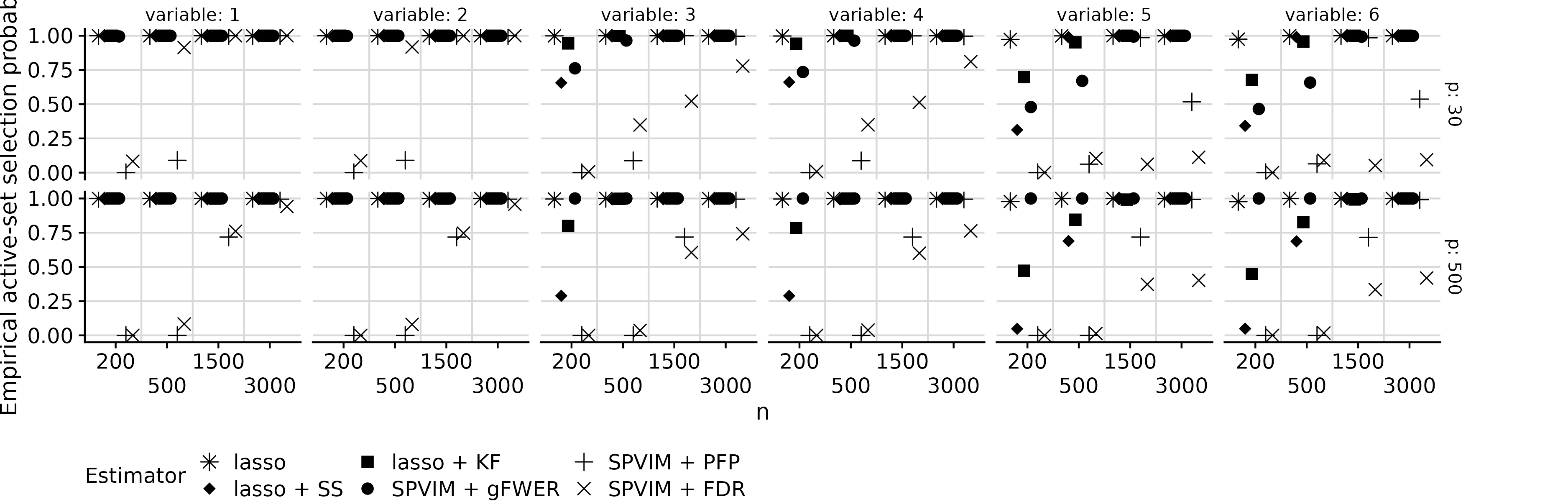}
  \caption{Empirical selection probability for each active-set variable vs $n$ for each estimator and dimension with missing data proportion equal to 0, in Scenario 3 (a linear model for the outcome and nonnormal features), when the data are completely observed.}
  \label{fig:scenario-3-probs-0}
\end{figure}

\begin{figure}
  \centering
  \includegraphics[width=1\textwidth]{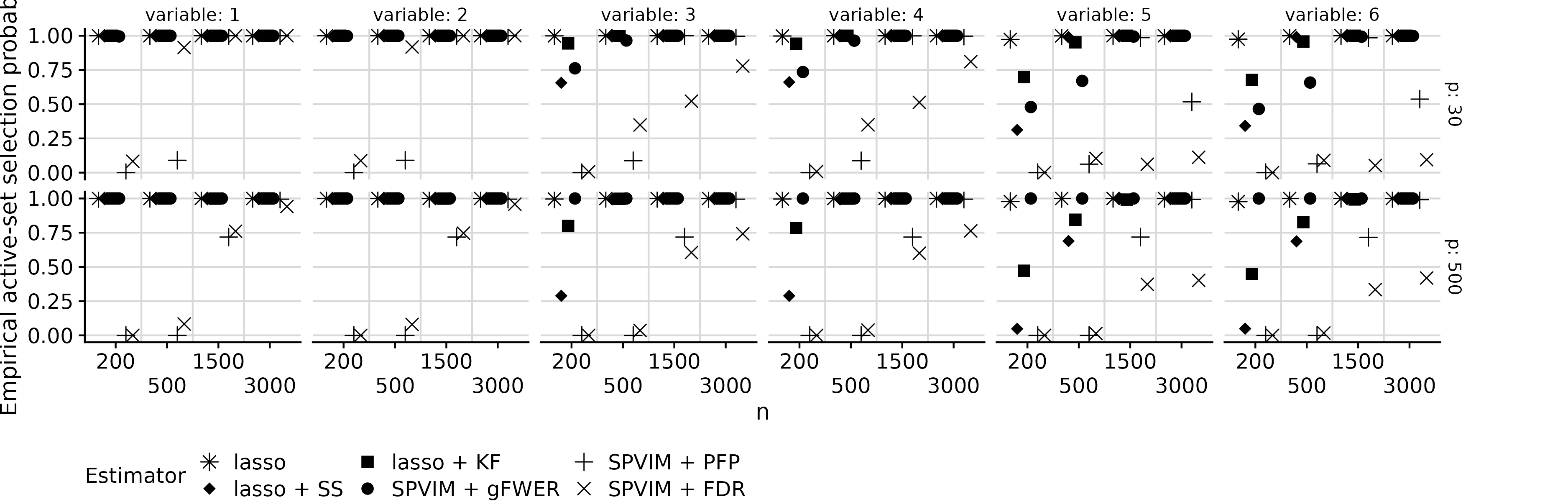}
  \caption{Empirical selection probability for each active-set variable vs $n$ for each estimator and dimension with missing data proportion equal to 0, in Scenario 4 (a nonlinear model for the outcome and multivariate normal features), when the data are completely observed.}
  \label{fig:scenario-4-probs-0}
\end{figure}

\begin{figure}
  \centering
  \includegraphics[width=1\textwidth]{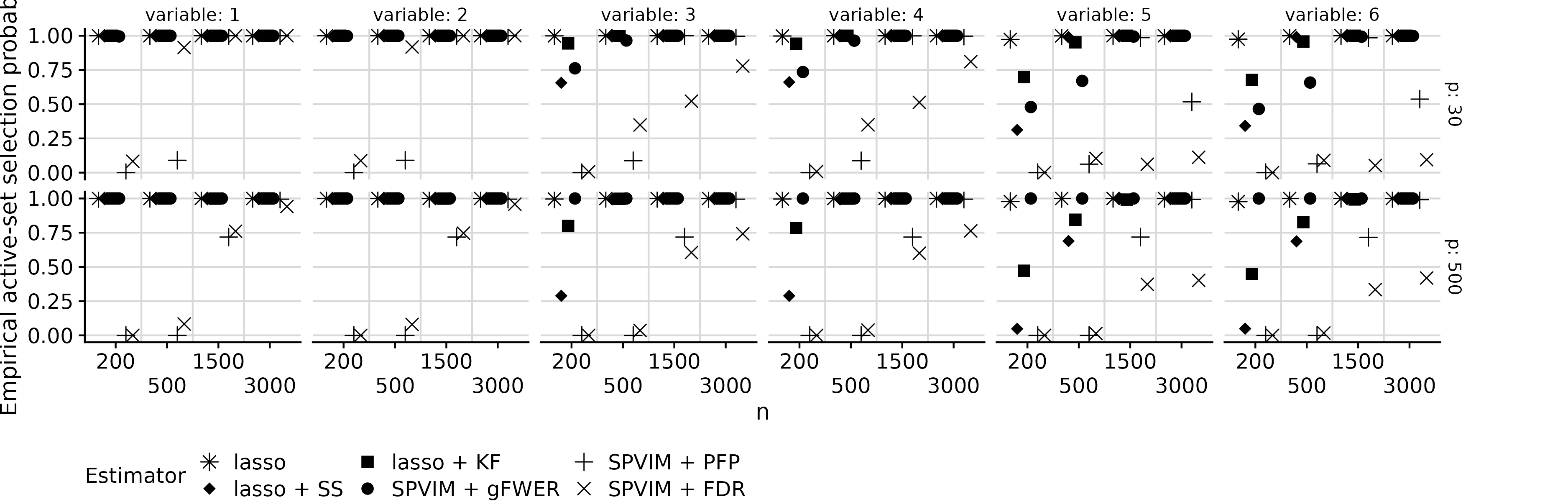}
  \caption{Empirical selection probability for each active-set variable vs $n$ for each estimator and dimension with missing data proportion equal to 0, in Scenario 5 (a nonlinear model for the outcome and nonnormal features), when the data are completely observed.}
  \label{fig:scenario-5-probs-0}
\end{figure}

\begin{figure}
  \centering
  \includegraphics[width=1\textwidth]{{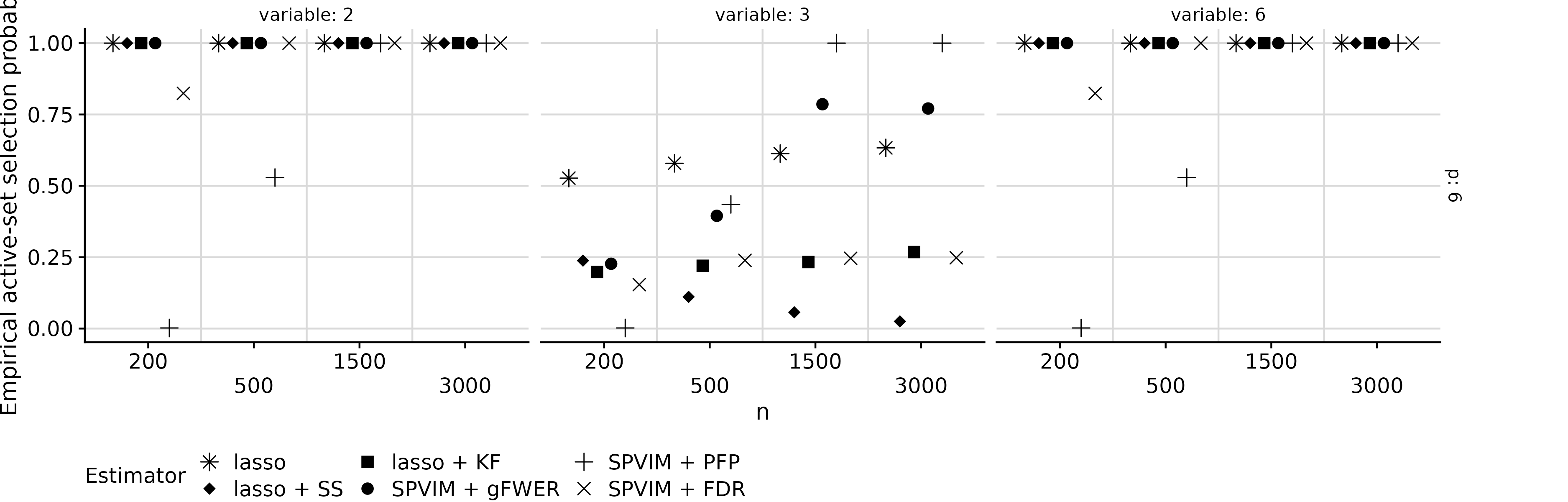}}
  \caption{Empirical selection probability for each active-set variable vs $n$ for each estimator, in Scenario 6 (a weak linear model for the outcome and normal features), when the data are completely observed.}
  \label{fig:scenario-6-probs-0}
\end{figure}

\begin{figure}
  \centering
  \includegraphics[width=1\textwidth]{{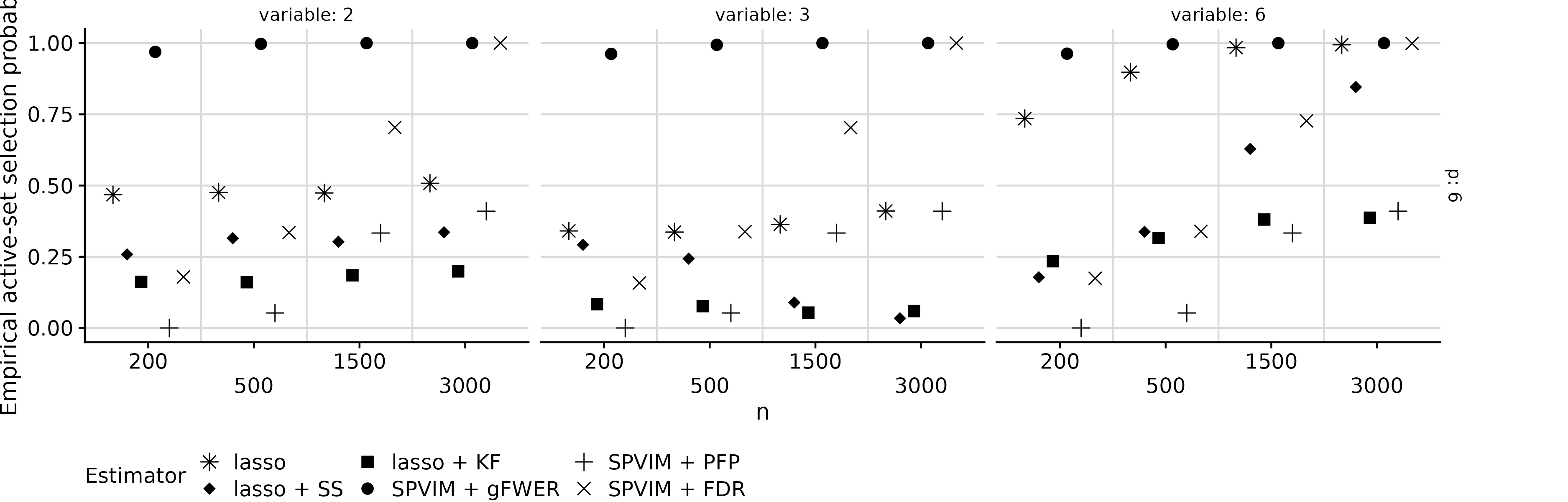}}
  \caption{Empirical selection probability for each active-set variable vs $n$ for each estimator, in Scenario 7 (a weak linear model for the outcome and correlated normal features), when the data are completely observed.}
  \label{fig:scenario-7-probs-0}
\end{figure}

\begin{figure}
  \centering
  \includegraphics[width=1\textwidth]{{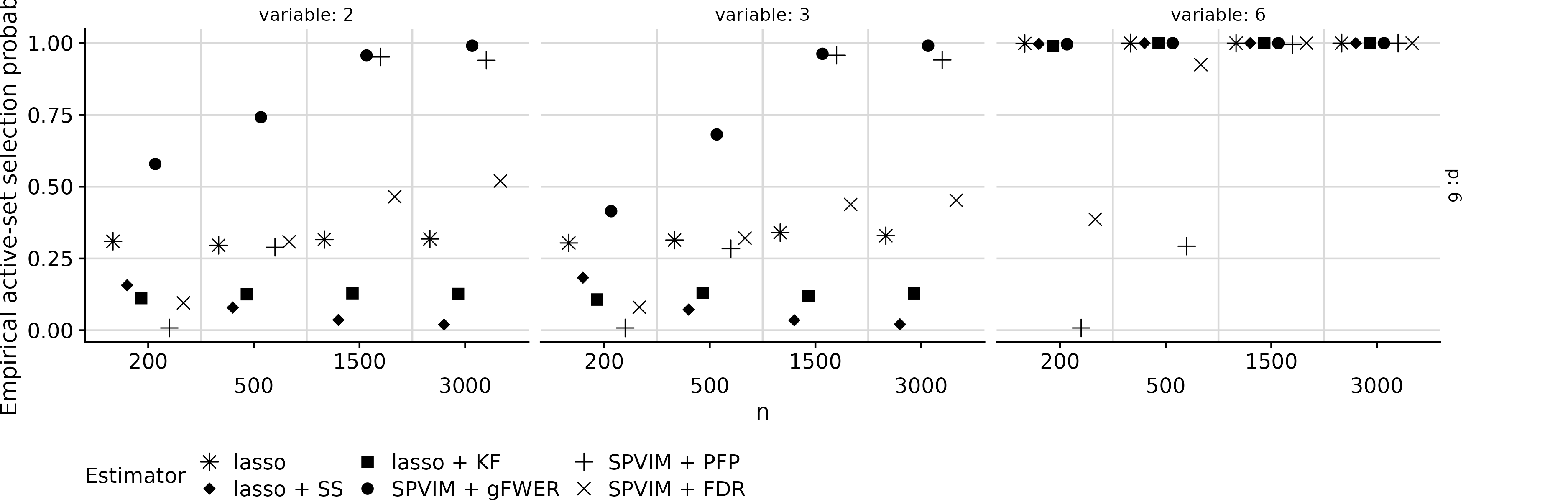}}
  \caption{Empirical selection probability for each active-set variable vs $n$ for each estimator, in Scenario 8 (a weak nonlinear model for the outcome and normal features), when the data are completely observed.}
  \label{fig:scenario-8-probs-0}
\end{figure}

\subsection{Summary of results from Scenarios 1--8}

Taken together, these results suggest that (a) as the missing data proportion increases, performance of all procedures tends to degrade; (b) the outcome distribution (linear vs nonlinear) appears to have a larger effect on test-set AUC than the covariate distribution (normal vs nonnormal); (c) weakly important variables are less likely to be selected by lasso-based procedures than strongly important variables; and (d) correlation causes further degradation in performance for lasso-based methods. Variable selection performance (sensitivity and specificity) is similar asymptotically across Scenarios 1 and 3--5. This last finding is surprising, since the variable selection performance of the lasso is not guaranteed in misspecified settings. However, as we saw in Scenarios 2, 7, and 8, in adversarial cases the lasso-based estimators can have poor variable selection performance, as suggested by theory. Additionally, in the plots describing empirical selection probability for lasso-based estimators, we saw that while lasso-based procedures may have good overall selection performance, some important variables may still be missed, even in the non-adversarial settings. In contrast, our intrinsic variable selection procedure is more robust to model misspecification. Finally, we saw that our proposal performs comparably to commonly used variable selection procedures in settings both with and without missing data when lasso-based estimators are correctly specified.

\section{Additional details for the pancreatic cancer analysis}

We had two overall objectives:
\begin{enumerate}
  \item separate mucinous cysts from non-mucinous cysts, where a mucinous cyst is thought to have some malignant potential; and
  \item separate cysts with high maglinant potential from cysts with low or no malignant potential.
\end{enumerate}
To meet these objectives, we want to assess both individual biomarkers and panels of biomarkers, both using continuous markers and binary calls.

\subsection{Data preprocessing}

To create analysis data from the raw data, we selected the following variables: participant ID, institution, the entire set of continuous biomarkers and binary calls (listed in Table~\ref{tab:biomarker_definitions}). The proportion of missing data in the biomarkers ranged from a minimum of 24.5\% to a maximum of 68.3\%; the median proportion of missing data was 31\%.

\begin{table}
    \caption{All biomarkers of interest for the pancreatic cancer analysis.}
    \label{tab:biomarker_definitions}
    \begin{tabular}{l|l}
        Biomarker & Description \\
        \hline
        CEA & Carcinoembryonic antigen. Serum levels may be elevated in some types \\
        & of cancer (e.g., colorectal cancer, pancreatic cancer). \\
        CEA mucinous call & Binary indicator of whether $\text{CEA} > 192$. \\
        ACTB & Actin Beta \citep{hata2017}\\
        Molecules score & Methylated DNA levels of selected genes \citep{hata2017}\\
        Molecules neoplasia call & Binary indicator of whether molecules score $> 25$\\
        Telomerase score & Telomerase activity measured using \\
        &  telomere repeat amplification protocol \citep{hata2016}\\
        Telomerase neoplasia call & Binary indicator of whether telomerase score $> 730$\\
        AREG score & Amphiregulin (AREG) overexpression \citep{tun2012}\\
        AREG mucinous call & Binary indicator of whether AREG score $> 112$\\
        Glucose score & Glucometer glucose level \citep{zikos2015}\\
        Glucose mucinous call & Binary indicator of whether glucose score $< 50$\\
        Combined mucinuous call & Binary indicator of whether AREG score $> 112$ and \\
        & glucose score $< 50$\\
        Fluorescence score & Fluorescent protease activity \citep{ivry2017}\\
        Fluorescence mucinuous call & Binary indicator of whether fluorescence score $> 1.23$\\
        DNA mucinous call & Presence of mutations in a DNA sequencing panel \citep{singhi2018}\\
        DNA neoplasia call (v1) & Binary indicator of methylated DNA levels of selected genes being \\
        & above a threshold \citep{majumder2019}\\
        DNA neoplasia call (v2) & Binary indicator of methylated DNA levels of selected genes being \\
        & above a threshold \citep{majumder2019}\\
        MUC3AC score & Expression of protein Mucin 3AC\\
        MUC5AC score & Expression of protein Mucin 5AC \citep{cao2013}\\
        Ab score & Monoclonal antibody reactivity \citep{das2014}\\
        Ab neoplasia call & Binary indicator of whether Ab score $> 0.104$\\
    \end{tabular}
\end{table}

\subsection{Imputing missing data}

Our analyses are all based on multiple imputation via chained equations \citep[MICE, implemented in the R package \texttt{mice};][]{vanbuuren2007,vanbuuren2010}. For $i = 1, \ldots, n$ and $j = 1, \ldots, r$ (where $n = 321$ is the sample size and $r = 21$ denotes the total number of biomarkers), we denote the $i$th measurement of biomarker $j$ by $X_{ij}$ and the outcome of interest by $Y_i$. We used the following model to impute missing biomarker values:
\begin{align*}
  X_{i, j, \text{mis}} &\sim Y_i + X_{i, j, \text{obs}} + \text{Institution}_i.
\end{align*}
These models allow us to relate observed biomarker values (and the institution at which each specimen was collected) to the unobserved biomarker values. All imputations were performed using a maximum of 20 iterations and predictive mean matching \citep[PMM;][]{vanbuuren2010} to create 10 fully-imputed datasets. In some cases, the PMM algorithm failed to converge; in these cases, we used tree-based imputation.

\subsection{Variable selection procedures}

We use the same variable selection procedures as in the main manuscript: stability selection within bootstrap imputation (denoted by lasso + SS (LJ)) or bootstrap imputation with bolasso for variable selection (denoted by lasso + SS (BI-BL)), with final predictions made using logistic regression; and intrinsic selection designed to control the gFWER, PFP, and FDR, both with and without using Rubin's Rules via Lemma~\ref{lem:mi-normality} (denoted SPVIM + \{gFWER, PFP, FDR\}, respectively), with final predictions made using the Super Learner, with library described in Table~\ref{tab:sl-algs-data-analysis}. We based tuning parameter selection on a similar setting from the simulations: in this case, the sample size is 321 and there are 21 biomarkers, so we set $k = 5$, $q = 0.8$, the number of variables selected in each bootstrap run of stability selection equal to 9 (based on a target per-family error rate of $p (0.04)$ and threshold of 0.9).

\subsection{Assessing prediction performance}

Assessing prediction performance is complicated by both the imputation step and the initial variable selection step. To address this, we performed imputation within cross-fitting within Monte-Carlo sampling; this provides an unbiased assessment of the entire procedure, from imputation to variable selection to prediction. More specifically, for each of 100 replicates and each outcome, we performed the procedure outlined in Algorithm~\ref{alg:cf-impute-within-mc-pool}.

\vspace{.1in}
\begin{algorithm}
\caption{Imputation and pooled variable selection within cross-fitting and Monte-Carlo sampling}
\label{alg:cf-impute-within-mc-pool}
\begin{algorithmic}[1]
\vspace{.1in}
  \For{$b = 1, \ldots, 50$}
    \State generate a random vector $B_n \in \{1, \ldots, 5\}^n$ by sampling uniformly from $\{1, \ldots, 5\}$ with replacement, and for each $v \in \{1, \ldots, 5\}$, denote by $D_v$ the data with index in $\{i: B_{n,i} = v\}$;
    \For{$v = 1, \ldots, 5$}
      \State if using a bootstrap imputation-based procedure, create 100 bootstrap datasets based on the data in $\cup_{j \neq v}D_j$ and a single imputed dataset for each;
      \State create 10 imputed datasets $\{Z_{k,-v}\}_{k=1}^{10}$ based on the data in $\cup_{j \neq v}D_j$ using MICE;
      \State create 10 imputed datasets $\{Z_{k,v}\}_{k=1}^{10}$ based on the data in $D_v$ using MICE;
      \State apply the chosen variable selection procedure on the training data, resulting in a final set of selected variables $S_v$;
      \For{$k = 1, \ldots, 10$}
        \State train the chosen prediction algorithm on the training data $Z_{k,-v}$ using only variables in $S_v$;
        \State obtain $\text{AUC}_{k,v}$ and its associated variance $\text{var(AUC)}_{k,v}$ by predicting on the withheld test data $Z_{k,v}$ and measure prediction performance using AUC;
      \EndFor
      \State combine the AUCs and associated variance estimators into $\text{AUC}_v$ and $\text{var(AUC)}_v$ using Rubin's rules;
    \EndFor
    \State compute $\text{CV-AUC}_b = \frac{1}{5}\sum_{v=1}^v \text{AUC}_v$ and $\text{var(CV-AUC)}_b = \frac{1}{5}\sum_{v=1}^v \text{var(AUC)}_v$;
  \EndFor
  \State compute overall performance by averaging over the Monte-Carlo iterations.
\vspace{.1in}
\end{algorithmic}
\end{algorithm}

\subsection{Obtaining a final set of selected biomarkers}

We obtain a final set of selected biomarkers by applying the variable selection procedure to the full set of observations for each imputed dataset.

\subsection{Super Learner specification}

As in the simulations, we used a different specification for the internal Super Learner in the intrinsic selection procedure (max. depth 4 boosted trees (all tuning parameter values in Table~\ref{tab:sl-algs-data-analysis}) with pre-screening via univariate rank correlation with the outcome) and all other Super Learners (Table~\ref{tab:sl-algs-data-analysis}). In all cases, the final Super Learner fit for prediction performance of the selected set of variables used the candidate learners in Table~\ref{tab:sl-algs-data-analysis}.

\begin{table}
    \centering
    \begin{tabular}{c|ccc}
       Candidate Learner & R & Tuning Parameter & Tuning parameter  \\
       & Implementation & and possible values & description\\ \hline
        Random forests & \texttt{ranger} & \texttt{max.depth}  & Maximum tree depth \\
        &  & $\in\{1, 10, 20, 30, 100, \infty\}$ & \\ \hline
        Gradient boosted & \texttt{xgboost} & \texttt{max.depth} $= \{4\}$ &  Maximum tree depth\\
        trees & & \texttt{nrounds} $\in\{100, 500, 2000\}$ & Number of boosting \\
        & & & iterations \\ \hline
        Elastic net & \texttt{glmnet} & mixing parameter $\alpha$ & Trade-off between  \\
        &  & $\in \{0, \frac{1}{4}, \frac{1}{2}, \frac{3}{4}, 1\}$ & $\ell_1$ and $\ell_2$ regularization${}^{\ddagger}$ \\ \hline
    \end{tabular}
    \caption{Candidate learners in the Super Learner ensemble for the pancreatic cyst data analysis along with their R implementation, tuning parameter values, and description of the tuning parameters. All tuning parameters besides those listed here are set to their default values. In particular, the random forests are grown with \texttt{mtry} $=\sqrt{p}^{\dagger}$, a minimum node size of 5 for continuous outcomes and 1 for binary outcomes, and a subsampling fraction of 1; the boosted trees are grown with shrinkage rate of 0.1 and a minimum of 10 observations per node; and the $\ell_1$ tuning parameter for the elastic net is determined via 10-fold cross-validation. \\
    ${}^{\dagger}$: $p$ denotes the total number of predictors. }
    \label{tab:sl-algs-data-analysis}
\end{table}

\section{Additional results from the pancreatic cyst analysis}

In the main manuscript, we performed an analysis with goal of predicting whether a cyst was mucinous, using Algorithm~\ref{alg:cf-impute-within-mc-pool} to assess prediction performance. In Table~\ref{tab:mucinous}, we present the biomarkers selected using each procedure. Here, we show results using this same algorithm for the outcome of whether a cyst has high malignancy potential.

We present the results of our analysis in Figure~\ref{fig:high_malignancy} and Table~\ref{tab:high_malignancy}. In Figure~\ref{fig:high_malignancy}, we see that the PFP- and FDR-controlling intrinsic selection procedures again select no variables, on average, as we saw in the analysis of the mucinous outcome in the main manuscript. Prediction performance is also poor for the lasso-based estimators. Compared to the mucinous outcome, we observe reduced prediction performance for the gFWER-controlling intrinsic selection procedure, with an estimated cross-validated AUC of 0.803 (95\% confidence interval [0.67, 0.936]). In Table~\ref{tab:high_malignancy}, we display the final set of biomarkers selected by each procedure. Several biomarkers are selected across all two or more procedures that selected any variables on the full dataset. An antibody score was selected across all three procedures. Variables appearing in two or more procedures included an ACTB score, four neoplasia calls (binary variables), a glucose score, a combined amphiregulin- and glucose-based mucinous call, a fluorescence score and its associated mucinous call, and an antibody-based neoplasia call. Selection across the majority of procedures suggests that these variables may useful for predicting whether a cyst has high malignancy potential. 

\begin{figure}
    \centering
    \includegraphics[width=1\textwidth]{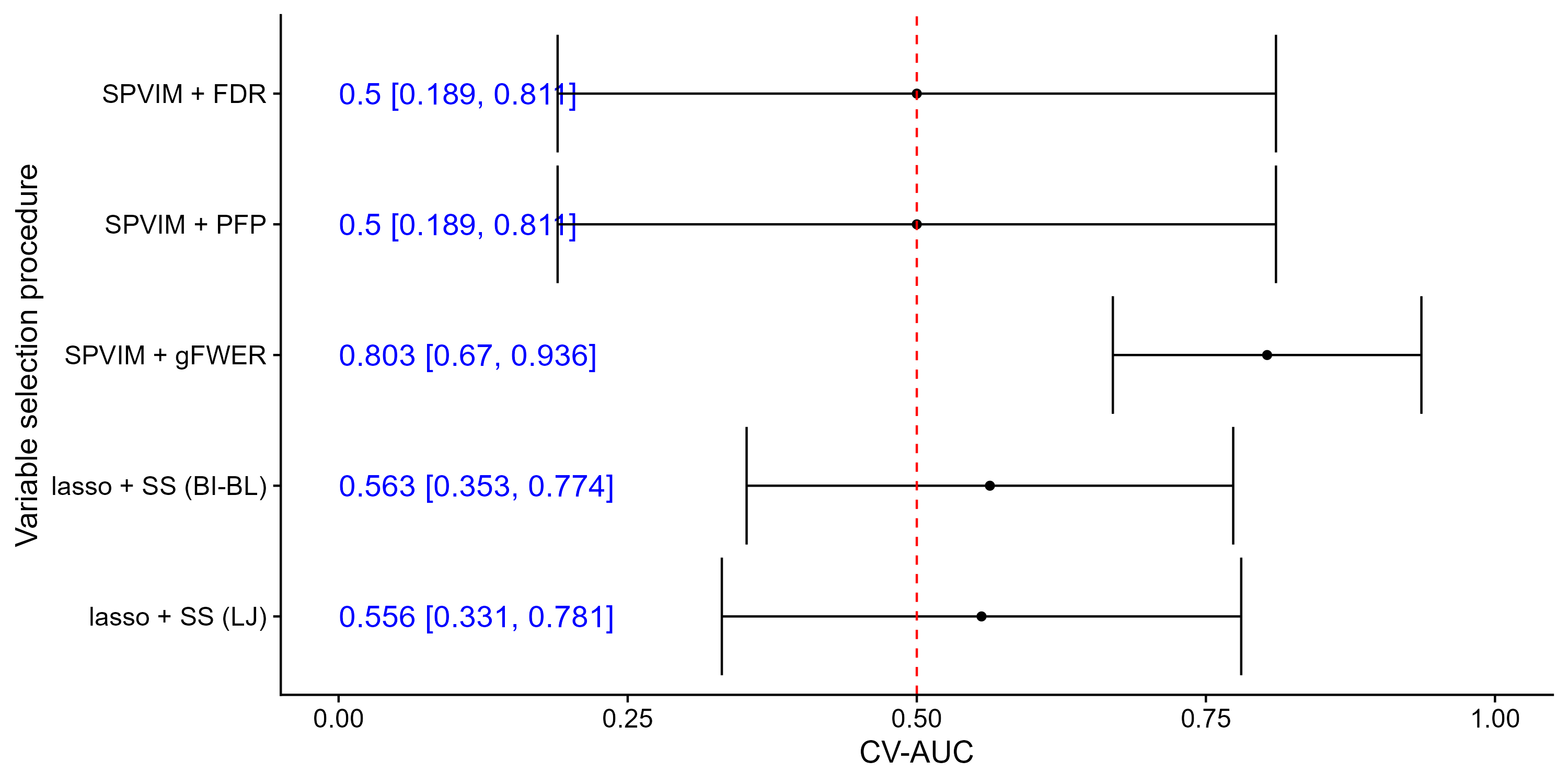}
    \caption{Cross-validated area under the receiver operating characteristic curve (CV-AUC) for predicting whether a cyst has high malignancy potential averaged over 100 replicates of the imputation-within-cross-validated procedure (Algorithm~\ref{alg:cf-impute-within-mc-pool}) for each variable selection algorithm. Prediction performance for lasso-based methods is based on logistic regression on the selected variables, while performance for Super Learner-based methods is based on a Super Learner. Error bars denote 95\% confidence intervals based on the average variance over the 100 replications.}
    \label{fig:high_malignancy}
\end{figure}

\input{plots/data-analysis_mucinous_selected}
\input{plots/data-analysis_high_malignancy_selected}

%% file: plots/data-analysis_mucinous_selected.tex
\begin{table}

\caption{\label{tab:mucinous}Biomarkers selected by each selection procedure for predicting whether a cyst is mucinous on the full imputed dataset. Full definitions of each variable are provided in the Supplementary Material.}
\centering
\begin{tabular}[t]{>{\raggedright\arraybackslash}p{10em}|l|>{\raggedright\arraybackslash}p{3em}|>{\raggedright\arraybackslash}p{3em}|r|>{}p{3em}|>{}p{3em}|>{}p{3em}|>{}p{5em}}
\hline
Biomarker & lasso + SS (LJ) & lasso + SS (BI-BL) & SPVIM + gFWER & Number of procedures\\
\hline
CEA & No & Yes & No & 1\\
\hline
CEA mucinous call & No & Yes & No & 1\\
\hline
ACTB & No & Yes & No & 1\\
\hline
Molecules (M) score & No & Yes & No & 1\\
\hline
M neoplasia call & No & Yes & Yes & 2\\
\hline
Telomerase (T) score & No & Yes & No & 1\\
\hline
T neoplasia call & No & Yes & No & 1\\
\hline
AREG (A) score & Yes & Yes & Yes & 3\\
\hline
A mucinous call & No & Yes & No & 1\\
\hline
Glucose (G) score & No & Yes & Yes & 2\\
\hline
G mucinous call & Yes & Yes & Yes & 3\\
\hline
A and G mucinous call & Yes & Yes & Yes & 3\\
\hline
Fluorescence (F) score & Yes & Yes & Yes & 3\\
\hline
F mucinous call & No & Yes & Yes & 2\\
\hline
DNA mucinous call & No & Yes & No & 1\\
\hline
DNA neoplasia call (v1) & No & Yes & No & 1\\
\hline
DNA neoplasia call (v2) & No & Yes & Yes & 2\\
\hline
MUC3AC score & Yes & Yes & Yes & 3\\
\hline
MUC5AC score & No & Yes & No & 1\\
\hline
Ab score & No & Yes & No & 1\\
\hline
Ab neoplasia call & No & Yes & Yes & 2\\
\hline
\end{tabular}
\end{table}

%% file: plots/data-analysis_high_malignancy_selected.tex
\begin{table}

\caption{\label{tab:high_malignancy}Biomarkers selected by each selection procedure for predicting whether a cyst has high malignancy potential on the full imputed dataset. Full definitions of each variable are provided in Table~\ref{tab:biomarker_definitions}.}
\centering
\begin{tabular}[t]{>{\raggedright\arraybackslash}p{10em}|l|>{\raggedright\arraybackslash}p{3em}|>{\raggedright\arraybackslash}p{3em}|r|>{}p{3em}|>{}p{3em}|>{}p{3em}|>{}p{5em}}
\hline
Biomarker & lasso + SS (LJ) & lasso + SS (BI-BL) & SPVIM + gFWER & Number of procedures\\
\hline
CEA & No & Yes & No & 1\\
\hline
CEA mucinous call & No & Yes & No & 1\\
\hline
ACTB & No & Yes & Yes & 2\\
\hline
Molecules (M) score & No & Yes & No & 1\\
\hline
M neoplasia call & No & Yes & Yes & 2\\
\hline
Telomerase (T) score & No & Yes & No & 1\\
\hline
T neoplasia call & Yes & Yes & No & 2\\
\hline
AREG (A) score & No & Yes & No & 1\\
\hline
A mucinous call & No & Yes & No & 1\\
\hline
Glucose (G) score & No & Yes & Yes & 2\\
\hline
G mucinous call & No & Yes & No & 1\\
\hline
A and G mucinous call & No & Yes & Yes & 2\\
\hline
Fluorescence (F) score & No & Yes & Yes & 2\\
\hline
F mucinous call & No & Yes & Yes & 2\\
\hline
DNA mucinous call & No & Yes & No & 1\\
\hline
DNA neoplasia call (v1) & No & Yes & Yes & 2\\
\hline
DNA neoplasia call (v2) & No & Yes & Yes & 2\\
\hline
MUC3AC score & No & Yes & No & 1\\
\hline
MUC5AC score & No & Yes & No & 1\\
\hline
Ab score & Yes & Yes & Yes & 3\\
\hline
Ab neoplasia call & No & Yes & Yes & 2\\
\hline
\end{tabular}
\end{table}